\documentclass[11pt]{article}

\usepackage{footmisc}
\usepackage{fullpage}
\usepackage{amssymb}
\usepackage{amsmath}
\usepackage{amsthm}
\usepackage{multirow}
\usepackage{enumerate}
\usepackage{graphicx}
\usepackage{mathrsfs}
\usepackage[utf8]{inputenc}
\usepackage{cite}
\usepackage{hyperref}
\usepackage{cleveref}
\usepackage{mathtools}
\usepackage{amsfonts}
\usepackage[T1]{fontenc}
\usepackage{mathpazo}
\usepackage{bm}
\usepackage{isomath}
\usepackage{changepage}
\usepackage[usenames,dvipsnames,svgnames,table]{xcolor}
\usepackage [english]{babel}
\usepackage{arydshln}

\newtheorem{remark}{Remark}

\definecolor{blueviolet}{RGB}{60,50,200}
\definecolor{oliveg}{RGB}{40,200,30}
\hypersetup{colorlinks=true,       % false: boxed links; true: colored links
    linkcolor=blueviolet,
    citecolor=oliveg,
}

\usepackage{tikz}
\usetikzlibrary{calc}
\usepackage{algorithm}
\usepackage{algcompatible}
\usepackage{color}
\usepackage{todonotes}

\theoremstyle{definition}
\newtheorem{definition}{Definition}[section]
\newtheorem{problem}{Problem}
\theoremstyle{plain}
\newtheorem{lemma}{Lemma}[section]
\newtheorem{theorem}{Theorem}
\newtheorem{fact}{Fact}

\newtheorem{corollary}{Corollary}[section]
\newtheorem{claim}{Claim}[section]

\newtheorem{observation}{Observation}[section]
\newtheorem{proposition}{Proposition}[section]

\newcommand{\der}[3]{\partial_{#1}^{#2}{~#3}}

\newcommand{\app}[4]{\left\langle \pi_{#1}(\der{#2}{#3}{#4}) \right\rangle}

\newcommand{\diag}{\textnormal{diag}}

\newcommand{\CAL}[1]{\mathcal{#1}}

\newcommand{\F}{\mathbb{F}}

\newcommand{\N}{\mathbb{N}}
\newcommand{\Q}{\mathbb{Q}}
\newcommand{\Z}{\mathbb{Z}}

\newcommand{\cL}{\mathcal{L}}
\newcommand{\cR}{\mathcal{R}}

\newcommand{\cD}{\mathcal{D}}
\newcommand{\cN}{\mathcal{N}}
\newcommand{\C}{\mathbb{C}}
\newcommand{\R}{\mathbb{R}}

\newcommand{\poly}{\textnormal{poly}}
\newcommand{\qpoly}{\textnormal{quasi-poly}}

\newcommand{\rank}{\textnormal{rank}}

\newcommand{\adj}{\textnormal{adj}}
\renewcommand{\char}{\textnormal{char}}

\newcommand{\VP}{\mathsf{VP}}
\newcommand{\VNP}{\mathsf{VNP}}

\newcommand{\APP}{\mathsf{APP}}
\newcommand{\GL}{\text{GL}}
\newcommand{\veca}{{\mathbf{a}}}

\newcommand{\vecx}{{\mathbf{x}}}
\newcommand{\vecy}{{\mathbf{y}}}
\newcommand{\vecz}{{\mathbf{z}}}
\newcommand{\vecu}{{\mathbf{u}}}
\newcommand{\vecv}{{\mathbf{v}}}
\newcommand{\vecw}{{\mathbf{w}}}

\mathchardef\mhyphen="2D

\newcommand{\ext}{\text{ext}}

\renewcommand\thefootnote{\textcolor{red}{\arabic{footnote}}}
\makeatletter
\newcommand\footnoteref[1]{\protected@xdef\@thefnmark{\ref{#1}}\@footnotemark}
\makeatother

\begin{document}
\title{Learning sums of powers of low-degree polynomials in the non-degenerate case}
\author{Ankit Garg\\
\normalsize{Microsoft Research India}\\
\normalsize{\tt{garga@microsoft.com}}
\and {Neeraj Kayal}\\
\normalsize{Microsoft Research India}\\
\normalsize{\tt{neeraka@microsoft.com}}
\and {Chandan Saha}\\
\normalsize{Indian Institute of Science}\\
\normalsize{{\tt chandan@iisc.ac.in}}}
\maketitle
\thispagestyle{empty}

\begin{abstract}
%A homogeneous $\Sigma \wedge \Sigma \Pi^{[t]}(s)$ formula for an $n$-variate degree-$d$ polynomial $f$ is an expression of the form
%$$f = c_1Q_1^{m_1} + \ldots + c_s Q_s^{m_s},$$
%where each $c_i\in \F^{\times}$, $Q_i$ is a homogeneous polynomial of degree $t_i \leq t$, and $t_im_i = d$. Learning  $\Sigma \wedge \Sigma \Pi^{[t]}$ circuits of size $\sigma$ is the problem of efficiently computing a small (preferably $\poly(\sigma)$ size) circuit for an input homogeneous $\Sigma \wedge \Sigma \Pi^{[t]}$ circuit of size $\sigma$ which is given as a black-box. It is an extremely challenging problem, primarily due to known size-efficient depth reductions from general circuits to $\Sigma \wedge \Sigma \Pi^{[t]}$ circuits. In this paper, we give a $\poly(\sigma, s^t)$-time proper learning algorithm for size-$\sigma$ homogeneous $\Sigma \wedge \Sigma \Pi^{[t]}(s)$ formulas that satisfy certain \emph{non-degeneracy conditions}. We show that a random $\Sigma \wedge \Sigma \Pi^{[t]}(s)$ formula is non-degenerate with high probability, if $t$ and $s$ are not too large. \\

We develop algorithms for writing a polynomial as sums of powers of low degree polynomials. Consider an $n$-variate degree-$d$ polynomial $f$ which can be written as
$$f = c_1Q_1^{m} + \ldots + c_s Q_s^{m},$$
where each $c_i\in \F^{\times}$, $Q_i$ is a homogeneous polynomial of degree $t$, and $t m = d$. In this paper, we give a $\poly((ns)^t)$-time learning algorithm for finding the $Q_i$'s given (black-box access to) $f$, if the $Q_i's$ satisfy certain \emph{non-degeneracy} conditions and $n$ is larger than $d^2$. The set of degenerate $Q_i$'s (i.e., inputs for which the algorithm does not work) form a non-trivial variety and hence if the $Q_i$'s are chosen according to any reasonable (full-dimensional) distribution, then they are non-degenerate with high probability (if $s$ is not too large). This problem generalizes symmetric tensor decomposition, which corresponds to the $t = 1$ case and is widely studied, having many applications in machine learning. Our algorithm (for $t=2$) allows us to solve the moment problem for mixtures of zero-mean Gaussians in the non-degenerate case.

Our algorithm is based on a scheme for obtaining a learning algorithm for an arithmetic circuit model from lower bound for the same model, provided certain non-degeneracy conditions hold. The scheme reduces the learning problem to the problem of decomposing two vector spaces under the action of a set of linear operators, where the spaces and the operators are derived from the input circuit and the complexity measure used in a typical lower bound proof. The non-degeneracy conditions are certain restrictions on how the spaces decompose. Such a scheme is present in a rudimentary form in an earlier work \cite{KayalS19}. Here, we make it more general and detailed, and potentially applicable to learning other circuit models.

An exponential lower bound for the representation above (also known as homogeneous $\Sigma \wedge \Sigma \Pi^{[t]}$ circuits) is known using the shifted partials measure. However, the number of linear operators in shifted partials is exponential and also the non-degeneracy condition emerging out of this measure is unlikely to be satisfied by a random $\Sigma \wedge \Sigma \Pi^{[t]}$ circuit when the number of variables is large with respect to the degree. We bypass this hurdle by proving a lower bound (which is nearly as strong as the previous bound) using a novel variant of the partial derivatives measure, namely \emph{affine projections of partials} ($\APP$). The non-degeneracy conditions appearing from this new measure are satisfied by a random $\Sigma \wedge \Sigma \Pi^{[t]}$ circuit. The $\APP$ measure could be of independent interest for proving other lower bounds.     
\end{abstract}

\clearpage

\newpage
\thispagestyle{empty}
\tableofcontents
\newpage

\pagenumbering{arabic}

\section{Introduction} \label{sec:intro}

Arithmetic circuits form a natural model for computing polynomials. They compute polynomials using basic arithmetic operations such as addition and multiplication.\footnote{One can also allow division, but it is a classical result that one can eliminate division operations from an arithmetic circuit without too much blow up in the circuit size \cite{strassen1973vermeidung}.} Formally, an arithmetic circuit is a directed acyclic graph such that the sources are labelled with variables or constants from the underlying field, the internal nodes (gates) are labelled with the arithmetic operations and the sink(s) outputs the polynomial(s) computed by the circuit.\footnote{Sometimes, one can let the edges going into addition gates to be labelled with field constants to allow scalar linear combinations. But, all these models can be interconverted to each other without too much blowup in the circuit size.} \emph{Size} of a circuit is the number of edges in the underlying graph, and \emph{depth} is the length of a longest path from a source to a sink node. The three main questions of interest regarding arithmetic circuits are the following:

\begin{itemize}
\item \textbf{Lower bounds.} Is there an "explicit" polynomial that requires super-polynomial sized arithmetic circuits to compute? This\footnote{One could also consider various other notions of explicitness while framing the lower bounds question.} is the famed $\VP$ vs $\VNP$ question (an arithmetic analogue of the P vs NP question\footnote{Or rather $\#$P vs NC.}). 
\item \textbf{Polynomial Identity Testing (PIT).} Here the question is, given\footnote{Either as a black-box or explicitly.} an arithmetic circuit, determine if its output is identically zero. There is an easy randomized algorithm for this problem (plug in random values and check if the output is zero). Finding a deterministic algorithm is a major open question in this field.
\item \textbf{Reconstruction.} Here the question is, given\footnote{Again, either as a black box or explicitly.} a polynomial, find the smallest (or approximately smallest) arithmetic circuit computing it.
\end{itemize}

For all the above questions, there is very little progress on them for general arithmetic circuits. So, a lot of effort has gone into studying them for restricted classes of arithmetic circuits (like constant depth, multilinear, set-multilinear, non-commutative circuits etc.). We refer the interested reader to the excellent surveys \cite{ShpilkaY10, ChenKW11, saptharishi2015survey} on this topic.
\\

A lot of interconnections are known between the three above-mentioned problems, some of which we touch upon below.
\begin{itemize}
\item \textbf{Lower bounds and PIT.} There are several connections that go both ways between lower bounds and PIT. It is known that lower bounds for general arithmetic circuits would imply PIT algorithms via the hardness vs randomness tradeoff (in the algebraic setting) \cite{kabanets2004derandomizing, dvir2010hardness}. Furthermore, non-trivial (deterministic) PIT algorithms also imply lower bounds \cite{heintz1980testing, kabanets2004derandomizing, agrawal2005proving}. While these concrete connections are not always present for restricted circuit models, several PIT algorithms have been inspired by corresponding lower bounds, e.g., \cite{raz2005deterministic, ForbesS13, oliveira2016subexponential, forbes2015deterministic}.
\item \textbf{Lower bounds and reconstruction.} It is known that worst case reconstruction of a circuit model implies lower bound for the same model \cite{FortnowK09, Volkovich16} (also see the discussion in Section \ref{sec:hardness of learning}). In the other direction, several reconstruction algorithms are inspired by lower bounds for the corresponding models \cite{KlivansS06, ForbesS13, GuptaKL11,GKQ14, KayalNST17, KayalNS19, KayalS19} (see also the discussion in Section \ref{sec:previous work}).
\item \textbf{PIT and reconstruction.} In one direction, a deterministic (worst case) reconstruction algorithm clearly implies a deterministic PIT algorithm (both in the black-box model) since the reconstruction algorithm would have to output an extremely small circuit when the circuit computes the zero polynomial. Randomized (or average-case) reconstruction algorithms may not have anything to do with deterministic PIT algorithms, of course. In the other direction, as discussed in \cite{ShpilkaY10}, black-box PIT algorithms seemingly can help in designing reconstruction algorithms. %(even when no randomized reconstruction algorithm is known). 
This is because a black-box PIT algorithm outputs a list of evaluation points such that any circuit from the class being considered evaluates to a non-zero value on at least one of points, and hence any two circuits in the class computing different polynomials evaluate to a different value on at least one of the points.\footnote{Assuming the class is closed under subtractions.} So, the list of evaluation points determines the circuit and now it remains to be seen if one can efficiently reconstruct the circuit from these evaluations.\footnote{Of course, a random set of points forms a hitting set and it seems hard to reconstruct the circuit given its evaluations on random points. However, the hitting sets constructed for deterministic PIT algorithms typically have a lot of special structures which could be exploited for reconstruction.} The reconstruction algorithms for sparse polynomials and constant top fan-in depth three circuits and read-once algebraic branching programs are some examples of reconstruction using PIT ideas \cite{KlivansS01, Shpilka09, KarninS09, ForbesS13}. Of course, deterministic PIT algorithms can also be sometimes used to get deterministic reconstruction algorithms when randomized ones are known \cite{KlivansS06, ForbesS13}.
\end{itemize}

To summarize, the three main problems in arithmetic complexity are richly interrelated and progress on one question spurs progress on the others. Hence, it is imperative to find more connections between these problems. This paper continues the line of work in \cite{KayalS19} on building a new connection between lower bounds and reconstruction. We build on the work of \cite{KayalS19} to further develop a meta framework\footnote{In \cite{KayalS19}, this framework is present in a rudimentary form.} that yields reconstruction algorithms in the \emph{non-degenerate} setting from lower bounds for the corresponding circuit models. In addition to developing this framework further, we implement this framework to learn sums of powers of low degree polynomials in the non-degenerate case (described in Section \ref{sec:model and results}). We remark that assuming some kind of non-degeneracy conditions might be essential for designing efficient learning algorithms; otherwise for most circuit models, one will have to assume constant top fan-in to get polynomial time algorithms. This is because of various hardness results about reconstruction in the worst case (see Section \ref{sec:hardness of learning}). The usefulness of assuming non-degeneracy conditions is best illustrated by the following example.
\\

Consider the model of homogeneous depth three powering circuits. This corresponds to the representation
$$
f(\vecx) = \sum_{i=1}^s \ell_i(\vecx)^d,
$$
where $\ell_i$'s are linear polynomials. Finding such a decomposition with the minimum possible $s$ is NP-hard even for degree $d = 3$ (this corresponds to symmetric tensor decomposition) \cite{Hastad90, Shitov16}. Regardless of the NP-hardness, one can design algorithms for this model under reasonable assumptions and such algorithms are widely used in machine learning (see Section \ref{sec:gaussians_intro}). One such algorithm, attributed to Jennrich \cite{harshman1970foundations, leurgans1993decomposition}, says that given $f(\vecx) = \sum_{i=1}^s \ell_i(\vecx)^3$ with $s \le n$ and $\ell_i$'s linearly independent, we can find the $\ell_i$'s in polynomial time.\footnote{There are natural extensions of this result for larger values of $d$ that can handle a larger number of components (roughly matching the best lower bounds we can prove for this model), e.g., see \cite{de2007fourth, AndersonBGRV14, BhaskaraCMV14, KayalS19}. Other algorithms include \cite{Kayal11, Garcia-MarcoKP18}.} A couple of things to notice about the assumptions are:
\begin{itemize}
\item $s \le n$: The number of summands that the algorithm can handle is (up to a small constant) the best known lower bound we can prove for this model (sums of cubes of linear forms or order-$3$ symmetric tensor decomposition).
\item The set of inputs for which the algorithm does not work, i.e., when $\ell_i$'s are linearly dependent, form a non-trivial variety (if $s \le n$). So, the algorithm would work for "random" $\ell_i$'s with high probability.
\end{itemize}

We hope to generalize the above kind of non-degenerate case learning algorithms to other circuit models. The circuit size which one might be able to handle if one implements the meta framework will depend on the lower bound one can prove for the circuit model. Since tensor decomposition algorithms (which corresponds to reconstruction for a very simple arithmetic circuit model) are widely used in machine learning, our meta framework raises the exciting possibility of importing techniques from arithmetic complexity to machine learning via reconstruction of various circuit models in the non-degenerate case. We mention one such possibility in Section \ref{sec:gaussians_intro}.
\\

%We want to remark that while assuming non-degeneracy conditions may seem like a disadvantage compared to learning algorithms that work in the worst case, it does allow us to design learning algorithms for a wider range of models (usually the number of summands $s$ is assumed to be constant in worst case reconstruction algorithms) which have a higher chance of leading to algorithms running well in practice.

Let us briefly describe the roadmap for the rest of this section now. In Section \ref{sec:model and results}, we describe our main results about learning sums of powers of low degree polynomials in the non-degenerate case. In Section \ref{sec:learning from lower bound}, we describe our techniques: the meta framework for turning lower bounds into reconstruction algorithms, the implementation for sums of powers of low degree polynomials and the non-degeneracy conditions needed for our algorithm to work. In Section \ref{sec:gaussians_intro}, we describe the connection to mixtures of Gaussians. Finally, in Sections \ref{sec:hardness of learning} and \ref{sec:previous work}, we review hardness results about reconstruction and some previous work.

\subsection{The model and our results} \label{sec:model and results}

We study the learning problem for an interesting subclass of depth four arithmetic circuits which is a generalization of depth three powering circuits or symmetric tensors. A circuit in this class, computing an $n$-variate degree-$d$ polynomial $f \in \F[\vecx]$, is an expression
\begin{equation} \label{eqn:sum of powers}
f = c_1Q_1^m + \ldots + c_s Q_s^m,
\end{equation}
where each $c_i\in \F^{\times}$, $Q_i$ is a homogeneous polynomial\footnote{The result in this paper holds even if $f = c_1Q_1^{m_1} + \ldots + c_s Q_s^{m_s}$, where each $Q_i$ (not necessarily homogeneous) has degree $t_i \leq t$ and $t_im_i = d$ for all $i\in[s]$. We present the analysis assuming homogeneity and uniform exponent $m$ for simplicity of exposition.} of degree $t$, and $tm = d$. Such a circuit is called a homogeneous $\Sigma \wedge \Sigma \Pi^{[t]}(s)$ circuit.\footnote{Technically, the expressions of this kind are known as $\Sigma \wedge \Sigma \Pi^{[t]}(s)$ \emph{formulas}. But, there is only a minor distinction between formulas and circuits in the constant depth case. Furthermore, in the random circuit setting, even this minor distinction is not there.} The parameter $t$ is typically much smaller than $d$. \\

We show that a homogeneous $\Sigma \wedge \Sigma \Pi^{[t]}$ circuit can be reconstructed efficiently if it satisfies certain \emph{non-degeneracy} conditions. We defer stating these conditions precisely to the end of this section, but it is worth mentioning that a \emph{random} $\Sigma \wedge \Sigma \Pi^{[t]}$ circuit is non-degenerate with high probability. In other words, if the coefficients of the monomials in $Q_1, \ldots, Q_s$, in Equation \eqref{eqn:sum of powers}, are chosen uniformly at random from a sufficiently large subset of $\F$ then the resulting circuit is non-degenerate with high probability. In this sense, \emph{almost all} homogeneous $\Sigma \wedge \Sigma \Pi^{[t]}$ circuits can be reconstructed efficiently. The following theorem is proved in Section \ref{sec:learning}. We will assume that factoring univariate polynomials over $\F$ can be done in randomized polynomial time\footnote{Univariate polynomials over finite fields can be factored in randomized polynomial time \cite{Berlekamp70}, and over $\Q$, they can be factored in deterministic polynomial time \cite{LenstraLL82}.}.

\begin{theorem}[Learning \emph{non-degenerate} sums of powers of low degree polynomials] \label{thm:learning sums of powers} 
Let $n,d,s,t \in \N$ such that $n \geq d^2$, $2 \leq t \leq \sqrt{\frac{\log d}{10 \cdot \log \log d}}$, $|\F| \geq (ns)^{150 \cdot t}$, $\char(\F) = 0$ or $> 2d$ and \linebreak $s \leq \min(n^{\frac{d}{1100 \cdot t^2}}, \exp(n^{\frac{1}{30 \cdot t^2}}))$. Then, there is a randomized algorithm which when given black-box access to an $n$-variate degree-$d$ non-degenerate polynomial $f = c_1Q_1^m + \ldots + c_s Q_s^m$, where each $c_i\in \F^{\times}$, $Q_i$ is a homogeneous polynomial of degree $t$, and $tm = d$ and the total number of monomials in $Q_i$'s is $\sigma$, outputs (with high probability) $Q_1',\ldots, Q_s'$ such that there exist a permutation $\pi: [s] \rightarrow [s]$ and non-zero constants $c_1',\ldots, c_s'$ so that $Q_i' = c_i' Q_{\pi(i)}$ for all $i \in [s]$. The running time of the algorithm is $\poly(n, \sigma, s^t)$.\footnote{Here, $\exp(x) = 2^x$. Once we know $Q'_1, \ldots, Q'_s$, we can determine the non-zero constants $d_1, \ldots, d_s$ such that $f = d_1Q_1'^m + \ldots + d_s Q_s'^m$ in randomized polynomial time.}
\end{theorem}

\textbf{Remarks.}
\begin{enumerate}
%\item \textit{Proper learning.}~ The learning is proper as the output is a homogeneous $\Sigma \wedge \Sigma \Pi^{[t]}$ circuit. 
%\item \textit{Minimum size formula.}~ It would follow from the definition of non-degeneracy (Definition \ref{defn: non-degeneracy}) that a non-degenerate homogeneous $\Sigma \wedge \Sigma \Pi^{[t]}$ circuit computing $f$ is the \emph{smallest} homogeneous $\Sigma \wedge \Sigma \Pi^{[t]}$ circuit computing $f$, and our algorithm outputs this smallest formula\footnote{i.e., if the formula $c_1Q_1^m + \ldots + c_s Q_s^m$ (in Equation \ref{eqn:sum of powers}) is non-degenerate then the output is $\tilde{c}_1 \tilde{Q}_1^m + \ldots + \tilde{c}_s \tilde{Q}_s^m$, where $\tilde{c}_i \in \F^{\times}$ and $\tilde{Q}_i = \hat{c}_i Q_i$ such that $\hat{c}_i \in \F^{\times}$ for all $i \in [s]$.}. As a consequence, we get that a random homogeneous $\Sigma \wedge \Sigma \Pi^{[t]}$ circuit computing $f$ is the \emph{smallest} homogeneous $\Sigma \wedge \Sigma \Pi^{[t]}$ circuit for $f$. %To our knowledge, such a result is not known for general circuits. In other words, it is not known if a ``random" circuit computing $f$ is the smallest circuit for $f$. 
\item \textit{Non-degeneracy.}~ The non-degeneracy conditions are explicitly mentioned in Section \ref{subsubsec:non-degeneracy}.
\item \textit{Bounds on $t$ and $s$.}~ The upper bounds on the parameters $t$ and $s$ in Theorem \ref{thm:learning sums of powers} originate from our analysis (especially, the part in Section \ref{sec: random formula} showing that a random $\Sigma \wedge \Sigma \Pi^{[t]}$ circuit is non-degenerate with high probability). We have not optimized this analysis in an attempt to keep it relatively simple. Given that the lower bounds (stated in Theorem \ref{thm:lower bound}) hold for a large range of $t$ and $s$, it may be possible to tighten our analysis significantly.
\item \textit{$t > 1$ substantially different from $t = 1$.} While we state the theorem for slightly super-constant values of $t$, one should think of $t$ being a constant as the main setting (in which case, the running time of our algorithm is polynomial). Even the $t = 2$ case, which was open before, is substantially different from the $t = 1$ case (as discussed in Section \ref{subsubsec:implementation}) and is relevant to the problem of mixtures of Gaussians (see Section \ref{sec:gaussians_intro}).
\item \textit{Uniqueness of $Q_i$'s.} A corollary of the analysis of our algorithm is that for a non-degenerate $f = c_1Q_1^m + \ldots + c_s Q_s^m$ (which holds if the $Q_i$'s are chosen randomly), this representation is of the smallest size and also unique. That is, if $f = \widetilde{c}_1 \widetilde{Q}_1^m + \ldots + \widetilde{c}_s \widetilde{Q}_{\widetilde{s}}^m$ with $\widetilde{s} \le s$, then $\widetilde{s} = s$ and there exist a permutation $\pi: [s] \rightarrow [s]$ and non-zero constants $d_1,\ldots, d_s$ such that $\widetilde{Q}_i = d_i Q_{\pi(i)}$ for all $i \in [s]$.
\end{enumerate}

Non-degeneracy conditions are satisfied if we choose the coefficients of the $Q_i$'s randomly. This gives us the following corollary.

\begin{corollary}[Learning \emph{random} sums of powers of low degree polynomials]
Let $n,d,s,t \in \F$ be as in Theorem \ref{thm:learning sums of powers}. There is a randomized algorithm which when given black-box access to an $n$-variate degree-$d$ polynomial $f = c_1Q_1^m + \ldots + c_s Q_s^m$, where each $c_i\in \F^{\times}$, $Q_i$ is a homogeneous polynomial of degree $t$, and $tm = d$ and the coefficients of $Q_i$'s are chosen uniformly and independently at random from a set $S \subset \F$ of size $|S| \ge (ns)^{150 \cdot t}$, outputs (with high probability) $Q_1',\ldots, Q_s'$ such that there exist a permutation $\pi: [s] \rightarrow [s]$ and non-zero constants $c_1',\ldots, c_s'$ so that $Q_i' = c_i' Q_{\pi(i)}$ for all $i \in [s]$. The running time of the algorithm is $\poly((ns)^t)$.
\end{corollary}

\subsection{Techniques: Learning from lower bounds} \label{sec:learning from lower bound}

The novelty of our approach lies in the use of lower bound techniques in the design of learning algorithms. Such connections are known for certain classes of Boolean circuits, in particular $\text{AC}^0$ and $\text{AC}^{0}[p]$ circuits \cite{LinialMN93, CarmosinoIKK16}. The influence of lower bound techniques on learning is also apparent in the case of ROABP reconstruction \cite{BeimelBBKV00,KlivansS06}. However, our approach differs substantially from these previous works and also uses lower bounds to design algorithms in the \emph{non-degenerate} case. At a high level, our technique can be summarized as a fancy {\em reduction to linear algebra}. In Section \ref{subsubsec:typical_lb}, we will see how lower bounds are typically proven in arithmetic complexity. Section \ref{subsubsec:learning_recipe} describes our meta framework of turning lower bounds into learning algorithms. In Section \ref{subsubsec:implementation}, we discuss how implement the framework for learning sums of powers of low degree polynomials. Finally in Section \ref{subsubsec:non-degeneracy}, we state the non-degeneracy conditions we require explicitly.

\subsubsection{A typical lower bound proof}\label{subsubsec:typical_lb}

Many of the circuit classes for which good lower bounds are known are of the form $T_1 + \ldots + T_s$, where each polynomial $T_i$ is ``simple'' in some sense\footnote{For example, in case of $\Sigma \wedge \Sigma \Pi^{[t]}(s)$ circuits, $T_i$ is a power of a degree-$t$ polynomial. One can also get such representations from general circuits by various depth reduction theorems \cite{AgrawalV08, Koiran12, Tavenas13}.}. The lower bound problem for such a class $\CAL{C}$ is to find an explicit polynomial $f$ such that any representation of the form $f = T_1 + T_2 + \ldots + T_s$,	where each $T_i$ is a simple polynomial, requires $s$ to be large. A typical lower bound strategy finds such an $f$ by constructing a set of linear maps $\mathcal{L}$ from the vector space of polynomials to some appropriate vector space such that the following properties hold: %($\cL \subseteq \{L: \F[\vecx] \rightarrow W\}$) 
	\begin{itemize}
		\item $\dim (\langle \mathcal{L} \circ T \rangle)$ is {\em small} (say $\le r$) for every simple polynomial $T$, \footnote{Here, $\langle S \rangle$ denotes the $\F$-linear span of a set of polynomials $S$.}
		\item $\dim (\langle \mathcal{L} \circ f \rangle)$ is {\em large} (say $\ge R$).
	\end{itemize}
Then, we have
\begin{align}
&\: f = T_1 + T_2 + \ldots + T_s \nonumber \\
&\Rightarrow \langle \mathcal{L} \circ f \rangle \subseteq \langle \mathcal{L} \circ T_1 \rangle + \cdots + \langle \mathcal{L} \circ T_s \rangle \label{eqn:lower_bound} \\
&\Rightarrow \dim (\langle \mathcal{L} \circ f \rangle) \le \dim (\langle \mathcal{L} \circ T_1 \rangle) + \cdots + \dim (\langle \mathcal{L} \circ T_s \rangle). \nonumber
\end{align}
This implies that $ s\geq R/r$. Now, let us see how the set of linear maps $\CAL{L}$ could potentially play a key role in learning class $\CAL{C}$. 

\subsubsection{Reduction to vector space decomposition - a recipe for learning}\label{subsubsec:learning_recipe}

The corresponding learning problem for $\CAL{C}$ is the following: Given a polynomial $f$ that can be expressed as $f = T_1 + \ldots + T_s$, where each $T_i$ is a simple polynomial, can we efficiently recover the $T_i$'s? It turns out that the set of linear maps $\CAL{L}$ (used to prove lower bounds) can now be used to devise an efficient learning algorithm via the following meta-algorithm, which works if the expression $T_1 + \ldots + T_s$ is ``non-degenerate''. Let us explain what we mean by non-degeneracy. One might expect that if the $T_i$'s are chosen randomly, then for some choice of linear maps $\cL$, the subspace condition in Equation (\ref{eqn:lower_bound}) becomes an equality and the sums become direct sums, i.e.,
\begin{align}
\langle \mathcal{L} \circ f \rangle = \langle \mathcal{L} \circ T_1 \rangle \oplus \cdots \oplus \langle \mathcal{L} \circ T_s \rangle. \label{eqn:direct_sum}
\end{align}
Existence of linear maps $\cL$ satisfying Equation (\ref{eqn:direct_sum}) for random $T_i$'s is the starting point for our learning framework. A couple of things are important to state here:
\begin{itemize}
\item If Equation (\ref{eqn:direct_sum}) is satisfied for even a single choice of $T_i$'s, then it is satisfied for random $T_i$'s because of the Schwarz-Zippel lemma \cite{Schwartz80, Zippel79}.
\item Equation (\ref{eqn:direct_sum}) implies a \emph{tight separation} within class $\CAL{C}$. So, a prerequisite for a lower bound method to be useful for our learning framework is that it should be able to prove a tight separation for that model. In fact, Equation (\ref{eqn:direct_sum}) is usually proven by exhibiting an explicit polynomial for which the linear maps in question yield a tight separation (see Lemma \ref{lem: tight lower bound}).
\end{itemize}

We will also need that $\cL = \cL_2\circ \cL_1$, i.e., $\cL$ is a combination of two sets of linear maps\footnote{For example, $k^{\text{th}}$ order partial derivatives are a composition of $(k-1)^{\text{th}}$ order partial derivatives and first order partial derivatives, i.e., $\partial_{\vecx}^{k} = \partial_{\vecx}^{k-1} \circ \partial_{\vecx}$.} and Equation (\ref{eqn:direct_sum}) holds for both $\cL_1$ and $\cL$. We say that the expression $T_1 + \cdots + T_s$ is non-degenerate if Equation (\ref{eqn:direct_sum}) holds for both $\cL_1$ and $\cL$. Now given\footnote{This framework should be applicable given only black-box access to $f$ using standard tricks (as we show for our problem) and we won't go into these details in this overview.} $f = T_1 + \cdots + T_s$, we have access to
\begin{eqnarray} \label{eqn:meta-algo nondegeneracy}
						U &:=& \langle \CAL{L}_1 \circ f \rangle = \langle \CAL{L}_1 \circ T_1 \rangle ~\oplus~ \ldots ~\oplus~ \langle \CAL{L}_1 \circ T_s \rangle  ~~~~~\text{and} \nonumber \\ 
						V &:=& \langle \CAL{L}_2 \circ \cL_1 \circ f \rangle = \langle \CAL{L}_2 \circ (\CAL{L}_1 \circ T_1) \rangle ~\oplus~ \ldots ~\oplus~ \langle \CAL{L}_2 \circ (\CAL{L}_1 \circ T_s) \rangle . \label{eqn:direct_sum_L}
				\end{eqnarray}
If we can recover the $\langle \cL_1 \circ T_i \rangle$ for all $i$, then usually one can recover the $T_i$'s\footnote{For example, one can recover a homogeneous polynomial if given all its degree-$k$ partial derivatives.}. Towards this, a crucial property of the linear maps $\cL_2$ (from $U$ to $V$) is that $\cL_2$ maps each component space $\langle \cL_1 \circ T_i \rangle$ to the corresponding component space of $V$. This motivates the following problem.

\begin{problem}[Vector space decomposition] Given two vector spaces $U$ and $V$ and a set of linear maps $\cL$ from $U$ to $V$, find a decomposition
\begin{eqnarray} \label{eqn:meta-algo nondegeneracy}
						&U  = U_1 ~\oplus~ \ldots ~\oplus~ U_s   \nonumber \\ 
						&V  = V_1 ~\oplus~ \ldots ~\oplus~ V_s , \nonumber
				\end{eqnarray}
such that $\langle \cL \circ U_i \rangle \subseteq V_i$ for all $i \in [s]$ (if such a decomposition exists). Moreover, we can ask that each of the pairs $(U_i, V_i)$ be further indecomposable with respect to $\cL$.
\end{problem}

Amazingly, polynomial time algorithms are known for a symmetric version of this problem, where $U = V$ and we require $U_i = V_i$.\footnote{The symmetric version is known as \emph{module decomposition} in the literature.} The algorithm was discovered in \cite{ChistovIK97} based on the algorithms developed for decomposition of algebras (e.g., see \cite{FriedlR85,Ronyai90,Eberly91}). The algorithm works over finite fields, $\C$ and $\R$ (if the input is over $\Q$, then the algorithm outputs a decomposition over an extension field). We give a simple reduction in Section \ref{sec: space to module decomposition} that reduces the vector space decomposition problem to the symmetric version. However, since we are in a specialized setting, we can design a simpler algorithm using the ideas in \cite{ChistovIK97} that also works over $\Q$ (this is important for some potential applications like mixtures of Gaussians). \\

Thus, we are capable of doing a decomposition like the one in Equation (\ref{eqn:direct_sum_L}). But, why should we end up with the same decomposition? Certainly there are cases where the decompositions are not unique. For example, if $U = V$ and $\cL$ just consists of the identity map, then any decomposition into one-dimensional spaces is a valid one. However, there is a characterization of all decompositions in the symmetric setting (Krull-Schmidt theorem, Theorem \ref{thm:Krull-Schmidt} in Section \ref{sec:appendix adjoint}) and it extends to vector space decomposition via our reduction (Corollary \ref{corrLuniqueness_vector_space_decomp}). In many settings (including the one in this paper), this characterization helps in proving the uniqueness of decomposition. \\

Finally, the meta algorithm is stated in Algorithm \ref{metaalg:learning from lower bounds}, which works under the assumptions:
\begin{enumerate}
\item\label{item:meta_framwork_1} The following direct sum structure holds,
\begin{eqnarray}
						U &:=& \langle \CAL{L}_1 \circ f \rangle = \langle \CAL{L}_1 \circ T_1 \rangle ~\oplus~ \ldots ~\oplus~ \langle \CAL{L}_1 \circ T_s \rangle  ~~~~~\text{and} \nonumber \\ 
						V &:=& \langle \CAL{L}_2 \circ \cL_1 \circ f \rangle = \langle \CAL{L}_2 \circ (\CAL{L}_1 \circ T_1) \rangle ~\oplus~ \ldots ~\oplus~ \langle \CAL{L}_2 \circ (\CAL{L}_1 \circ T_s) \rangle . \label{eqn:direct_sum_L_2}
\end{eqnarray}
\item\label{item:meta_framwork_2} (\ref{eqn:direct_sum_L_2}) is the \emph{unique} indecomposable decomposition for the vector spaces $U$ and $V$ w.r.t. $\cL_2$.
\item\label{item:meta_framwork_3} One can recover $T_i$ from $\langle \cL_1 \circ T_i \rangle$ efficiently.
\end{enumerate}

Next, we will discuss how we prove these assumptions for our setting, namely sums of powers of low degree polynomials.

\begin{algorithm}
	\caption{Meta algorithm: Learning from lower bounds} \label{metaalg:learning from lower bounds}
	\begin{algorithmic}
		\STATE \textbf{Input}: $f = T_1 + \cdots + T_s$.
		\STATE \textbf{Output}: $T_1',\ldots, T_s'$ such that there exists a permutation $\sigma: [s] \rightarrow [s]$ so that $T_{i}' = T_{\sigma(i)}$.
	\end{algorithmic}
	\begin{algorithmic}[1]
		\STATEx
		\STATE Take an appropriate set of linear maps $\cL = \cL_2 \circ \cL_1$. Compute $U := \langle \CAL{L}_1 \circ f \rangle$ and $V := \langle \CAL{L}_2 \circ \cL_1 \circ f \rangle$.
		\STATE Obtain a (further indecomposable) vector space decomposition of $U$ and $V$ with respect to $\cL_2$, namely $U = U_1 \oplus \cdots \oplus U_s$ and $V = V_1 \oplus \cdots \oplus V_s$.
		\STATE Compute $T_i'$ from $U_i$ (assuming $U_i = \langle \CAL{L}_1 \circ T_i' \rangle$). 				
	\end{algorithmic}
\end{algorithm}

\subsubsection{Implementation for sums of powers of low degree polynomials}\label{subsubsec:implementation}

In this section, we discuss how we implement the meta algorithm described above for sums of powers of low degree polynomials. As discussed, the main ingredient in the learning algorithm is the lower bound. So, at first we need to understand how lower bounds are proven for this model \cite{Kayal12eccc, GKKS14, KayalSS14}. Let us first consider the setting of sums of powers of linear forms\footnote{Also known in the literature as symmetric tensor decomposition or the Waring rank problem.}, i.e., $t = 1$ \cite{Nisan91, NisanW97}. Our goal is to find a degree-$d$ homogeneous polynomial $f$ such that any expression of the form:
$$
f = \ell_1^d + \cdots + \ell_s^d,
$$
with $\ell_i$'s linear, requires a large value of $s$. The set of linear maps here will be $\cL = \partial_\vecx^{\lfloor d/2 \rfloor}$, i.e., all partial derivatives of order $\lfloor d/2 \rfloor$. Then, it is easy to see that $\dim(\langle \cL \circ \ell^d\rangle) \le 1$ for all linear polynomials $\ell$. Thus, any $f$ has a lower bound of $s \ge \dim(\langle \cL \circ f\rangle)$. One can easily design polynomials with large dimension for the partial derivatives, e.g., the elementary symmetric polynomial of degree $d$ in $n$ variables satisfies $\dim(\langle \cL \circ f\rangle) = {n \choose \lfloor d/2 \rfloor}$.\footnote{If we use this lower bound with the recipe in Section \ref{subsubsec:learning_recipe}, then one gets a close variant of Jennrich's algorithm for symmetric tensor decomposition.} \\

For a long time, it was not known how to generalize these super-polynomial lower bounds even to the $t = 2$ case. What goes wrong is that $\dim(\langle \cL \circ Q^m\rangle)$ is no longer small, for $\cL = \partial_\vecx^{k}$, when $Q$ is a degree-$t$ homogeneous polynomial with $t \ge 2$. For example, $\dim(\langle \cL \circ (x_1^2 + \cdots + x_n^2)^m\rangle) \ge {n \choose k}$ for $k \le m \le n$. However, one can still say something about the partial derivatives. If we take any $\alpha \in \Z_{\ge}^n$ with $|\alpha| := \sum_{i=1}^n \alpha_i = k$, then $Q^{m - k}$ divides $\partial_\vecx^\alpha Q^m$ for $k \le m$. Hence, any $\partial_\vecx^\alpha Q^m$ is of the form $Q^{m-k} R$, where $R$ is homogeneous of degree $k(t-1)$. Now, the main observation in \cite{Kayal12eccc} was that we can make use of this special property of powers of low degree polynomials by using the \emph{shifted} partial derivatives measure which is defined as follows,
$$\mathsf{SP}_{k, \ell}(f) := \dim \left \langle \vecx^{\ell} \cdot \partial^k_{\vecx}~ f\right \rangle.$$ 
That is, we take all $k^{\text{th}}$ order partial derivatives of $f$ and then multiply by all degree-$\ell$ monomials and then take the dimension of their span. Now, for any $\alpha, \beta$ such that $|\alpha| = k$ and $\beta = \ell$, $\vecx^{\beta} \partial_\vecx^{\alpha} Q^m$ is of the form $Q^{m-k} R$, where $R$ is a homogeneous polynomial of degree $\ell + k(t-1)$. Hence
$$
\mathsf{SP}_{k, \ell}(Q^m) \le {n+ \ell + k(t-1) -1 \choose \ell + k(t-1)},
$$
where if $k,\ell$ are not too large, then one can expect $\mathsf{SP}_{k, \ell}(f)$, for an appropriately chosen $f$, to be close to the number of operators ${n+ \ell - 1 \choose \ell} {n + k -1 \choose k}$. Indeed, with appropriate choices of $k$ and $\ell$, this can be used to prove exponential lower bounds for the model sums of powers of low degree polynomials \cite{Kayal12eccc} and other more general models as well \cite{GKKS14, KayalSS14, KumarS14, FournierLMS15, KayalLSS17, KumarS17}. In all of these lower bounds, the value of $\ell$ is chosen to be comparable or larger than $n$. This makes the number of linear maps exponential and hence not suitable for designing an efficient algorithm. In fact, even ignoring the large number of linear maps, with such a large value of $\ell$, the shifted partials measure is unlikely to satisfy the direct sum property in Equation (\ref{eqn:direct_sum_L_2}) (see Section \ref{sec: limitation of SP}) when the number of variables is large with respect to the degree. A natural way to decrease the number of linear maps is to project to a smaller number of variables. To our surprise, if we project down to a smaller number of variables, one does not need shifts at all to prove the lower bound! We call this new measure \emph{affine projections} of partials, which we define next.\footnote{The word affine is added to avoid confusion with another kind of projection (namely multilinear projection) which is usually done in the literature to prove lower bounds for depth four arithmetic circuits.} \\

\textbf{Affine projections of partials -- a novel adaptation of the partial derivatives measure.}~ Let $f$ be a polynomial in variables $\vecx = (x_1, \ldots, x_n)$ and $L = (\ell_1(\vecz), \ldots, \ell_n(\vecz))$ a tuple of linear forms\footnote{More generally, $\ell_1(\vecz), \ldots, \ell_n(\vecz)$ are affine forms, but for this work it suffices to take them as linear forms.} in variables $\vecz = (z_1, \ldots, z_{n_0})$. The parameter $n_0$ would be much smaller than $n$ in this paper. Let $\pi_L$ be the following affine projection map from $\F[\vecx]$ to $\F[\vecz]$,
$$\pi_L(f) := f(\ell_1(\vecz), \ldots, \ell_n(\vecz)).$$
For a set $S \subseteq \F[\vecx]$, the projection $\pi_L$ is naturally defined as,
$$\pi_L(S) := \{\pi_L(f): f \in S\}.$$
Recall that $\der{\vecx}{k}{f}$ is the set of all $k$-th order partial derivatives of $f$ and $\langle S \rangle$ is the $\F$-linear span of a set of polynomials $S$. The \emph{affine projections of partials} ($\APP$) measure is defined as,
\begin{flalign} \label{eqn:the measure}
\text{(The measure)} \quad \quad\quad\quad \quad \quad \quad \APP_{k,n_0}(f) := \max_L~ \dim \app{L}{\vecx}{k}{f}, &&
\end{flalign}
where the maximum is taken over all $n$-tuple $L = (\ell_1(\vecz), \ldots, \ell_n(\vecz))$ of linear forms in $\F[\vecz]$. It is easy to verify that for any $f,g \in \F[\vecx]$ the following linearity property is satisfied,
\begin{flalign} \label{eqn:linearity of the measure}
\text{(Linearity of the measure)} \quad \quad \APP_{k,n_0}(f+g) \leq \APP_{k,n_0}(f) + \APP_{k,n_0}(g). &&
\end{flalign}

The $\APP$ measure can be alternatively defined using a random affine projection $\pi_L$. The observation below is an easy consequence of the Schwartz-Zippel lemma \cite{Schwartz80,Zippel79}.
\begin{observation} \label{obs:random L}
If $|\F| \geq 10 \cdot (d-k)\cdot {n+k-1 \choose k}$ and every coefficient of the linear forms in $L = (\ell_1(\vecz), \ldots, \ell_n(\vecz))$ is chosen from a set $S \subseteq \F$ of size $10 \cdot (d-k)\cdot {n+k-1 \choose k}$ then with probability at least $0.9$, 
$$\APP_{k,n_0}(f) = \dim \app{L}{\vecx}{k}{f}.$$
\end{observation}

\begin{remark}
\emph{Similarity with the skewed partials measure.} The $\APP$ measure is akin to the skewed partials ($\mathsf{SkP}$) measure introduced in \cite{KayalNS16} -- SkP is a special case of $\APP$. We show that the lower bound proof works with the $\mathsf{SkP}$ measure, but to ensure that the hard polynomial $f_{n,d}$ is multilinear we need affine projections (particularly, p-projections). However, the main reason for us to work with general/random affine projections (instead of $\mathsf{SkP}$ or p-projections) is to make the learning algorithm require as weak a non-degeneracy condition as possible. 
\end{remark}

Let us give some intuition as to why the $\APP$ measure can yield lower bounds for sums of powers of low degree polynomials. As we say above, any $\partial_\vecx^\alpha Q^m$, with $|\alpha| = k$, is of the form $Q^{m-k} R$, where $R$ is homogeneous of degree $k(t-1)$ (recall $Q$ is homogeneous of degree $t$). This implies that $\APP(Q^m) \le {n_0 + k(t-1) - 1 \choose k(t-1)}$. However, for an appropriately chosen homogeneous degree-$d$ polynomial $f$, one can expect that $\APP(f)$ could be as large as $\min\{{n_0 + d-k-1 \choose d-k}, {n+k-1 \choose k}\}$. The first upper bound is from the fact that after derivatives and projection, we are in the space of degree-$(d-k)$ polynomials in $n_0$ variables and the second is from the fact that we have ${n+k-1 \choose k}$ linear maps in $\APP$. With appropriate choices of $n_0, k$ and the polynomial $f$, we can prove the following lower bound result which comes close to the best known lower bounds via shifted partials.

\vspace{2mm}
\begin{theorem} [Lower bound for homogeneous $\Sigma \Pi \Sigma \Pi^{[t]}$ circuits using $\APP$] \label{thm:lower bound}
The $\APP$ measure, defined above, can be used to prove the following lower bound for homogeneous $\Sigma \Pi \Sigma \Pi^{[t]}$ circuits.
\begin{itemize}
\item \underline{High $t$ case}: Let $n,d,t \in \N$ such that $n \geq d^2$ and $\ln \frac{n}{d} \leq t \leq \frac{d}{4\cdot e^{10} \cdot \ln d}$. There is a family of $n$-variate degree-$d$ multilinear polynomials $\{f_{n,d}\}$ in $\VNP$ such that any homogeneous $\Sigma \Pi \Sigma \Pi^{[t]}(s)$ circuit computing $f_{n,d}$ must have $$s = \left(\frac{n}{d}\right)^{\Omega\left(\frac{d}{t \ln t}\right)}.$$ 
\item \underline{Low $t$ case}: Let $n,d,t \in \N$ such that $n \geq d^{20}$ and $1 \leq t \leq \min \left\{ \frac{\ln n}{6e \cdot \ln d},~ d \right\}$. There is a family of $n$-variate degree-$d$ multilinear polynomials $\{f_{n,d}\}$ in $\VNP$ such that any homogeneous $\Sigma \Pi \Sigma \Pi^{[t]}(s)$ circuit computing $f_{n,d}$ must have $$s = n^{\Omega\left(\frac{d}{t}\right)}.$$ 
\end{itemize}
\end{theorem}

\begin{remark}
 \emph{More general circuit model.} The above lower bound (proved in Section \ref{sec:lowerbound}) is for the class of homogeneous $\Sigma \Pi \Sigma \Pi^{[t]}$ circuits, which contains the class of homogeneous $\Sigma \wedge \Sigma \Pi^{[t]}$ circuits. In fact, the $\APP$ measure can be used to give a super-polynomial lower bound for general homogeneous depth four circuits -- we skip the proof of this fact here.
\end{remark}

Of course, an $n^{\Omega\left(\frac{d}{t}\right)}$ lower bound is already known for homogeneous $\Sigma \Pi \Sigma \Pi^{[t]}$ circuit using the shifted partials measure \cite{KayalSS14, FournierLMS15}. We get nearly the same lower bound by replacing ``shifts'' by an affine projection. As discussed above, this change is essential to satisfy the direct sum property in Equation (\ref{eqn:direct_sum_L_2}). In terms of lower bounds, this means that we show an explicit polynomial which can be computed by a homogeneous $\Sigma \wedge \Sigma \Pi^{[t]}(s)$ circuit but not by \emph{any} homogeneous $\Sigma \wedge \Sigma \Pi^{[t]}(s-1)$ circuit.

\begin{lemma} [Tight separation for homogeneous $\Sigma \wedge \Sigma \Pi^{[t]}$ circuits] \label{lem: tight lower bound}
Suppose $n, d, s, t, \F$ satisfy the conditions in Theorem \ref{thm:learning sums of powers}. Then, there is a family of explicit\footnote{That is, in time $\poly(n,s)$, we can output the $\Sigma \wedge \Sigma \Pi^{[t]}(s)$ circuit computing the polynomial.} $n$-variate degree-$d$ polynomials computable by homogeneous $\Sigma \wedge \Sigma \Pi^{[t]}(s)$ circuits but not by \emph{any} homogeneous $\Sigma \wedge \Sigma \Pi^{[t]}(s-1)$ circuit.
\end{lemma}

The above lemma (whose proof is implicit in Section \ref{sec: random formula}) allows us to argue that a random $\Sigma \wedge \Sigma \Pi^{[t]}$ circuit is non-degenerate with high probability (Item \ref{item:meta_framwork_1} in the assumptions for Algorithm \ref{metaalg:learning from lower bounds}). Let us now briefly explain how we prove the uniqueness of decomposition and describe our algorithm for recovering a term from the corresponding vector space of polynomials (Items \ref{item:meta_framwork_2} and \ref{item:meta_framwork_3}).
\\

\textbf{Uniqueness of vector space decomposition.} The \emph{adjoint algebra} helps us prove uniqueness of decomposition in our setting. Let us discuss what that is. Recall that we have vector spaces $U, V$ and a set of linear maps $\cL_2$ from $U$ to $V$. The adjoint algebra of $\cL_2$ is defined as follows:
\begin{equation}
\adj(\CAL{L}_2) := \left\{ (D,E) ~~:~~D: U \rightarrow U,~ E: V \rightarrow V \text{ and } K \circ D=E\ \circ K ~\text{ for all }~ K \in \CAL{L}_2\right\}.
\end{equation}
Suppose $U = U_1 \oplus \cdots \oplus U_s$, $V = V_1 \oplus \cdots \oplus V_s$ is an indecomposable decomposition with respect to $\cL_2$. If it so happens that for every $(D,E) \in \adj(\cL_2)$, there exist constants $\alpha_1,\ldots, \alpha_s$ such that $Du_i = \alpha_i u_i$ for all $u_i \in U_i$, then the decomposition is unique (this follows from Corollary \ref{cor:diagonalizable adjoint} and \ref{corrLuniqueness_vector_space_decomp}). We show this is the case in our setting. We deviate from the framework in Section \ref{subsubsec:learning_recipe} a little bit to simplify the analysis. Due to this, the vector spaces in our case turn out to be $V = \left\langle \pi_L(Q_1)^{m-k}, \ldots, \pi_L(Q_s)^{m-k}\right\rangle$ and $W = W_1 \oplus \ldots \oplus W_s$, $W_i := \app{P}{\vecz}{k}{\pi_L(Q_i)^{m-k}}$, and $\cL_2 = \pi_{P}(\partial_{\vecz}^k ~)$ are linear maps from $V$ to $W$ (changing notation to match with Section \ref{sec:learning}). Here, $L$ is a random projection onto $n_0$ variables $\vecz$, and $P$ is a random projection onto $m_0$ variables $\vecw$. \\

\textbf{Recovery of the terms from the corresponding vector spaces.} Due to the above-mentioned deviation from the framework, the final problem we have to solve is the following: Given access to the random projections $\pi_L(Q_1),\ldots, \pi_L(Q_s)$, recover the $Q_i$'s. First of all, if one is given multiple projections $\pi_L(Q)$, then it is not hard to recover $Q$. However, we have multiple polynomials and for each random projection, we could be given the polynomials in an arbitrary order. This makes the recovery slightly non-trivial. For details of how we solve this, see the analysis of Steps \ref{mainalgo:step7}-\ref{mainalgo:step8} of Algorithm \ref{alg:learning sums of powers} in Section \ref{sec:analysis}.
\\

Next, we explicitly state the non-degeneracy conditions we require. The details of our algorithm and analysis can be found in Section \ref{sec:learning}. As mentioned, we deviate from the general recipe to simplify the analysis, but the recipe provides the intuition and forms the backbone of the algorithm.

\subsubsection{Non-degeneracy conditions}\label{subsubsec:non-degeneracy}
In this section, we state the non-degeneracy conditions that our algorithm requires. Let $C$ be a homogeneous $\Sigma \wedge \Sigma \Pi^{[t]}(s)$ circuit computing an $n$-variate degree-$d$ polynomial
\begin{equation} \label{eqn:non-degenerate formula}
f = c_1Q_1^m + \ldots + c_sQ_s^m.
\end{equation}
\textbf{Notations.}~ Let $\vecz = (z_1, \ldots, z_{n_0})$ be a set of $n_0 = \lfloor n^{\frac{1}{3 \cdot t}}\rfloor$ variables and $\vecw = (w_1, \ldots, w_{m_0})$ a set of $m_0 = \lfloor {n}^{\frac{1}{15 \cdot t^2}}\rfloor$ variables. Let $L = (\ell_1(\vecz), \ldots, \ell_{n}(\vecz))$ be a tuple of $n$ linear forms in $\F[\vecz]$ and $P = (p_1(\vecw), \ldots, p_{n_0}(\vecw))$ a tuple of $n_0$ linear forms in $\F[\vecw]$. For every such tuples of linear forms $L$ and $P$, we can define the following spaces and polynomials by setting $k = \left\lceil \frac{130 \cdot t \cdot \log s}{\log n} \right\rceil$:   
\begin{itemize}
\item $U := \app{L}{\vecx}{k}{f}$ and $U_i := \app{L}{\vecx}{k}{Q_i^m}$,
\item $G_i := \pi_L(Q_i)$ and $g_0(\vecz) := G_1^{e} + \ldots + G_s^{e}$, where $e = m-k$,
\item $\widetilde{U_i} := \left \langle \vecz^{2k(t-1)} \cdot G_i^{e} \right \rangle$, where $\vecz^{2k(t-1)}$ is the set of all $\vecz$-monomials of degree $2k(t-1)$,
\item $W := \app{P}{\vecz}{k}{g_0}$ and $W_i := \app{P}{\vecz}{k}{G_i^{e}}$.
\end{itemize}
Observe that
$$U \subseteq U_1 + \ldots + U_s, ~~~\text{and}~~~  U_i \subseteq \left \langle \vecz^{k(t-1)} \cdot \pi_L(Q_i)^{m-k} \right \rangle ~~~\text{implying}~~~ \dim{U_i} \leq {n_0 + k(t-1) -1 \choose k(t-1)},$$
$$W \subseteq W_1 + \ldots + W_s, ~~~\text{and}~~~  W_i \subseteq \left \langle \vecw^{k(t-1)} \cdot \pi_P(G_i)^{e-k} \right \rangle ~~~\text{implying}~~~ \dim{W_i} \leq {m_0 + k(t-1) -1 \choose k(t-1)}.$$

\begin{definition} [Non-degeneracy] \label{defn: non-degeneracy}
The circuit $C$, given by Equation \eqref{eqn:non-degenerate formula}, is \emph{non-degenerate} if there exist $L$ and $P$ (as above) such that the following conditions are satisfied:
\begin{enumerate}
\item \label{cond1} $U = U_1 \oplus \ldots \oplus U_s$ ~and~ $\dim U_i = {n_0 + k(t-1) -1 \choose k(t-1)}$ for all $i \in [s]$,
\item \label{cond2} $W = W_1 \oplus \ldots \oplus W_s$ ~and~ $\dim W_i = {m_0 + k(t-1) -1 \choose k(t-1)}$ for all $i \in [s]$,
\item \label{cond3} $\widetilde{U_1} + \ldots + \widetilde{U_s} = \widetilde{U_1} \oplus \ldots \oplus \widetilde{U_s}$ for all $i \in [s]$,
\item \label{cond4} $[G_1^e]_{z_1=0}, \ldots, [G_s^e]_{z_1=0}$ are $\F$-linearly independent.
\end{enumerate}
\end{definition}
\vspace{2mm}
Condition \ref{cond1} and \ref{cond2} constitute the main part of non-degeneracy. Condition \ref{cond3} and \ref{cond4} have been added to aid our analysis and keep it relatively simple; it may be possible to dispense with these conditions completely perhaps by altering Conditions \ref{cond1} and \ref{cond2} slightly.  \\

We prove the following lemma in Section \ref{sec: random formula}, which says that if we choose the $Q_i$'s randomly in Equation (\ref{eqn:non-degenerate formula}), then the circuit is non-degenerate with high probability.

\begin{lemma} [Random $\Sigma\wedge\Sigma\Pi^{[t]}$ circuits are non-degenerate] \label{lem:non-degenerate_random}
Suppose the coefficients of $Q_i$'s are chosen uniformly and independently at random from a set $S \subset \F$ of size $|S| \ge (n s)^{150 \cdot t}$. Then, with probability $1 - o(1)$, the non-degeneracy conditions in Definition \ref{defn: non-degeneracy} are satisfied.
\end{lemma}

\subsection{Moment problem for mixtures of Gaussians}\label{sec:gaussians_intro}

In this section, we discuss an application to learning parameters of mixtures of Gaussians from the moments.
%\Ankit{Expand on below.}   
Symmetric tensor decomposition (and also general tensor decomposition) has a lot of applications in both supervised and unsupervised machine learning, e.g., in independent component analysis, learning latent variable models, hidden Markov models, topic models and mixture of Gaussians \cite{MosselR05, comon2010handbook, hsu2012spectral, AnandkumarHK12, HsuK13, anandkumar2014tensor, AndersonBGRV14, BhaskaraCMV14}. In our language, (symmetric) tensor decomposition corresponds to the $t=1$ case of "sums of powers of degree-$t$ polynomials". It is conceivable that learning algorithms which handle more general circuit classes as compared to tensor decomposition will handle a richer class of learning models than mentioned above. One example is given by the mixture of Gaussians model which has a rich history, e.g., see \cite{pearson1894contributions, reynolds1995robust, Dasgupta99, PermuterFJ03, vempala2004spectral, BelkinS10, MoitraV10, AndersonBGRV14, BhaskaraCMV14, GeHK15, regev2017learning}. While special cases of this problem can be solved by reduction to (symmetric) tensor decomposition \cite{HsuK13, AndersonBGRV14, BhaskaraCMV14}, general cases correspond to "sums of powers of quadratics" (and slightly more general models), i.e., the $t=2$ case. This observation is implicit in \cite{GeHK15}. \\

Let us briefly describe the mixture of Gaussians problem. We are given samples from a mixture of Gaussians in $n$ dimensions, $\cD = \sum_{i=1}^s w_i \cN(\mu_i, \Sigma_i)$, where $w_i, \mu_i, \Sigma_i$ denote the weight, mean and covariance matrix of the $i^{\text{th}}$ Gaussian respectively. The goal is to recover the parameters $\left(w_i, \mu_i, \Sigma_i \right)_{i=1}^s$ upto some error. There have been many algorithms developed for the mixture of Gaussians problem and they make varying assumptions about the input parameters. These assumptions can be grouped into three broad categories:
\begin{enumerate}
\item \textbf{Worst case}, e.g., \cite{pearson1894contributions, BelkinS10, MoitraV10}. Here, no assumptions are made on the input parameters. Due to this, the running time is exponential in the number of components $s$. In the worst case, even information theoretically, one needs exponential (in $s$) number of samples to learn the parameters \cite{MoitraV10, AndersonBGRV14}.
\item \textbf{Separation assumptions}, e.g., \cite{Dasgupta99, sanjeev2001learning, achlioptas2005spectral, kannan2005spectral, dasgupta2007probabilistic, brubaker2008isotropic, kumar2010clustering, regev2017learning}. Here, one assumes that the parameters (either the means or the covariance matrices) are well separated. Running times are typically polynomial in the number of components $s$.
\item \textbf{Smoothed setting}, e.g., \cite{HsuK13, AndersonBGRV14, BhaskaraCMV14, GeHK15}. Here, one does smoothed analysis for the problem, i.e., one starts with worst case parameters and perturbs them with some noise, and the goal is to solve the resulting instance. Somewhat surprisingly, the problem becomes easier as the dimension becomes larger.\footnote{This is because we can get more information from lower order moments in the higher dimensional case, which are easier to estimate.} Running times are typically polynomial in the dimension $n$, as long as the number of components $s \le \poly(n)$.\footnote{The exponent of the polynomial running time will depend on the exponent of the polynomial controlling the relation between $s$ and $n$.}
\end{enumerate}

Of course, it is best to have algorithms in the worst case but this usually means running time exponential time in the number of components.\footnote{This should be compared with some of the worst case reconstruction algorithms which run in time exponential in the top fan-in of the circuit.} Making reasonable assumptions on the parameters helps in designing more efficient algorithms (when the number of components are growing). The two assumptions here, separation and the smoothed setting, are incomparable. For example, the algorithms developed in the smoothed setting can also handle instances where the parameters are not well separated. Our contribution towards the mixture of Gaussians problem should be understood in the context of smoothed analysis of the problem. We also remark that there has been a lot of work in designing algorithms for mixtures of Gaussians in the robust setting, i.e., when some of the samples are corrupted adversarially \cite{diakonikolas2019robust, kothari2017outlier, kothari2018robust, hopkins2018mixture}. These algorithms usually work by estimating the moments from the corrupted data, so are in some sense complementary to algorithms that recover the parameters from the moments.
\\

In the smoothed setting, the papers \cite{AndersonBGRV14, BhaskaraCMV14} give a polynomial time algorithm for the special case of mixtures of Gaussians with diagonal covariance matrices (i.e., the different dimensions of every component of the mixture are independent) when the number of components $s \le \poly(n)$. For Gaussians mixtures with general covariance matrices, \cite{GeHK15} gave a polynomial time algorithm when $s \le \sqrt{n}$. We make a step towards getting a polynomial time algorithm for Gaussians mixtures with general covariance matrices when $s \le \poly(n)$. Our algorithm for learning sums of powers of low degree polynomials (specifically sums of powers of quadratics) allows us to recover the parameters of a \emph{non-degenerate} mixture of \emph{zero-mean} Gaussians given access to its $O(1)$-order \emph{exact} moments when $s \le \poly(n)$.

\begin{lemma}[Learning mixtures of zero-mean Gaussians] \label{lemm:intro_Gaussians} 
There is a randomized polynomial time algorithm which when given access to exact $O(1)$-order moments of a mixture of non-degenerate\footnote{The non-degeneracy condition is satisfied if $\Sigma_i = A_i A_i^T$ and the entries of $A_i$'s are choosen uniformly and independently from a set $S \subset \Q$ with $|S| \ge \poly(n,s)$, and of course all $w_i$'s are non-zero.} zero-mean Gaussians, $\cD = \sum_{i=1}^s w_i \cN(\mu_i, \Sigma_i)$, recovers the parameters $\left(w_i, \Sigma_i \right)_{i=1}^s$ with probability $1 - o(1)$.
\end{lemma}

We believe that a few modifications to our algorithm will allow one to get a polynomial time algorithm for general mixtures of Gaussians in the smoothed case. Let us remark on the main differences between our current algorithmic guarantee and the above goal.

\begin{itemize}
\item \textbf{Extension to general means.} Our current statement only holds for zero-mean Gaussians because applying the method of moments to zero-mean mixtures naturally leads to sums of powers of quadratics. This is just to keep the analysis simple and clean. We believe our algorithm should extend to sums of \emph{products} of low degree polynomials and the learning problems coming from moments of mixtures of general mean Gaussians lie somewhere between sums of powers of quadratics and sums of products of quadratics.
\item \textbf{Exact vs inexact moments.} Our algorithm assumes that the $O(1)$-order moments of the mixture are given exactly. Using samples, we can only approximate the moments ($O(1)$-order moments can be approximated to $1/\poly(n)$ accuracy using $\poly(n)$ samples). We leave it open for future work to modify our algorithm to handle $1/\poly(n)$ error.
\item \textbf{Smoothed vs non-degenerate setting.} Note that we state our result in the non-degenerate setting. This seems to be the right kind of assumption when given access to exact moments. If one is given access to inexact moments, one will need to control appropriate condition numbers whence smoothed setting is the right assumption to make.
\end{itemize}

We hope that our techniques will lead to progress on the smoothed analysis of mixtures of general Gaussians and also influence algorithms which work under other kinds of assumptions.

\subsection{Hardness of learning circuits in the worst case} \label{sec:hardness of learning}
In this section, we review some of the hardness results on learning circuits in order to gauge the difficulty of the problem for our circuit model and to place our result in context. Hardness of learning has been intensely studied for Boolean circuits as compared to arithmetic circuits. We state a few of these results from the Boolean world with the intent of drawing analogy. \\

\textbf{MCSP and approximate MCSP.}~ Circuit reconstruction is the arithmetic analogue of exact learning \cite{Angluin87} Boolean circuits from membership queries. Exact learning is closely related to the minimum circuit size problem (MCSP). MCSP for Boolean circuits is the following: Given the $N = 2^n$ size truth-table of an $n$-variate Boolean function $f$ and a number $s$, check if $f$ is computable by a Boolean circuit of size at most $s$. Analogously, in case of MCSP for arithmetic circuits, we are given the $N = {n+d \choose d}$ coefficient vector of an $n$-variate degree-$d$ polynomial $f$ and are required to determine if there is an arithmetic circuit of size at most $s$ computing $f$. %Efficient reconstruction necessarily gives an approximation of the minimum circuit size, and so, hardness of approximate\footnote{meaning multiplicative factor approximation of the minimum circuit size} MCSP implies hardness of reconstruction of general circuits. 
MCSP for Boolean circuits is not in P assuming the existence of cryptographically secure one-way functions \cite{KabanetsC00}. In fact, $N^{1-o(1)}$-approximate\footnote{Meaning multiplicative factor approximation of the minimum circuit size} MCSP for Boolean circuits is not in P under the same assumption \cite{AllenderH17}.\footnote{However, proving MCSP is NP-hard is quite demanding as that would imply $EXP \neq ZPP$ \cite{MurrayW17}, which is a long-standing open problem.} Analogous results about MCSP for arithmetic circuits are not known. However, drawing analogy with the Boolean world, it is plausible that $N^{1-\delta}$-approximate MCSP for general arithmetic circuits is not in P, for every constant $\delta > 0$. Here, $N = {n+d \choose d}$ is the size of the input coefficient vector. Such a hardness result may even be true for $d = n^{O(1)}$.  \\

MCSP for a circuit class $\CAL{C}$ is defined like MCSP except that we are now interested in checking if the input $f$  has a $\CAL{C}$-circuit of size at most $s$. It is known that the arithmetic analogue of MCSP is NP-hard for set-multilinear depth three circuits (tensors) \cite{Hastad90} and for depth three powering circuits (symmetric tensors) \cite{Shitov16}. It is also known that there is a constant $\delta \approx 0.0005$ such that $(1 + \delta)$-approximate MCSP is NP-hard for set-multilinear depth three circuits \cite{Swernofsky18, BlaserIJL18, SongWZ19}. Similar hardness results are known about MCSP for restricted Boolean circuit classes, e.g., MCSP for DNF is NP-hard \cite{Masek79,Czort99}. In fact, there is a $\delta > 0$ such that $(\log N)^{\delta}$-approximate MCSP is NP-hard for DNF \cite{Feldman09}\footnote{In contrast, a greedy algorithm solves $(\log N)$-approximate MCSP for DNF in $\poly(N)$ time \cite{Johnson74,Lovasz75,Chvatal79}. If the input is a DNF instead of a truth-table then we know of the following result: $s^{\frac{1}{4} - \epsilon}$ factor approximation of the minimum DNF size of an input DNF of size $s$ is not in P, for every $\epsilon \in (0,1)$, assuming $\Sigma_2^{p} \nsubseteq  \text{DTIME}(n^{O(\log n)})$ \cite{Umans99}.}. Approximate MCSP is a difficult problem even for $\text{AC}^0$ circuits: For every $\delta, \epsilon > 0$ there is a $h$ such that $N^{1-\delta}$-approximate MCSP for depth-$h$ circuits is not in BPP unless $m$-bit Blum integer\footnote{An integer of the form $pq$ for primes $p$ and $q$, where $p = q = 3 \mod 4$.} factorization is in $\text{BPTIME}(2^{m^\epsilon})$ \cite{AllenderHMPS08}\footnote{Similar results have also been shown for $\text{NC}^1$ and $\text{TC}^{0}$ circuits \cite{AllenderKR03}.}.  \\

\textbf{Learning implies lower bounds.}~ It was shown in \cite{FortnowK09} that a randomized polynomial-time (worst case, improper) reconstruction algorithm for an arithmetic circuit class $\CAL{C}$ implies the existence of a polynomial that can be computed on Boolean inputs in $BPEXP$ and cannot be computed by circuits in $\CAL{C}$ of polynomial size\footnote{A deterministic analogue of this result was shown in \cite{Volkovich16}.}. Similar results hold for Boolean circuits. Thus, if we are aiming for \emph{worst-case} polynomial-time reconstruction then we must necessarily focus on classes for which super-polynomial lower bounds are known. In fact, to our knowledge, all efficient reconstruction algorithms, in the worst or the average case, that are known till date are either for models for which non-trivial lower bounds are known (or for models that are incomplete). \\

\textbf{Hardness of PAC learning depth three arithmetic circuits.}~  \cite{KlivansS09} showed that PAC-learning depth three arithmetic circuit cannot be done in polynomial time unless the length of a shortest nonzero vector of an $n$-dimensional lattice can be approximated to within a factor of $\tilde{O}(n^{1.5})$ in polynomial time by a quantum algorithm. What this means is that it is hard to PAC learn the class of Boolean functions which match the output of polynomial-sized depth three arithmetic circuits on the Boolean hypercube. %The difference in the complexity of exact and approximate learning is also visible for $\text{AC}^0$ circuits: In the model of PAC learning under the uniform distribution, $\text{AC}^0$ circuits can be learned in quasi-polynomial time \cite{LinialMN93}, whereas the best exact learning algorithm for DNF has complexity $2^{\tilde{O}(n^{\frac{1}{3}})}$ \cite{KlivansS04}. 
%Exact learning is, by definition, harder than approximate\footnote{meaning the values of the output hypothesis and the input function are equal at most evaluation points} learning. If the output hypothesis of a learning algorithm for an arithmetic circuit class is an arithmetic circuit then the learning must be exact, as two distinct polynomials differ in their values at a lot of points, if the underlying field is sufficiently large. This is the reason reconstruction is usually defined as an exact learning process for arithmetic circuits. On the other hand, approximate learning makes sense for Boolean circuits, and also for arithmetic circuits if the output hypothesis is not a polynomial and we require the hypothesis to match the input polynomial at a large fraction of points from a specific domain like $\{-1,1\}^n$.
\\

\textbf{Membership queries versus random samples.}~ Membership queries provide an interactive model of learning as compared to learning from random samples. For some circuit classes, we know of efficient learning from membership queries but not learning from random samples. For instance, there is a deterministic polynomial-time algorithm for interpolating sparse polynomials using membership queries \cite{KlivansS01}, but the same is not known using random samples. The best known complexity for learning an $s$-sparse $n$-variate degree-$d$ real polynomial with error $\epsilon$ from random samples from the real cube $[-1,1]^n$ is $\poly(n,s,2^d,\frac{1}{\epsilon})$ \cite{AndoniPV014}. If the random samples are restricted to the Boolean hypercube $\{-1,1\}^n$ then exact learning of $s$-sparse real polynomials can be done in $\poly(n,2^s)$ time provided the input polynomial satisfies a certain property\footnote{This property is satisfied with high probability if the coefficients of the input polynomial are perturbed slightly by a random noise. Removing this condition on the input polynomial would immediately improve the state-of-the-arts of learning $\omega(1)$-juntas.} \cite{KocaogluSDK14}. Another example, in the Boolean world, is the quasi-polynomial time algorithm for PAC learning $\text{AC}^{0}[p]$ circuits under the uniform distribution from membership queries \cite{CarmosinoIKK16}. It is not known if the same can be achieved without membership queries (like in the case of $\text{AC}^{0}$-learning \cite{LinialMN93}). \\

%\Ankit{Merge below with the above.}

\textbf{Difficulty of learning homogeneous $\Sigma \wedge \Sigma \Pi^{[t]}$ circuits.}~ It turns out though that learning $\Sigma \wedge \Sigma \Pi^{[t]}$ circuits \emph{in the worst-case} is quite challenging. The reason is, if we can learn homogeneous $\Sigma \wedge \Sigma \Pi^{[t]}$ circuits efficiently then we can solve approximate MCSP for polynomial-size ABP\footnote{Algebraic branching programs (ABP) form a powerful circuit class -- a circuit can be converted to an ABP with only a quasi-polynomial blow up in size.} efficiently. This can be argued as follows: Suppose $f$ is a homogeneous $n$-variate polynomial of degree $d = n^{O(1)} < n$ that is computable by an ABP of size $\poly(n)$. Let $t \ll d$ be a number dividing $d$. Then, $f$ can be computed by a homogeneous $\Sigma \wedge \Sigma \Pi^{[t]}$ circuit\footnote{This circuit is obtained by homogenizing the ABP and then dividing it into pieces of length $t$ and multiplying out each piece. This gives a $\Sigma\Pi\Sigma\Pi^{[t]}$ circuit which is converted to a $\Sigma\wedge\Sigma\Pi^{[t]}$ circuit using Fischer's formula \cite{Fischer94}.} of size $n^{O(\frac{d}{t})}$. Now suppose we are able to learn homogeneous $\Sigma \wedge \Sigma \Pi^{[t]}$ circuits of size $\sigma$ in time $\poly(\sigma^t)$, for $t = \omega(1)$, and output $\poly(\sigma)$-size circuits. Then, we would succeed in solving $N^{o(1)}$-approximate MCSP for $\poly(n)$-size ABP in $\poly(N)$-time, where $N = {n+d \choose d}$. This appears to be a difficult task with our current knowledge on MCSP\footnote{See the discussion on MCSP at the start of this section.}. If such an efficient approximate MCSP for ABP is unattainable then there is no hope of learning homogeneous $\Sigma \wedge \Sigma \Pi^{[t]}$ circuits in $\poly(\sigma^t)$ time in the worst-case (where the output is a $\poly(\sigma)$-size circuit), for $t=\omega(1)$.

\subsection{Related work} \label{sec:previous work}
In view of the fact that learning general arithmetic circuits is probably a hard problem, research has focused on learning interesting special classes of circuits. Here, we give a brief account of some of these results from the literature. \\

\textbf{Low depth circuits.} A deterministic polynomial-time learning algorithm for $\Sigma\Pi$ circuits or sparse polynomials was given in \cite{KlivansS01}. Learning $\Pi\Sigma\Pi$ circuits in randomized polynomial-time follows from the classical circuit-factorization algorithm of \cite{KaltofenT90}. General $\Sigma\Pi\Sigma$ circuits are much harder to learn: It follows from a depth reduction result \cite{AgrawalV08, Koiran12, GKKS16, Tavenas13} that polynomial-time learning for $\Sigma\Pi\Sigma$ circuits implies sub-exponential time learning for general circuits. In \cite{Shpilka09}, a randomized $\qpoly(n,d,|\F|)$-time proper\footnote{\label{fn:rank}Provided the circuit satisfies a certain rank condition} learning algorithm was given for $\Sigma\Pi\Sigma$ circuits with two product gates over finite fields; if the circuit is additionally multilinear then the running time is $\poly(n,|\F|)$. The algorithm was derandomized and generalized  in \cite{KarninS09} to handle $\Sigma\Pi\Sigma$ circuits with constant number of product gates. Over fields of characteristic zero, \cite{Sinha16} gave a randomized proper\footref{fn:rank} learning algorithm for $\Sigma\Pi\Sigma$ circuits with two product gates. A randomized polynomial-time proper learning algorithm is known for multilinear $\Sigma\Pi\Sigma\Pi$ circuits with top fan-in two over any field \cite{GuptaKL12}. Recently, a deterministic proper learning algorithm is given in \cite{BhargavaSV19} for multilinear $\Sigma\Pi\Sigma\Pi$ circuits with constant top fan-in over finite fields; the running time is quasi-polynomial in the size of the circuit and $|\F|$. \\

\textbf{Read-once formulas and ABPs.} A deterministic polynomial-time proper learning algorithm is known for read-once formulas \cite{MinahanV18, ShpilkaV14}. Read-once oblivious algebraic branching programs (ROABPs) form an important subclass of ABPs that captures several other interesting and well-studied circuit models. There is a randomized polynomial-time proper learning algorithm for ROABP \cite{KlivansS06, BeimelBBKV00} which was derandomized in quasi-polynomial time in \cite{ForbesS13}. The method used for ROABP reconstruction can be adapted to give learning algorithms for set-multilinear ABPs and non-commutative ABPs \cite{ForbesS13}. \\

\textbf{Reconstruction under non-degeneracy conditions.} Reconstruction in the worst case appears to be an extremely hard problem even for circuit models for which good lower bounds are known. It is natural to ask -- Can we use the techniques used for proving lower bound for a circuit class $\CAL{C}$ to learn \emph{almost all} $\CAL{C}$-circuits? The notion of ``almost all $\CAL{C}$-circuits" is formalized as random $\CAL{C}$-circuits under some natural distribution, or preferably, as $\CAL{C}$-circuits satisfying a set of clearly stated non-degenerate conditions such that a random $\CAL{C}$-circuit (under any natural distribution) is non-degenerate with high probability. In \cite{GuptaKL11}, a randomized polynomial-time proper learning algorithm was given for non-degenerate\footnote{The papers \cite{GuptaKL11,GKQ14} state the results for random formulas, but it is not difficult to state the non-degeneracy conditions by taking a closer look at the algorithms.} multilinear formulas having fan-in two. A randomized polynomial-time proper learning algorithm for non-degenerate regular formulas having fan-in two was given in \cite{GKQ14}. An efficient randomized reconstrution for non-degenerate homogeneous ABPs of width at most $\frac{\sqrt{n}}{2}$ is presented in \cite{KayalNS19}. All the above reconstruction algorithms are implicitly connected to the corresponding lower bounds: a quasi-polynomial lower bound for multilinear formulas was already shown in \cite{Raz09}, a quasi-polynomial lower bound for regular formulas was proven in \cite{KayalSS14}\footnote{Note that here the lower bound comes later than the average case reconstruction algorithm; in fact, the ideas arising out of the reconstruction algorithm were helpful in proving the lower bound.} and a width lower bound of $n$ is also known for homogeneous ABPs \cite{Kumar19}. Recently, \cite{KayalS19} gave a randomized polynomial-time proper learning algorithm for non-degenerate homogeneous depth three circuits depending very explicitly on the ideas used in proving an exponential lower bound for this model \cite{NisanW97}. Also, randomized polynomial-time proper learning algorithms for non-degenerate depth three powering circuits are given in \cite{KayalS19, Garcia-MarcoKP18, Kayal12} which have implicit connections to the corresponding lower bound methods. \\

%\Ankit{Merge below with the above.}
\textbf{Tensor decomposition.}~ Tensor decomposition (which is the same as reconstruction of depth three set-multilinear circuits) has garnered a lot of attention in the machine learning community and a lot of algorithms have been developed for it \cite{harshman1970foundations, leurgans1993decomposition, de2007fourth, anandkumar2014tensor, AndersonBGRV14, BhaskaraCMV14, barak2015dictionary, ge2015decomposing, hopkins2016fast, ma2016polynomial}. Comparing to our model, symmetric tensor decomposition corresponds to learning sums of powers of linear forms. We do not know of any work that designs algorithms for sums of powers of degree-$t$ polynomials for $t > 1$, except for the work of \cite{GeHK15}. An algorithm for learning sums of cubes of quadratics in the non-degenerate case, with the number of summands upper bounded by $\sqrt{n}$ is implicit in \cite{GeHK15}. Their approach is to reduce the problem to tensor decomposition. However, we believe such an approach cannot be made to handle larger number of summands (say $\poly(n)$) even in the quadratic case as the lower bounds for sums of powers of quadratics need substantially newer ideas than the linear case as discussed in Section \ref{subsubsec:implementation}.
\\

\textbf{Lower bounds and PIT for $\Sigma \wedge \Sigma \Pi^{[t]}$ circuits.}~ Homogeneous $\Sigma \wedge \Sigma \Pi^{[t]}$ circuits have been well-studied in the context of lower bound and polynomial identity testing (PIT). Understanding how to prove lower bound for this model played a vital role in the proof of the exponential lower bound for homogeneous depth four circuits that emerged out of a chain of work \cite{Kayal12eccc,GKKS14,KayalSS14,FournierLMS15,KayalLSS17,KumarS17}. Also, a deterministic $s^{O(t \log s)}$-time black-box PIT algorithm is known for homogeneous $\Sigma \wedge \Sigma \Pi^{[t]}(s)$ circuits \cite{Forbes15}. %Reconstruction being a fundamental problem alongside lower bound and PIT, it is natural to investigate the complexity of learning this model. Besides, such a learning algorithm has potentially interesting applications in unsupervised learning. 
\\

\textbf{Improper learning for sums of powers of low degree polynomials.}~ In this paper, we focus on proper learning (i.e., the input and output representations are from the same circuit class). However, to our knowledge, there is no known efficient learning algorithm for sums of powers of degree-$t$ polynomials (worst case or average case) even for $t=2$, and even in the improper setting. For the $t=1$ case, a polynomial-time improper learning algorithm (worst case) follows from ROABP reconstruction \cite{BeimelBBKV00,KlivansS06} as sums of powers of linear forms (depth three powering circuits) is a subclass of ROABP. But, reconstruction for ROABP does not give a learning algorithm for sums of powers of quadratics as there is a power of a quadratic that requires an exponential-size ROABP \cite{Forbes15}. \\

To summarize, we have efficient learning algorithms (even under non-degeneracy conditions) only for some models for which good lower bounds are known (or for models that are incomplete). Moreover, barring a few exceptions like ROABP, read-once formulas and sparse polynomials, the circuit models for which efficient learning is known do have the fan-in of the sum gates very small (mostly bounded by a constant). In comparison, our strategy for translating techniques from lower bounds to learning works for a  much larger additive fan-in. Such a translation is only possible for learning under non-degeneracy condition as worst-case learning is arguably much harder\footnote{See the discussion on the difficulty of learning homogeneous $\Sigma \wedge \Sigma\Pi^{[t]}$ circuits in the worst case in Section \ref{sec:hardness of learning}.} than proving lower bound.

\subsection{Roadmap of the paper}

In Section \ref{sec:learning}, we state our algorithm for learning sums of powers of low degree polynomials and its analysis, and also prove that the non-degeneracy conditions are satisfied in the random case. In Section \ref{sec:gaussians}, we provide an algorithm for recovering the parameters of a non-degenerate mixture of zero-mean Gaussians given access to its exact $O(1)$-order moments. In Section \ref{sec:lowerbound}, we show how our new lower bound measure $\APP$ can be used to give alternate proofs of the lower bounds for homogeneous $\Sigma \Pi \Sigma \Pi^{[t]}$ circuits (almost matching the best known lower bounds which use the shifted partials measure). Section \ref{sec:open_problems} contains some of the interesting open problems and directions for future work. \\

In the Appendix, Section \ref{sec:appendix adjoint} mentions some important facts about the adjoint algebra and a proof of the uniqueness of decomposition in our setting. In Section \ref{sec: space to module decomposition}, we show a reduction from the vector space decomposition\footnote{In fact a generalization of it.} problem to the module decomposition problem. Section \ref{sec: limitation of SP} contains a discussion about why the shifted partials measure is unlikely to satisfy the non-degeneracy conditions required for our learning framework. Finally, Sections \ref{sec:proofs from sec learning} and \ref{sec:proofs from sec lowerbound} supply the missing proofs from Sections \ref{sec:learning} and \ref{sec:lowerbound}, respectively.

\section{Learning sums of powers of low degree polynomials} \label{sec:learning}
We prove Theorem \ref{thm:learning sums of powers} in this section. Our algorithm is an implementation of the lower bound to learning strategy (proposed in Section \ref{sec:learning from lower bound}) for homogeneous $\Sigma \wedge \Sigma \Pi^{[t]}$ circuits\footnote{We deviate from the strategy slightly, by introducing an intermediate \emph{multi-gcd} step (see Algorithm \ref{alg:learning sums of powers}), in order to make the analysis simpler.}. Section \ref{sec:algorithm} describes the algorithm. Section \ref{sec:analysis} describes the analysis of the algorithm. In Section \ref{sec: random formula}, we prove that a random $\Sigma \wedge \Sigma \Pi^{[t]}$ circuit satisfies our non-degeneracy conditions. Finally Section \ref{sec:parameters} lists some relations between various parameters which are needed for the analysis to work. \\

For simplicity of presentation, we will assume that $\F$ is a \emph{finite field} of sufficiently large size and characteristic. The analysis goes through over any $\F$ that satisfies the restrictions on size and characteristic stated in Theorem \ref{thm:learning sums of powers} -- we simply have to work with a sufficiently large subset of $\F$ and do a few minor changes to the algorithm and its analysis.

\subsection{The algorithm} \label{sec:algorithm}
We are given black-box access to an $n$-variate degree-$d$ polynomial $f$ that is computed by a homogeneous $\Sigma \wedge \Sigma \Pi^{[t]}(s)$ formula $C$, i.e.,
\begin{equation} \label{eqn:input sum of powers}
f(\vecx) = c_1Q_1^m + \ldots + c_s Q_s^m,
\end{equation}
where each $c_i\in \F^{\times}$, $Q_i$ is a homogeneous polynomial of degree $t$, and $tm = d$. Moreover, formula $C$ is non-degenerate (see Definition \ref{defn: non-degeneracy}). Assume that the algorithm knows\footnote{If $s$ is unknown, we can simply go over $s$ incrementally (starting from $1$ and going up to the upper bound stated in Theorem \ref{thm:learning sums of powers}) and run the algorithm for each $s$. A randomized identity test at the end of the algorithm determines if we have learnt the circuit correctly with high probability.} $d, t$ and $s$ and these parameters satisfy the conditions stated in Theorem \ref{thm:learning sums of powers}. The task is to output a homogeneous $\Sigma \wedge \Sigma \Pi^{[t]}(s)$ formula for $f$ efficiently. The parameters $k, n_0$ and $m_0$, in Algorithm \ref{alg:learning sums of powers}, are chosen according to Proposition \ref{prop:setting parameters to satisfy a bunch of conditions}, which is stated in Section \ref{sec:parameters}. A \emph{random linear form} is a linear form whose coefficients are chosen independently and uniformly at random from $\F$. A \emph{tuple of $n$ random linear forms} is a tuple of $n$ independently chosen random linear forms. 

\begin{algorithm}
	\caption{Learning sums of powers of degree-$t$ polynomials } \label{alg:learning sums of powers}
	\begin{algorithmic}
		\STATE \textbf{Input}: Black-box access to an $f \in \F[\vecx]$ that is computed by a non-degenerate homogeneous $\Sigma \wedge \Sigma \Pi^{[t]}(s)$ formula $C$ (as in Equation \eqref{eqn:input sum of powers}), i.e., $f = c_1Q_1^m + \ldots + c_s Q_s^m$.
		\STATE \textbf{Output}: A non-degenerate homogeneous $\Sigma \wedge \Sigma \Pi^{[t]}(s)$ formula computing $f$.
	\end{algorithmic}
	\begin{algorithmic}[1]
		\STATEx
		\STATEx \begin{center}\textcolor{Gray}{/* Constructing two sets of linear operators and obtaining the relevant vector spaces  */} \end{center}
		\STATE \label{mainalgo:step1} Pick a tuple of $n$ random linear forms $L = (\ell_1(\vecz), \ldots, \ell_n(\vecz))$, where $|\vecz| = n_0$. Let $\CAL{L}_1$ be the set of operators $\pi_{L}(\partial_{\vecx}^k ~)$.
		\STATE \label{mainalgo:step2} Compute black-box access to a basis of $U = \left\langle \CAL{L}_1 \circ f \right\rangle = \app{L}{\vecx}{k}{f}$.
		\STATE \label{mainalgo:step3} (\emph{Multi-gcd step})~ Compute black-box access to a basis of $V = \left\langle \pi_L(Q_1)^{m-k}, \ldots, \pi_L(Q_s)^{m-k}\right\rangle$ using the basis of $U$. It will follow from Observation \ref{obs: random L gives tilde U a direct sum structure} that
								$$V = V_1 \oplus \ldots \oplus V_s,~~~~~~ \text{(with high probability)}$$
		where $V_i = \left\langle G_i^e \right\rangle$, $G_i = \pi_L(Q_i)$ and $e = m-k$. Get black-box access to a random $g(\vecz)$ in $V$. 
		\STATE \label{mainalgo:step4} Pick a tuple of $n_0$ random linear forms $P = (p_1(\vecw), \ldots, p_{n_0}(\vecw))$, where $|\vecw| = m_0$. Let $\CAL{L}_2$ be the set of operators $\pi_{P}(\partial_{\vecz}^k ~)$.
		\STATE \label{mainalgo:step5} Compute black-box access to a basis of $W = \app{P}{\vecz}{k}{g}$. It will follow from Proposition \ref{prop: decomposition of W},
								$$W = W_1 \oplus \ldots \oplus W_s,~~~~~~ \text{(with high probability)}$$
		where $W_i := \app{P}{\vecz}{k}{G_i^{e}} = \left\langle \CAL{L}_2 \circ V_i\right\rangle$.
		\STATEx
		\STATEx \begin{center}\textcolor{Gray}{/* Decomposing the vector spaces */} \end{center}
		\STATE \label{mainalgo:step6} Compute black-box access to bases of $V_1, \ldots, V_s$ by decomposing $V$ and $W$ under the action of $\CAL{L}_2$ into indecomposable subspaces. 
		\STATEx
		\STATEx \begin{center}\textcolor{Gray}{/* Recovering the terms of the formula */} \end{center}
		\STATE \label{mainalgo:step7} Run Steps \ref{mainalgo:step1}-\ref{mainalgo:step6} ``$d$ times''\footnotemark ~to compute black-box access to $c'_1 Q_1(\vecx)^e, \ldots, c'_s Q_s(\vecx)^e$ for some constants $c'_1, \ldots, c'_s \in \F^{\times}$.
		\STATE \label{mainalgo:step8} Compute (dense representations of) the polynomials $\hat{c}_1 Q_1(\vecx), \ldots, \hat{c}_s Q_s(\vecx)$ for some constants $\hat{c}_1, \ldots, \hat{c}_s \in \F^{\times}$. Output a homogeneous $\Sigma \wedge \Sigma \Pi^{[t]}(s)$ formula for $f$.
								
	\end{algorithmic}
\end{algorithm}
\footnotetext{this will be explained in the analysis of this step}

\subsection{Analysis of the algorithm} \label{sec:analysis}
We analyze the correctness and efficiency of the algorithm in this section. The three main segments of the algorithm are Steps \ref{mainalgo:step1}-\ref{mainalgo:step5}, Step \ref{mainalgo:step6} and Steps \ref{mainalgo:step7}-\ref{mainalgo:step8} -- we examine these one by one. The missing proofs of the technical statements are given in Section \ref{sec:proofs from sec learning} of the appendix.

\subsubsection*{Steps \ref{mainalgo:step1}-\ref{mainalgo:step5}: Constructing two sets of linear operators and obtaining the relevant vector spaces}
Let $\sigma$ be the size of the non-degenerate formula $C$ that computes $f$. In Step \ref{mainalgo:step2}, the algorithm computes black-box access to a basis of $U = \app{L}{\vecx}{k}{f}$, where $L$ is a tuple of $n$ random linear forms in $n_0$ many $\vecz$-variables. This can be done in $\poly(\sigma, s^t)$ time, with success probability $1-o(1)$, due to the choice of $k$ (see Proposition \ref{prop:setting parameters to satisfy a bunch of conditions}) and the following easily verifiable fact.

\begin{fact} \label{fact: computing a basis of the space of partials is easy}
Given black-box access to an $n$-variate degree-$d$ polynomial $f(\vecx)$, black-box access to the polynomials in $\der{\vecx}{k}{f}$ can be computed in determimistic $\poly((nd)^k)$ time. Given black-box access to $n$-variate degree-$d$ polynomials $g_1, \ldots, g_r$, black-box access to a basis of $\left \langle g_1, \ldots, g_r \right \rangle$ can be computed in randomized $\poly(n,d,r)$ time with probability at least $1 - \frac{rd}{|\F|}$.
\end{fact}

\begin{observation} \label{obs: random L gives U a direct sum structure}
Let $U_i = \app{L}{\vecx}{k}{Q_i^m}$. With probability $1 - o(1)$ over the randomness of $L$,
$$U = U_1 \oplus \ldots \oplus U_s ~~\text{ and }~~ \dim U_i = {n_0 + k(t-1) -1 \choose k(t-1)} ~\text{ for all } i \in [s].$$
\end{observation}

It follows from the above observation that $U_i = \left \langle \vecz^{k(t-1)} \cdot \pi_L(Q_i)^{m-k} \right \rangle$. \\

\textbf{The multi-gcd step.}~ In Step \ref{mainalgo:step3}, the algorithm computes a basis of $V = \left\langle G_1^{e}, \ldots, G_s^{e}\right\rangle$, where $G_i = \pi_L(Q_i)$ and $e = m-k$. This step can be executed in $\poly(\sigma, s^t)$ time as follows:  

\begin{observation} \label{obs: random L gives tilde U a direct sum structure}
Let $\widetilde{U_i} := \left \langle \vecz^{2k(t-1)} \cdot G_i^{e} \right \rangle$. With probability $1 - o(1)$ over the randomness of $L$,
$$\widetilde{U_1} + \ldots + \widetilde{U_s} = \widetilde{U_1} \oplus \ldots \oplus \widetilde{U_s}.$$
\end{observation}  

It follows from the above that $V = V_1 \oplus \ldots \oplus V_s$, where $V_i = \left\langle G_i^e \right\rangle$. The observation also helps prove the next proposition which gives a way to compute a basis of $V$ efficiently. 

\begin{proposition} \label{prop: multi-gcd}
Let $r = {n_0 + k(t-1) - 1 \choose k(t-1)}$ and $f_1, \ldots, f_{sr}$ be a basis of $U$. Let $z_1$ and $z_2$ be two distinct variables in $\vecz$. Then, the following statements hold:
\begin{enumerate}
\item If $g(\vecz) \in V$ then there exist $a_1, \ldots a_{sr}, b_1, \ldots, b_{sr} \in \F$ such that
$$\frac{a_1f_1 + \ldots + a_{sr}f_{sr}}{z_1^{k(t-1)}} = \frac{b_1f_1 + \ldots + b_{sr}f_{sr}}{z_2^{k(t-1)}} = g.$$
\item If there exist $a_1, \ldots a_{sr}, b_1, \ldots, b_{sr} \in \F$ such that
$$\frac{a_1f_1 + \ldots + a_{sr}f_{sr}}{z_1^{k(t-1)}} = \frac{b_1f_1 + \ldots + b_{sr}f_{sr}}{z_2^{k(t-1)}},$$
then $\frac{a_1f_1 + \ldots + a_{sr}f_{sr}}{z_1^{k(t-1)}}$ is a polynomial $g(\vecz)$ in $V$.  
\end{enumerate} 
\end{proposition}

Let $f_1, \ldots, f_{sr}$ be the basis of $U$ obtained in Step \ref{mainalgo:step2}. It follows from the above proposition that a basis of the space $\CAL{S}$ defined by all $(a_1, \ldots a_{sr}, b_1, \ldots, b_{sr}) \in \F^{2sr}$ satisfying
\begin{equation} \label{eqn:linear system for multi-gcd}
a_1 \cdot z_2^{k(t-1)}f_1 ~+~ \ldots ~+~ a_{sr} \cdot z_2^{k(t-1)}f_{sr} ~-~ b_1 \cdot z_1^{k(t-1)}f_1 ~-~ \ldots ~-~ b_{sr} \cdot z_1^{k(t-1)}f_{sr} ~=~ 0
\end{equation} 
gives a basis of $z_1^{k(t-1)} \cdot V$ (and a basis of $z_2^{k(t-1)} \cdot V$). As we have black-box access to $f_1, \ldots, f_{sr}$, we can plug in $2sr$ random values to the $\vecz$-variables in Equation \eqref{eqn:linear system for multi-gcd} and derive a linear system in the ``variables'' $a_1, \ldots a_{sr}, b_1, \ldots, b_{sr}$. A solution to this system gives a basis of $\CAL{S}$ with probability $1 - o(1)$, thereby giving black-box access to a basis of $z_1^{k(t-1)} \cdot V$. Now, using black-box polynomial factorization\footnote{In fact, a simpler argument works here as the factorization is special.} \cite{KaltofenT90}, we get black-box access to a basis of $V$. This completes the multi-gcd step. \\

Let $g_1, \ldots, g_s$ be the basis of $V$ obtained in Step \ref{mainalgo:step3}. This step also chooses black-box access to a random $g \in V$, i.e., $g = a_1g_1 + \ldots + a_sg_s$, where $a_1, \ldots, a_s$ are picked independently and uniformly at random from $\F$. In Step \ref{mainalgo:step5}, the algorithm computes black-box access to a basis of $W = \app{P}{\vecz}{k}{g}$, where $P$ is a tuple of $n_0$ random linear forms in $m_0$ many $\vecw$-variables. By Fact \ref{fact: computing a basis of the space of partials is easy}, this can be done in $\poly(\sigma)$ time with success probability $1-o(1)$. 

\begin{proposition} \label{prop: decomposition of W}
Let $W_i = \app{P}{\vecz}{k}{G_i^e}$. With probability $1 - o(1)$ over the randomness of $L$ and $P$,
$$W = W_1 \oplus \ldots \oplus W_s ~~\text{ and }~~ \dim W_i = {m_0 + k(t-1) -1 \choose k(t-1)} ~\text{ for all } i \in [s].$$
\end{proposition}

\subsubsection*{Step \ref{mainalgo:step6}: Decomposing the vector spaces} 
In Step \ref{mainalgo:step6}, the algorithm computes black-box access to bases of $V_1 = \left\langle G_1^e \right\rangle, \ldots, V_s = \left\langle G_s^e \right\rangle$ by decomposing the spaces $V$ and $W$ under the action of the set of operators $\CAL{L}_2 = \pi_{P}(\partial_{\vecz}^k ~)$. We now explain how this is carried out efficiently. 

\begin{definition} [Indecomposable decomposition] \label{defn: invariant and indecomposable subspaces}
Let $V$ and $W$ be two vector spaces and $\CAL{L}$ a set of linear operators from $V$ to $W$. A decomposition of the spaces $V$ and $W$ as:
\begin{eqnarray*}
V &=& V_1 \oplus \ldots \oplus V_s \\
W &=& W_1 \oplus \ldots \oplus W_s
\end{eqnarray*}
is \emph{indecomposable} under the action of $\CAL{L}$ if the following hold for every $i \in [s]$,
\begin{enumerate} [(a)]
\item $\left \langle \CAL{L} \circ V_i \right \rangle \subseteq W_i$
\item There do \emph{not} exist spaces $V_{i1}, V_{i2}, W_{i1}, W_{i2}$ such that 
\begin{eqnarray*}
V_i = V_{i1} \oplus V_{i2}, &\quad \quad& W = W_{i1} \oplus W_{i2} ~~\text{and} \\
\left \langle \CAL{L} \circ V_{i1} \right \rangle \subseteq W_{i1}, &\quad \quad&  \left \langle \CAL{L} \circ V_{i2} \right \rangle \subseteq W_{i2}.
\end{eqnarray*}
\end{enumerate}  
\end{definition}
In our case $W = \left \langle \CAL{L}_2 \circ V \right \rangle$ and $W_i = \left \langle \CAL{L}_2 \circ V_{i} \right \rangle$ (by Proposition \ref{prop: decomposition of W}). Also, the decomposition $V = V_1 \oplus \ldots \oplus V_s$ and $W = W_1 \oplus \ldots \oplus W_s$ is indecomposable under the action of $\CAL{L}_2$ as $\dim V_i = 1$ for all $i \in [s]$. It remains to show that this indecomposable decomposition is unique and it can be computed efficiently. Towards this, we take inspiration from Section \ref{sec:appendix adjoint} (particularly, Corollary \ref{cor:diagonalizable adjoint}) and analyze a suitable adjoint algebra. \\

\noindent \textbf{The adjoint algebra.~} Recall, $g_1, \ldots, g_s$ is the basis of $V$ computed in Step \ref{mainalgo:step3}. Let $h_1, \ldots, h_{sq}$ be the basis of $W$ computed in Step \ref{mainalgo:step5}, where $q = {m_0 + k(t-1)-1 \choose k(t-1)} = |\vecw^{k(t-1)}|$. With regard to the bases $g_1, \ldots, g_s$ and $h_1, \ldots, h_{sq}$, every element of $\CAL{L}_2$ can be naturally identified with a $sq \times s$ matrix by identifying $V$ with $\F^s$ and $W$ with $\F^{sq}$. We will work with this matrix representation of the elements of $\CAL{L}_2$ which can be computed in $\poly(\sigma)$ time from black-box access to $g_1, \ldots, g_s$ and $h_1, \ldots, h_{sq}$ (using Fact \ref{fact: computing a basis of the space of partials is easy}, Proposition \ref{prop:setting parameters to satisfy a bunch of conditions} and solving linear systems). Let
\begin{equation} \label{eqn: adjoint for rectangular matrices}
\adj(\CAL{L}_2) := \left\{ (D,E) \in M_s(\F) \times M_{sq}(\F)~:~ KD=EK ~\text{ for all }~ K \in \CAL{L}_2\right\}.
\end{equation}
Observe that $\adj(\CAL{L}_2)$ is an $\F$-subalgebra of $M_s(\F) \times M_{sq}(\F)$ and a basis of $\adj(\CAL{L}_2)$ can be computed in $\poly(\sigma)$ time by solving a system of linear equations arising from the equation $KD = EK$ for all $K \in \CAL{L}_2$. Define
\begin{equation} \label{eqn: adjoint for rectangular matrices component 1}
\adj(\CAL{L}_2)_1 := \left\{ D \in M_s(\F) ~:~ \text{there exists an }~ E \in M_{sq}(\F) ~\text{ such that }~ (D,E) \in \adj(\CAL{L}_2) \right\}.
\end{equation}
Clearly, $\adj(\CAL{L}_2)_1$ is an $\F$-subalgebra of $M_s(\F)$ and computing a basis of $\adj(\CAL{L}_2)_1$ from a basis of $\adj(\CAL{L}_2)$ is a simple task. The following proposition shows that $\adj(\CAL{L}_2)_1$ is diagonalizable. \\

Let $A \in \GL_s(\F)$ be the basis change matrix from $(g_1, \ldots, g_s)$ to $(G_1^e, \ldots, G_s^e)$ and
$$\CAL{D} := \left\{ \diag(a_1, \ldots, a_s) ~:~ a_i \in \F ~\text{ for all }~ i \in [s]\right\} ~\subset~ M_s(\F).$$

\begin{proposition} \label{prop: adjoint is diagonalizable}
$A \cdot \adj(\CAL{L}_2)_1 \cdot A^{-1} = \CAL{D}$.
\end{proposition}

\underline{\emph{Diagonalizing $\adj(\CAL{L}_2)_1$}}.~ Use the basis of $\adj(\CAL{L}_2)_1$ to pick a random matrix $D \in_r \adj(\CAL{L}_2)_1$. By the above proposition, and as $|\F| \gg s^2$, the eigenvalues of $D$ are distinct with probability $1-o(1)$. Compute the eigenvalues $a_1, \ldots, a_s$ by factorizing\footnote{This is where we need the assumption that univariate polynomial factorization over $\F$ can be done in randomized polynomial time.} the characteristic polynomial of $D$. Now, compute an $\tilde{A} \in \GL_s(\F)$ such that $\tilde{A}D\tilde{A}^{-1} = \diag(a_1, \ldots, a_s)$ -- this can be done by solving a linear system. As $D$ has distinct eigenvalues and $ADA^{-1}$ is also diagonal, there exist a permutation matrix $P \in \GL_s(\F)$ and a diagonal matrix $S \in \GL_s(\F)$ such that
$$\tilde{A} = PS \cdot A.$$ 

\underline{\emph{Computing bases of $V_1, \ldots, V_s$}}.~ Observe that
\begin{eqnarray*}
(G_1^e ~ G_2^e ~ \ldots ~ G_s^e) \cdot A &=& (g_1 ~ g_2 ~ \ldots ~ g_s) \quad \quad \text{(by definition of $A$)}\\
\Rightarrow (G_1^e ~ G_2^e ~ \ldots ~ G_s^e) \cdot S^{-1}P^{-1} &=& (g_1 ~ g_2 ~ \ldots ~ g_s) \cdot \tilde{A}^{-1}.
\end{eqnarray*}
Thus, by computing black-box access to the entries of the vector $(g_1 ~ g_2 ~ \ldots ~ g_s) \cdot \tilde{A}^{-1}$, we get black-box access to $G_1^e, \ldots, G_s^e$ (up to permutation and scaling). By relabeling, we can assume that Step \ref{mainalgo:step6} computes black-box access to bases of $V_1 = \left \langle \pi_L(Q_1)^{e} \right \rangle, \ldots, V_s = \left \langle \pi_L(Q_s)^{e} \right \rangle$ in order.

\subsubsection*{Steps \ref{mainalgo:step7}-\ref{mainalgo:step8}: Recovering the terms of the formula}
At the end of Step \ref{mainalgo:step6}, we have black-box access to $c_1' \pi_L(Q_1)^{e}, \ldots, c_s' \pi_L(Q_s)^{e}$, where $c_1', \ldots, c_s' \in \F$ and $Q_1, \ldots, Q_s \in \F[\vecx]$ are unknown, but $L = (\ell_1(\vecz), \ldots, \ell_n(\vecz))$ is known from Step \ref{mainalgo:step1}. The idea now is to get black-box access to $c_1' Q_1(\vecx)^{e}, \ldots, c_s' Q_s(\vecx)^{e}$ by executing Steps \ref{mainalgo:step1}-\ref{mainalgo:step6} several times (each time by altering $L$ slightly). Hereafter, the algorithm uses black-box polynomial factorization \cite{KaltofenT90} to get black-box access to $\hat{c}_1 Q_1(\vecx), \ldots, \hat{c}_s Q_s(\vecx)$ for some constants $\hat{c}_1, \ldots, \hat{c}_s \in \F^{\times}$. Then, the sparse polynomial interpolation algorithm of \cite{KlivansS01} gives dense representations of the $\sigma$-sparse polynomials $\hat{c}_1 Q_1(\vecx), \ldots, \hat{c}_s Q_s(\vecx)$. Finally, we obtain a homogeneous $\Sigma \wedge \Sigma \Pi^{[t]}(s)$ formula computing $f$ by solving a linear system. Let us see how this idea is made to work. \\

\underline{\emph{Fixing the query points}}. The following remarks imply that the points at which the algorithm needs to query $c_1' Q_1(\vecx)^{e}, \ldots, c_s' Q_s(\vecx)^{e}$ in order to employ the black-box polynomial factorization algorithm and the sparse polynomial interpolation algorithm can be fixed \emph{a priori} right after Step \ref{mainalgo:step6}. \\  

\textbf{Remarks.}
\begin{enumerate}
\item The sparse polynomial interpolation algorithm of \cite{KlivansS01} works with non-adaptive queries, i.e., each subsequent query point does not depend on answers to the previous queries.
\item The black-box polynomial factorization algorithm of \cite{KaltofenT90} also works with non-adaptive queries. In other words, once the set of points at which we need to evaluate the irreducible factors of an input polynomial $f$ (given as a black-box) is fixed, the algorithm uses only non-adaptive queries to $f$ in order to compute these evaluations.
\end{enumerate}

\underline{\emph{Evaluating $c_1' Q_1(\vecx)^{e}, \ldots, c_s' Q_s(\vecx)^{e}$ at a query point}}. Let $\veca = (a_1, \ldots, a_n) \in \F^n$ be a query point. We wish to compute $c_1' Q_1(\veca)^{e}, \ldots, c_s' Q_s(\veca)^{e}$ from black-box access to $c_1' \pi_L(Q_1)^{e}, \ldots, c_s' \pi_L(Q_s)^{e}$, where $L = (\ell_1(\vecz), \ldots, \ell_n(\vecz))$ and $\vecz = (z_1, \ldots, z_{n_0})$. Let
$$\ell_l(\vecz) = r_{l1} z_1 + \ldots + r_{ln_0} z_{n_0},$$
where $r_{l1}, \ldots, r_{ln_0} \in \F$ are chosen uniformly and independently at random from $\F$ (in Step \ref{mainalgo:step1}) for all $l \in [n]$. Now, pick $(r_1, \ldots, r_n) \in_r \F^n$. For each $y \in \{1, \ldots, d\}$, define
$$\tilde{\ell_l} (y,\vecz) := (yr_l + (1-y)a_l)z_1 + r_{l2} z_2 + \ldots + r_{ln_0} z_{n_0},$$
for every $l \in [n]$. Observe that $r'_{l1} := yr_l + (1-y)a_l$ is uniformly distributed over $\F$ as $r_l$ is chosen randomly from $\F$. Moreover, $r_{11}', r_{12}, \ldots, r_{1n_0}, ~\ldots~, r_{n1}', r_{n2}, \ldots, r_{nn_0}$ are independent of each other as $r_1, \ldots, r_n$ are independently chosen. Hence,
$$\tilde{L}(y) := (\tilde{\ell_1} (y,\vecz), \ldots, \tilde{\ell_n} (y,\vecz))$$
is a tuple of random linear forms in the $\vecz$-variables for every $y \in \{1, \ldots, d\}$. If we execute Steps \ref{mainalgo:step1}-\ref{mainalgo:step6} by replacing $L$ by $\tilde{L}(y)$ in Step \ref{mainalgo:step1} then we will get black-box access to 
\begin{equation} \label{eqn: tilde L}
\tilde{c}_{\rho(1)} \cdot \pi_{\tilde{L}(y)} (Q_{\rho(1)})^e, ~\ldots~, \tilde{c}_{\rho(s)} \cdot \pi_{\tilde{L}(y)} (Q_{\rho(s)})^e, \quad \quad \text{with probability $1-o(1)$},
\end{equation}
where $\rho$ is an unknown permutation of $[s]$ and $\tilde{c}_1, \ldots, \tilde{c}_s \in \F^{\times}$ are also unknown. We can find $\rho$ efficiently as follows: Observe that $\tilde{L}(y)_{z_1=0} = L_{z_1=0}$. Hence, the ratio
\begin{eqnarray*}
\frac{c_i' \cdot \pi_{L_{z_1=0}}(Q_i)^{e}}{\tilde{c}_{j} \cdot \pi_{\tilde{L}(y)_{z_1=0}} (Q_{j})^e} = \frac{c_i' \cdot [G_i^e]_{z_1=0}}{\tilde{c}_{j} \cdot [G_j^e]_{z_1=0}} &=& \frac{c_i'}{\tilde{c}_{i}} \quad \quad \text{if $i = j$}, \\
		&=& \text{ a non-constant rational function in $\vecz$,} \quad \text{if $i \neq j$}.
\end{eqnarray*}
The second equality is because of Condition \ref{cond4} of the non-degeneracy condition (Definition \ref{defn: non-degeneracy}), which states that there is an $L$ such that $[G_1^e]_{z_1=0}, \ldots, [G_s^e]_{z_1=0}$ are $\F$-linearly independent. Thus, for a random $L$, $[G_1^e]_{z_1=0}, \ldots, [G_s^e]_{z_1=0}$ are $\F$-linearly independent with probability $1-o(1)$. Now, we can discover the permutation $\rho$ by evaluating the ratio
$$\frac{c_i' \cdot \pi_{L_{z_1=0}}(Q_i)^{e}}{\tilde{c}_{j} \cdot \pi_{\tilde{L}(y)_{z_1=0}} (Q_{j})^e}$$
at $\poly(d)$-many random points in $\F^{n_0-1}$ and checking if all the evaluations are the same. This process succeeds with probability $1 - o(1)$ (as $|\F|$ is sufficiently large) and also gives us $\frac{c_i'}{\tilde{c}_{i}}$ for all $i \in [s]$. From Equation \eqref{eqn: tilde L} and the knowledge of $\rho$ and $\frac{c_i'}{\tilde{c}_{i}}$ we obtain black-box access to
$$c_1' \cdot \pi_{\tilde{L}(y)}(Q_1)^{e}, ~\ldots~, c_s' \cdot \pi_{\tilde{L}(y)}(Q_s)^{e}.$$
By setting $z_1=1$ and $z_2= \ldots = z_s = 0$ we get 
$$p_i(y) := c_i' \cdot Q_i(yr_1 + (1-y)a_1, ~\ldots~, yr_n + (1-y)a_n)^e,$$
for every $i\in [s]$. As $y$ is arbitrarily fixed in $[d]$, we can compute $p_i(1), \ldots, p_i(d)$ for all $i \in [s]$ \footnote{By union bound, the total error probability remains $o(1)$ as $|\F|$ is sufficiently large.}. Treating $p_i(y)$ as a univariate polynomial in $y$ and observing that $\deg_y p_i < d$, we can interpolate the polynomial $p_i(y)$ from the above evaluations. Notice that $p_i(0) = c_i' Q_i(\veca)^{e}$ for all $i \in [s]$. \\

\underline{\emph{Outputting a $\Sigma \wedge \Sigma \Pi^{[t]}(s)$ formula}}. As explained before, the black-box polynomial factorization algorithm and the sparse polynomial interpolation algorithm together give us dense representations of the polynomials $\hat{c}_1 Q_1(\vecx), \ldots, \hat{c}_s Q_s(\vecx)$ for some unknown $\hat{c}_1, \ldots, \hat{c}_s \in \F$. We know that there exist $u_1, \ldots, u_s \in \F$ such that
$$f = u_1 \cdot [\hat{c}_1 Q_1(\vecx)]^m ~+~ \ldots ~+~ u_s \cdot [\hat{c}_s Q_s(\vecx)]^m.$$
By treating $u_1, \ldots, u_s$ as formal variables, we can obtain a linear system in $u_1, \ldots, u_s$ by evaluating $f$ and $[\hat{c}_1 Q_1(\vecx)]^m, \ldots, [\hat{c}_s Q_s(\vecx)]^m$ at $s$ random points in $\F^n$. As $[\hat{c}_1 Q_1(\vecx)]^m, \ldots, [\hat{c}_s Q_s(\vecx)]^m$ are $\F$-linearly independent (which follows from the non-degeneracy condition), the solution to the system gives $u_1, \ldots, u_s$ that satisfy the above equation with probability $1-o(1)$.

\subsection{A random $\Sigma\wedge\Sigma\Pi^{[t]}$ circuit is non-degenerate (proof of Lemma \ref{lem:non-degenerate_random})} \label{sec: random formula}
We show that a homogeneous $\Sigma \wedge \Sigma \Pi^{[t]}(s)$ formula $c_1Q_1^m + \ldots + c_s Q_s^m$ is non-degenerate with probability $1-o(1)$ if the coefficients of the degree-$t$ polynomials $Q_1, \ldots, Q_s$ are chosen independently and uniformly at random from a set $S \subseteq \F$ of size at least $(ns)^{150 \cdot t}$. For this, it is sufficient to show the existence of one non-degenerate homogeneous $\Sigma \wedge \Sigma \Pi^{[t]}(s)$ formula, as long as $|S| \gg 3ds \cdot {n_0 + 2kt -1 \choose 2kt}$ is sufficiently large (Proposition \ref{prop:setting parameters to satisfy a bunch of conditions}). This is because of Schwarz-Zippel lemma and the fact that all the non-degeneracy conditions are about vanishing of some determinant.

\subsubsection*{Construction of a non-degenerate homogeneous $\Sigma \wedge \Sigma \Pi^{[t]}$ formula}
We construct a homogeneous $\Sigma \wedge \Sigma \Pi^{[t]}(s)$ formula that satisfies non-degeneracy Conditions \ref{cond1}, \ref{cond3} and \ref{cond4} (Definition \ref{defn: non-degeneracy}). It would easily follow from this construction that there is a homogeneous $\Sigma \wedge \Sigma \Pi^{[t]}(s)$ formula that also satisfies Condition \ref{cond2} of non-degeneracy. Let $\vecx = \vecy \uplus \vecz$, where $\vecz = \{z_1, \ldots, z_{n_0}\}$. The value of $n_0$ is fixed in Proposition \ref{prop:setting parameters to satisfy a bunch of conditions}. Let $L$ be an $n$-tuple of linear forms in $\vecz$-variables that defines the following affine projection: All the $\vecy$-variables map to $0$, and $z_u$ maps to $z_u$ for all $u \in [n_0]$. Consider a homogeneous $\Sigma \wedge \Sigma \Pi^{[t]}(s)$ formula $C$ that computes
\begin{equation} \label{eqn:sum of powers of Qi s}
f = Q_1^m + \ldots + Q_s^m,
\end{equation}
where $Q_i = R_i(\vecy,\vecz) + G_i(\vecz)$ such that $R_i, G_i$ are degree-$t$ homogeneous polynomials, every monomial in $R_i$ has a $\vecy$-variable, and $G_i$ is $\vecy$-free for all $i \in [s]$. Clearly, $\pi_L(Q_i) = G_i$. We now construct $R_1, \ldots, R_s$ and $G_1, \ldots, G_s$ so that $C$ satisfies non-degeneracy Conditions \ref{cond1}, \ref{cond3} and \ref{cond4}. \\

\underline{\emph{Constructing $R_1, \ldots, R_s$}}. The number of $\vecz$-monomials of degree-$(t-1)$ is $b := {n_0 + t-2 \choose t-1}$. Let these monomials be $\gamma_1, \ldots, \gamma_b$. Consider a combinatorial design on the $\vecy$-variables, i.e., a system of $s$ subsets of $\vecy$-variables, namely $S_1, \ldots, S_s$, such that for all $i, j \in [s]$
\begin{eqnarray*}
|S_i| &=& kb \quad \quad \quad \text{and} \\
|S_i \cap S_j| &\leq& k-1 \quad \quad \text{if $i \neq j$}.
\end{eqnarray*}
Such a set-system (also known as the Nisan-Wigderson design) exists\footnote{in fact, it can be computed efficiently} if $|\vecy| = n - n_0 \geq (2kb)^2$ and $s \leq (kb)^k$ -- these two conditions are satisfied by Proposition \ref{prop:setting parameters to satisfy a bunch of conditions}. Denote the $kb$ distinct $\vecy$-variables in $S_i$ by $\{y_{ijl} ~:~ j \in [k] \text{ and } l \in [b]\}$ and define
$$R_i := \sum_{j \in [k],~ l \in [b]}{y_{ijl} \cdot \gamma_l}.$$
Let the spaces $U, U_1, \ldots, U_s$ be as in Definition \ref{defn: non-degeneracy}.
\begin{proposition} \label{prop: equality of cond1}
$U = U_1 + \ldots + U_s$ and $U_i = \left \langle \vecz^{k(t-1)} \cdot G_i^{m-k} \right \rangle$ for every $i \in [s]$.
\end{proposition}

\underline{\emph{Constructing $G_1, \ldots, G_s$}}. We set $G_1, \ldots, G_s$ in such a way that $U_1 + \ldots + U_s = U_1 \oplus \ldots \oplus U_s$. Let $p := \lfloor \frac{\sqrt{n_0}}{2} \rfloor$. Consider a combinatorial design on the $\vecz$-variables, i.e., a system of $s$ subsets of $\vecz$-variables, namely $\vecz_1, \ldots, \vecz_s$, such that for all $i, j \in [s]$
\begin{eqnarray*}
|\vecz_i| &=& p \quad \quad \quad \text{and} \\
|\vecz_i \cap \vecz_j| &\leq& \left \lfloor \frac{\min \left(t(m-k), \sqrt{n_0} \right)}{10} \right \rfloor \leq \left \lfloor \frac{p+1}{5} \right \rfloor, \quad \quad \text{if $i \neq j$}.
\end{eqnarray*} 
Such a set-system exists if 
$$s \leq p^{\left \lfloor \frac{\min \left(t(m-k), \sqrt{n_0} \right)}{10} \right \rfloor},$$
which is satisfied by Proposition \ref{prop:setting parameters to satisfy a bunch of conditions}. Let
\begin{equation} \label{eqn:Gi}
G_i := \left(\sum_{z \in \vecz_i}{z} \right)^t, \quad \quad \text{for all $i \in [s]$}.
\end{equation}
\begin{proposition} \label{prop: direct sum of cond1 and cond2}
If $m > 7k$ and $P_1, \ldots, P_s \in \F[\vecz]$ are such that $\deg_{\vecz} P_i \leq 2k(t-1)$ for all $i \in [s]$ and
$$P_1 \cdot G_1^{m-k} + \ldots + P_s \cdot G_s^{m-k} ~=~ 0,$$
then $P_1 = P_2 = \ldots = P_s = 0$.
\end{proposition}
The condition $m > 7k$ is satisfied by Proposition \ref{prop:setting parameters to satisfy a bunch of conditions}. So, Propositions \ref{prop: equality of cond1} and \ref{prop: direct sum of cond1 and cond2} imply that $U = U_1 \oplus \ldots \oplus U_s$ and $\widetilde{U_1} + \ldots + \widetilde{U_s} = \widetilde{U_1} \oplus \ldots \oplus \widetilde{U_s}$. The proof of Proposition \ref{prop: direct sum of cond1 and cond2} (in particular Observation \ref{obs: all monomials present}) also implies that $[G_1^{m-k}]_{z_1=0}, \ldots, [G_s^{m-k}]_{z_1=0}$ are $\F$-linearly independent. \\

In order to satisfy Condition \ref{cond2}, we draw an analogy between the polynomials $g_0 = G_1^e + \ldots + G_s^e$ and $f = Q_1^m + \ldots + Q_s^m$ (Equation \eqref{eqn:sum of powers of Qi s}). We mimic the above construction of $Q_1, \ldots, Q_s$ and $L$ that ensures $U = U_1 \oplus \ldots \oplus U_s$ and $\dim U_i = {n_0 + k(t-1) - 1 \choose k(t-1)}$ to construct $G_1, \ldots, G_s$\footnote{The construction of these $G_1, \ldots, G_s$ should not be confused with the choice of $G_i$ in Equation \eqref{eqn:Gi}. The choices of $Q_i$'s and $G_i$'s till Proposition \ref{prop: direct sum of cond1 and cond2} are used to show that Condition \ref{cond1}, \ref{cond3} and \ref{cond4} of non-degeneracy are satisfied, whereas we choose the $G_i$'s afresh to show that Condition \ref{cond2} of non-degeneracy is satisfied. Finally, a union bound ensures that all the conditions of non-degeneracy are satisfied with high probability.} and $P$ such that $W = W_1 \oplus \ldots \oplus W_s$ and $\dim W_i = {m_0 + k(t-1) - 1 \choose k(t-1)}$. For this, we need to satisfy $n_0 - m_0 \geq (2kc)^2$, $s \leq (kc)^k$ (where $c = {m_0 + t-2 \choose t-1}$), $e > 7k$ and 
$$s \leq \left(\left\lfloor \frac{\sqrt{m_0}}{2} \right\rfloor \right)^{\left \lfloor \frac{\min \left(t(e-k), \sqrt{m_0} \right)}{10} \right \rfloor}.$$
All these relations are taken care of by Proposition \ref{prop:setting parameters to satisfy a bunch of conditions}. 

\subsection{Setting of parameters} \label{sec:parameters}
From the statement of Theorem \ref{thm:learning sums of powers}, we have $n \geq d^2$,~ $t \leq \sqrt{\frac{\log d}{10 \cdot \log \log d}}$,~ $|\F| \geq (ns)^{150 \cdot t}$ and
$$s \leq \min \left(n^{\frac{d}{1100 \cdot t^2}}, \exp(n^{\frac{1}{30 \cdot t^2}}) \right).$$
We leave the proof of the following proposition as an exercise.

\begin{proposition} \label{prop:setting parameters to satisfy a bunch of conditions}
Let $n_0 = \lfloor n^{\frac{1}{3 \cdot t}}\rfloor$, $m_0 = \lfloor {n}^{\frac{1}{15 \cdot t^2}}\rfloor$, $k = \left\lceil \frac{130 \cdot t \cdot \log s}{\log n} \right\rceil$, and (borrowing notations from Section \ref{sec: random formula}) $b = {n_0 + t-2 \choose t-1}$ and $c = {m_0 + t-2 \choose t-1}$. Then, the following relations are satisfied:
\begin{enumerate}
\item \label{rel1} $|\F| \geq (ns)^{150 \cdot t} \gg d \cdot {n + k -1 \choose k}$,
\item \label{rel2} $|\F| \geq (ns)^{150 \cdot t} \gg 3ds \cdot {n_0 + 2kt -1 \choose 2kt}$,
\item \label{rel3} $(nd)^k = \poly(n, s^t)$,
\item \label{rel4} ${n_0 + 2kt -1 \choose 2kt} = \poly(n, s^t)$, 
\item \label{rel5} $n-n_0 \geq (2kb)^2$,
\item \label{rel6} $e = m-k > 7k$,
\item \label{rel7} $n_0 - m_0 \geq (2kc)^2$,
\item \label{rel8} $s \leq (kc)^k$,
\item \label{rel9} $s \leq \left(\left\lfloor \frac{\sqrt{m_0}}{2} \right\rfloor \right)^{\left \lfloor \frac{\min \left(t(e-k), \sqrt{m_0} \right)}{10} \right \rfloor}$.
\end{enumerate}
\end{proposition}

Relations \ref{rel1} and \ref{rel2} are used for applications of the Schwartz-Zippel lemma at various places to ensure that the error probability is bounded by $o(1)$. Relations \ref{rel3} and \ref{rel4} guarantee that the running time of the algorithm is $\poly(n, \sigma, s^t)$. Relations \ref{rel5}-\ref{rel9} are used in Section \ref{sec: random formula} to show that a random homogeneous $\Sigma \wedge \Sigma \Pi^{[t]}$ formula is non-degenerate with high probability.

\section{Moment problem for mixtures of zero-mean Gaussians} \label{sec:gaussians}

In this section, we describe an algorithm for learning the parameters of a mixture of Gaussians in the \emph{non-degenerate} case, given the moments of the mixture \emph{exactly}. Since we can only estimate the moments given samples from the mixture, it is an extremely interesting problem to modify our algorithm to make it work with inexact moments (in the smoothed analysis setting) and we leave it open for future work. Our algorithm also extends naturally to the general mean case but the analysis gets more complicated and we only focus on the zero-mean case for simplicity. \\

Consider a mixture of Gaussians, $\cD = \sum_{i=1}^s w_i \cN(0, \Sigma_i)$, $\sum_{i=1}^s w_i = 1$, $w_i > 0$ for all $i \in [s]$. Let us define quadratic polynomials, $Q_1,\ldots, Q_s$, corresponding to the covariance matrices, by $Q_i(\vecx) = \frac{1}{2} \vecx^T \Sigma_i \vecx$. Also define the polynomials 
$$
f_m = \sum_{i=1}^s w_i Q_i^m,~~~ f_{m-1} = \sum_{i=1}^s w_i Q_i^{m-1}.
$$

Also let us set $n_0 = \lfloor n^{\frac{1}{6}}\rfloor$, $m_0 = \lfloor n^{\frac{1}{60}}\rfloor$, $\ell =  \left\lceil \frac{260 \cdot \log s}{\log n} \right\rceil$, $m = 3 \ell$, $e = 2 \ell$. We will call the mixture $\cD$ \emph{non-degenerate} if $f_m$ and $f_{m-1}$ satisfy the non-degeneracy conditions in Definition \ref{defn: non-degeneracy} with $k = \ell$ and $k = \ell - 1$, respectively. We will need the following elementary proposition about the moment generating function of a Gaussian and mixture of Gaussians.

\begin{proposition}\label{prop:Gaussian_mixture_MGF}
Suppose $Y \sim \cN(\mu, \Sigma)$. Then $\mathbb{E}\left[ e^{\langle \vecx, Y \rangle}\right] = e^{\langle \mu, \vecx \rangle + \frac{1}{2} \vecx^T \Sigma \vecx}$. Similary for a mixture of Gaussians, $\cD = \sum_{i=1}^s w_i \cN(\mu_i, \Sigma_i)$, the moment generating function is $\sum_{i=1}^s w_i e^{\langle \mu_i, \vecx \rangle + \frac{1}{2} \vecx^T \Sigma_i \vecx}$.
\end{proposition}

The next lemma states an efficient algorithm for computing the parameters of a non-degenerate zero-mean Gaussian mixture given access to its exact moments.

\begin{lemma}\label{lem:non-degenerate_Gaussian_mixture}
Let $s \leq \exp\left(n^{\frac{1}{120}}\right)$. There is a randomized $\poly(n,b,s)$ time algorithm ($b$ denotes the total bit complexity of the parameters) that given black-box access to the exact $O(\log(s)/\log(n))$ order moments of a non-degenerate mixture of zero-mean Gaussians $\cD$, recovers its parameters $(w_i, \Sigma_i)_{i=1}^s$.
\end{lemma}

\begin{remark}
Black-box access to the moments means given a vector $\vecx \in \R^n$, access to the moments of the random variable $\langle \vecx, Y \rangle$, where $Y \sim \cD$. In other words, this means black-box access to the polynomials $f_m$ and $f_{m-1}$. Of course, when $s \le \poly(n)$, our algorithm only needs access to $O(1)$ order moments of the mixture $\cD$ in which case black-box access to the moments is immediate because we can compute all the moments explicitly. But we state our algorithm in this general form in hope of applicability in settings where black-box access to the moments might be available without an explicit access to the moments.
\end{remark}

\begin{proof}(Of Lemma \ref{lem:non-degenerate_Gaussian_mixture})

\begin{algorithm}
	\caption{Learning mixtures of zero-mean Gaussians } \label{alg:learning mixture of Gaussians}
	\begin{algorithmic}
		\STATE \textbf{Input}: Black-box access to the \emph{exact} moments of a \emph{non-degenerate} mixture of zero-mean Gaussians $\cD = \sum_{i=1}^s w_i \cN(0, \Sigma_i)$.
		\STATE \textbf{Output}: The parameters of the mixture $(w_i, \Sigma_i)_{i=1}^s$.
	\end{algorithmic}
	\begin{algorithmic}[1]
		\STATEx
		\STATE \label{gaussianalgo:step1} Use black-box access to the moments to get black-box access to the polynomials $f_m$ and $f_{m-1}$ (will be explained in the analysis).
		\STATE \label{gaussianalgo:step2} Use Algorithm \ref{alg:learning sums of powers} to obtain representations $f_m = \sum_{i=1}^s w_i' \left(Q_i'\right)^m$ and $f_{m-1} = \sum_{i=1}^s \widetilde{w_i} \left(\widetilde{Q_i}\right)^{m-1}$ (with $w_i', \widetilde{w}_i$ non-zero): we get $\left(w_i', Q_i' \right)_{i=1}^s$ and $\left(\widetilde{w}_i, \widetilde{Q}_i \right)_{i=1}^s$.
		\STATE \label{gaussianalgo:step3} Find a permutation $\sigma: [s] \rightarrow [s]$ such that $c_i = Q_i'/\widetilde{Q}_{\sigma(i)}$ is a constant (will be explained in the analysis why $\sigma$ is a permutation).
		\STATE \label{gaussianalgo:step4} Let us denote $d_i = \frac{\widetilde{w}_{\sigma(i)}}{w_i' c_i^{m-1}}$. Output $\left( w_i' d_i^m, \frac{1}{d_i} \Sigma_i'\right)_{i=1}^s$, where $\Sigma_i'$ is such that $Q_i'(\vecx) = \frac{1}{2} \vecx^T \Sigma_i' \vecx$.					
	\end{algorithmic}
\end{algorithm}

The algorithm is described in Algorithm \ref{alg:learning mixture of Gaussians}. Suppose $Y \sim \cD$. Then by Proposition \ref{prop:Gaussian_mixture_MGF}, the moment generating function of $\cD$ is given by
$$
\mathbb{E}\left[ e^{\langle \vecx, Y \rangle}\right] = \sum_{i=1}^s w_i e^{\frac{1}{2} \vecx^T \Sigma_i \vecx} = \sum_{i=1}^s w_i e^{Q_i(\vecx)}.
$$
Equating the degree $2m$ part on both sides, we get
$$
\frac{m!}{(2m)!} \cdot \mathbb{E}\left[\langle \vecx, Y \rangle^{2m}\right] = \sum_{i=1}^s w_i Q_i(\vecx)^m.
$$
Thus given black-box access to the moments of $\cD$, we can get black-box access to the polynomials $f_m$ and $f_{m-1}$. This explains Step \ref{gaussianalgo:step1} of the algorithm. Theorem \ref{thm:learning sums of powers} guarantees us that there exist permutations $\pi_1: [s] \rightarrow [s]$ and $\pi_2: [s] \rightarrow [s]$, and constants $c_1',\ldots, c_s'$ and $\widetilde{c}_1,\ldots, \widetilde{c}_s$ (all non-zero) such that
$$
Q_i' = c_i' Q_{\pi_1(i)},~~~ w_{\pi_1(i)} = w_i' (c_i')^m,~~~ \widetilde{Q}_i = \widetilde{c}_i Q_{\pi_2(i)} \text{ and } w_{\pi_2(i)} = \widetilde{w}_i (\widetilde{c}_i)^{m-1}
$$
for all $i \in [s]$. We also have that $\left(Q_i^m\right)_{i=1}^s$ are linearly independent (this is implied by the non-degeneracy condition) and hence the $Q_i's$ span distinct one-dimensional spaces. Thus, there is exactly one $j$ such that $Q_i'/\widetilde{Q}_{j}$ is a constant which is given by $j = \pi_2^{-1}(\pi_1(i))$. This explains Step \ref{gaussianalgo:step3} where $\sigma$ is given by $\pi_2^{-1} \circ \pi_1$. Now,
$$
c_i = \frac{Q_i'}{\widetilde{Q}_{\sigma(i)}} = \frac{c_i' Q_{\pi_1(i)}}{\widetilde{c}_{\sigma(i)} Q_{\pi_2(\sigma(i))}} = \frac{c_i'}{\widetilde{c}_{\sigma(i)}}.
$$
Then,
$$
d_i = \frac{\widetilde{w}_{\sigma(i)}}{w_i' c_i^{m-1}} = \frac{w_{\pi_2(\sigma(i))}}{\widetilde{c}_{\sigma(i)}^{m-1}} \cdot \frac{(c_i')^m}{w_{\pi_1(i)}} \cdot \frac{\widetilde{c}_{\sigma(i)}^{m-1}}{(c_i')^{m-1}} = c_i'.
$$
Thus we output $\left( w_{\pi_1(i)}, \Sigma_{\pi_1(i)} \right)_{i=1}^s$. This completes the proof.
\end{proof}

Next, we instantiate the above lemma with distributional assumptions on the covariance matrices which will satisfy the non-degeneracy condition with high probability.

\begin{corollary}\label{corr:random_Gaussian_mixture}
Let $s \leq \exp\left(n^{\frac{1}{120}}\right)$. There is a randomized $\poly(n,b,s)$ time algorithm ($b$ denotes the total bit complexity of the parameters) that given black-box access to the exact $O(\log(s)/\log(n))$ order moments of a random mixture of zero-mean Gaussians $\cD$, recovers its parameters $(w_i, \Sigma_i)_{i=1}^s$. Random here means that $\Sigma_i = A_i A_i^T$ the entries of $A_i$'s are chosen uniformly and indepdently at random from an arbitrary set $S \subset \Q$ of size $|S| \ge (ns)^{600}$ (and of course $w_i > 0$ for all $i$).
\end{corollary}

\begin{proof}
The non-degeneracy condition is given by non-vanishing of a (non-zero) polynomial \linebreak $p(M_1,\ldots, M_s)$ in the entries of symmetric matrices $M_1,\ldots, M_s$ of degree $D$ at most $(ns)^{300}$ (see Section \ref{sec: random formula}). First of all note that there exist PSD matrices $P_1,\ldots, P_s$ s.t. $p(P_1,\ldots, P_s) \neq 0$. This can be seen by choosing $P_1,\ldots, P_s$ to be symmetric diagonally dominant (SDD) and then applying the Schwarz-Zippel lemma. That is choose the diagonal entries of $P_i$'s from the interval $\{n^2 D^2, \ldots, n^2 D + D^2\}$ and the non-diagonal entries from the interval $\{0, \ldots, D^2\}$ (uniformly and independently). Then $p(P_1,\ldots, P_s) \neq 0$ with non-zero probability (by Schwarz-Zippel) and all the $P_i$'s are SDD and hence PSD.

Now consider the polynomial $q(A_1,\ldots, A_s) = p(A_1 A_1^T, \ldots, A_s A_s^T)$. As argued above, $q$ is a non-zero polynomial and of degree at most $2D$. Hence if we choose the entries of $A_i$'s uniformly and indepdently at random from an arbitrary set $S \subset \Q$ of size $|S| \ge (ns)^{600}$, then $q(A_1,\ldots, A_s) \neq 0$ w.p. $1-o(1)$ and hence $\Sigma_i$'s are non-degenerate w.p. $1 - o(1)$. Now the corollary follows from Lemma \ref{lem:non-degenerate_Gaussian_mixture}.
\end{proof}
\section{Lower bound for homogeneous $\Sigma \Pi \Sigma \Pi^{[t]}$ circuits using $\APP$} \label{sec:lowerbound}
We prove Theorem \ref{thm:lower bound} in this section. The idea is to choose two parameters $k$ and $n_0$ appropriately such that the measure $\APP_{k,n_0}$ of a term of a homogeneous $\Sigma \Pi \Sigma \Pi^{[t]}(s)$ circuit is ``small''. We then construct an explicit polynomial $f_{n,d}$ such that $\APP_{k,n_0}(f_{n,d})$ is ``high'' which leads to a lower bound on $s$. It is the choice of the measure $\APP$ that is novel in this lower bound proof. The missing proofs of the technical statements can be found in Section \ref{sec:proofs from sec lowerbound} of the appendix.

\subsection{High $t$ case} \label{sec:high t}
Let $n,d,t \in \N$ such that $n \geq d^2$ and $\ln \frac{n}{d} \leq t \leq \frac{d}{4\cdot e^{10} \cdot \ln d}$. We set a few parameters as follows:
\begin{itemize}
\item (Order of the derivatives)~ $k = \lfloor \delta \cdot \frac{d}{t} \rfloor$, where $\delta = \frac{1}{4e^{10}}$,
\item (Number of variables after affine projection)~ $n_0 = \lfloor c \cdot k \rfloor$, where $c = \frac{3}{4} \cdot \frac{\ln \frac{n}{k}}{\ln \frac{d}{k}}$. 
\end{itemize}

\begin{observation} \label{obs:parameters}
If the parameters $k, c$ and $n_0$ are chosen as above then $k \geq \lfloor \ln d \rfloor$, $c \geq \frac{3}{2}$ and $n_0 \leq \frac{d}{\ln \ln d}$.
\end{observation}

\begin{observation} \label{obs:trivial upper bound}
Let $f \in \F[\vecx]$ be a homogeneous $n$-variate degree-$d$ polynomial. Then, $$\APP_{k,n_0}(f) \leq {d-k+n_0-1 \choose n_0-1}.$$
\end{observation}

In Section \ref{sec:hard polynomial}, we construct an explicit family of homogeneous, multilinear polynomials $\{f_{n,d,t}\}_{n,d}$ in $\VNP$ such that $\APP_{k,n_0}(f_{n,d,t})$ equals the above upper bound (see Proposition \ref{prop:hard poly lower bound}).  \\

\textbf{Upper bounding the measure for a homogeneous $\Sigma \Pi \Sigma \Pi^{[t]}$ circuit.}~ Let $C$ be a polynomial computed by a homogeneous $\Sigma \Pi \Sigma \Pi^{[t]}(s)$ circuit, i.e.,
\begin{equation} \label{eqn:Sigma Pi Sigma Pi t formula}
C = Q_{11}Q_{12} \cdots Q_{1m_1} ~+~ \ldots ~+~ Q_{s1}Q_{s2} \cdots Q_{sm_s},
\end{equation}
where every $Q_{ij}$ is a homogeneous polynomial of degree at most $t$. By multiplying out factors if necessary, we can assume that all but one of the factors of $T_i = Q_{i1}Q_{i2} \cdots Q_{im_i}$ have degree in $[t,2t]$. So, $m_i \leq m := \lfloor \frac{d}{t} \rfloor + 1$ for all $i \in [s]$. By subadditivity of the measure, we infer the following:

\begin{proposition} \label{prop:circuit upper bound}
$\APP_{k,n_0}(C) \leq s \cdot {m \choose k} \cdot {n_0 + 2kt \choose n_0}$.
\end{proposition}   

Putting Propositions \ref{prop:circuit upper bound} and \ref{prop:hard poly lower bound} together, we get the desired lower bound.

\begin{proposition} \label{prop:lower bound on top fanin high t}
Any homogeneous $\Sigma \Pi \Sigma \Pi^{[t]}(s)$ circuit computing $f_{n,d,t}$ must satisfy
$$s \geq \frac{{d-k+n_0-1 \choose n_0-1}}{{m \choose k} \cdot {n_0 + 2kt \choose n_0}} = \left(\frac{n}{d}\right)^{\Omega\left(\frac{d}{t \ln t}\right)}.$$
\end{proposition}

\textbf{Remark.} Although, in our presentation, $f_{n,d,t}$ depends on $t$, it is easy to get rid of $t$ from the definition of the hard polynomial by using a simple interpolation trick (as in Lemma 14 of \cite{KayalSS14}).

\subsection{Low $t$ case} \label{sec:low t}
Let $n,d,t \in \N$ such that $n \geq d^{20}$ and $1 \leq t \leq \min \left\{ \frac{\ln n}{6e \cdot \ln d},~ d \right\}$. Set the parameters $k,n_0$ as follows:
\begin{itemize}
\item (Order of the derivatives)~ $k = \lceil \delta \cdot \frac{d}{t} \rceil$, where $\delta = \frac{1}{3e}$,
\item (Number of variables after affine projection)~ $n_0 = \lceil n^{\frac{k}{d}} \rceil$.
\end{itemize}
The hard polynomial $f_{n,d,t}$ is defined in Section \ref{sec:hard polynomial}. Propositions \ref{prop:circuit upper bound} and \ref{prop:hard poly lower bound} imply the following:

\begin{proposition} \label{prop:lower bound on top fanin low t}
Any homogeneous $\Sigma \Pi \Sigma \Pi^{[t]}(s)$ circuit computing $f_{n,d,t}$ must satisfy
$$s \geq \frac{{d-k+n_0-1 \choose n_0-1}}{{m \choose k} \cdot {n_0 + 2kt \choose n_0}} = n^{\Omega\left(\frac{d}{t}\right)}.$$
\end{proposition}

\subsection{The hard polynomial} \label{sec:hard polynomial}
Let the parameters $n,d,t,k,n_0$ be as in either Section \ref{sec:high t} or Section \ref{sec:low t}. In this section, we describe the construction of the hard polynomial $f_{n,d,t}$. Let $n_2 := n_0(d-k)$ and $n_1 := n - n_2$. Polynomial $f_{n,d,t}$ is a homogeneous, multilinear polynomial in two sets of variables $\vecy$ and $\vecu$ such that $|\vecy| = n_1$ and $|\vecu| = n_2$. Further, $\vecu = \vecu_1 \uplus \ldots \uplus \vecu_{d-k}$, where each set $\vecu_i$ has $n_0$ variables $\{u_{i,1}, \ldots, u_{i,n_0}\}$. \\

Consider all degree-$(d-k)$ set-multilinear monomials in the $\vecu$-variables with respect to the partition $\vecu = \vecu_1 \uplus \ldots \uplus \vecu_{d-k}$. Such a set-mulilinear monomial $\beta = u_{1,j_1}u_{2,j_2} \cdots u_{d-k,j_{d-k}}$ can be naturally identified with a function 
\begin{eqnarray*}
\phi_{\beta}: [d-k] &\rightarrow& [n_0] \\
							i &\mapsto& j_i.
\end{eqnarray*} 
We say $\phi_{\beta}$ is \emph{non-decreasing} if $\phi_{\beta}(i) \leq \phi_{\beta}(i+1)$ for all $i \in [d-k-1]$. Let $B := \{\beta: \phi_{\beta} \text{ is non-decreasing}\}$ and $\vecz = \{z_1, \ldots, z_{n_0}\}$ be a set of $n_0$ variables. Observe that there is a one-to-one correspondence between monomials in $B$ and $\vecz$-monomials of degree $d-k$ which is given by the projection map
\begin{eqnarray} \label{eqn:projection from u to z}
\pi: \vecu &\rightarrow& \vecz \nonumber \\
	 u_{i,j}	&\mapsto& z_j.
\end{eqnarray}
Hence, $\pi(B) = \vecz^{d-k}$ and $|B| = {d-k + n_0 -1 \choose n_0 -1}$. Order the monomials in $B$ lexicographically and call them $\left(\beta_1, \ldots, \beta_{d-k + n_0 -1 \choose n_0 -1} \right)$. There are ${n_1 \choose k}$ multilinear monomials in $\vecy$-variables of degree $k$.

\begin{proposition} \label{prop:number of y monomials is greater}
${n_1 \choose k} \geq {d-k+n_0-1 \choose n_0-1}$.
\end{proposition}

Order the multilinear degree-$k$ $\vecy$-monomials lexicographically and call the first ${d-k+n_0-1 \choose n_0-1}$ of them $\left(\mu_1, \ldots, \mu_{d-k + n_0 -1 \choose n_0 -1} \right)$. Define
$$f_{n,d,t}(\vecy,\vecu) := \sum_{i \in {d-k + n_0 -1 \choose n_0 -1}} {\mu_i \cdot \beta_i}.$$

It is an easy exercise to show that the family of polynomials defined by $f_{n,d,t}$ is in $\VNP$ as the coefficient of any given monomial in $f_{n,d,t}$ can be computed efficiently.

\begin{proposition} \label{prop:hard poly lower bound}
$\APP_{k,n_0}(f_{n,d,t}) = {d-k+n_0-1 \choose n_0-1}$.
\end{proposition}

\section{Conclusion and open problems} \label{sec:open_problems}

We develop a meta framework for turning lower bounds for arithmetic circuit classes into learning algorithms for the circuits classes in the \emph{non-degenerate} case. A rudimentary form of this framework was first used in \cite{KayalS19} to design learning algorithms for learning homogeneous depth three circuits in the average case. We use the framework to design learning algorithms for sums of powers of low degree polynomials. The problem of learning sums of powers of linear polynomials (aka symmetric tensor decomposition) has been extensively studied in areas across science and many algorithms have been developed for it (again in the non-degenerate case; in the worst case it is NP-hard \cite{Hastad90, Shitov16}). However, even for learning sums of quadratic polynomials, we are not aware of any algorithm in the literature, except for an algorithm implicit in \cite{GeHK15} which works in a limited range of parameters. The problem of learning sums of powers of quadratics has an intimate connection to the well known problem of mixtures of Gaussians (Sections \ref{sec:gaussians_intro} and \ref{sec:gaussians}). We hope that our paper will lead to further algorithms for learning arithmetic circuits and also new connections between learning arithmetic circuits and machine learning problems, which is promising since tensor decomposition (aka learning depth three set-multilinear circuits) has found so many applications in ML. We list some of the interesting open problems below.

\begin{itemize}
\item \textbf{Smoothed analysis of mixtures of (general) Gaussians.} One immediate open problem is to make our algorithm resilient to noise. This is relevant  to mixtures of Gaussians since given samples from the mixture, we can only estimate its moments (upto $1/\poly(n)$ error using $\poly(n)$ samples). We are hopeful that an appropriate modification of our algorithm will lead to polynomial time algorithm for mixtures of general Gaussians in the \emph{smoothed} setting and when the number of components $s \le \poly(n)$.
\item \textbf{Learning other arithmetic circuit classes.} It is natural to implement our framework for other arithmetic circuit classes for which we have lower bounds e.g. set-multilinear circuits, multilinear formulas, regular formulas etc. \cite{Raz09, KayalSS14}.
\item \textbf{More connections between learning arithmetic circuits and ML.} As already mentioned, tensor decomposition finds multiple applications in ML (e.g. see \cite{anandkumar2014tensor}). It is natural to wonder if algorithms for learning more general classes of arithmetic circuits will also find applications in ML. For example, if we had learning algorithms for higher depth set-multilinear circuits (say depth-$4$), can this be utilized to solve problems in ML which tensor decomposition couldn't solve?
\item \textbf{Combining SoS and our techniques.} One of the algorithmic techniques which is very successfully used to design algorithms for tensor decomposition is the Sum of Squares (SoS) method \cite{barak2015dictionary, ge2015decomposing, hopkins2016fast, ma2016polynomial, raghavendra2018high}. Can SoS be also used to design learning algorithms for sums of powers of low degree polynomials (these algorithms might also be more robust to noise)? Perhaps combining SoS with our techniques might help?
\item \textbf{New lower bounds using $\APP$.} Can the method of affine projections of partials, perhaps also combining with shifts, be used to prove new lower bounds? May be for depth-$5$ circuits?
\end{itemize}

\section*{Acknowledgments}
\label{sec:ack}
\addcontentsline{toc}{section}{\nameref{sec:ack}}
We would like to thank Youming Qiao for insightful discussions on simultaneous block-diagonalization of rectangular matrices during the workshop on \emph{Algebraic Methods} held at the Simons Institute for the Theory of Computing in December $2018$. We thank Youming particularly for his suggestion to analyze the adjoint algebra and for referring us to the paper \cite{ChistovIK97}. We thank Navin Goyal for multiple helpful discussions on learning mixtures of Gaussians and related problems and for referring us to the paper \cite{GeHK15}.  We would also like to thank Ravi Kannan for pointing a bug in the statement and proof of Corollary \ref{corr:random_Gaussian_mixture} in an earlier version.

\bibliographystyle{alpha}
\bibliography{references}

%\newpage
\appendix
\section{The adjoint algebra} \label{sec:appendix adjoint}

Let $U$ and $W$ be vector spaces and $\CAL{L}$ a set of linear operators from $U$ to $W$ such that $W = \langle \CAL{L} \circ U \rangle$. Suppose $U$ and $W$ decompose into indecomposable subspaces as:
\begin{equation*}
	U = U_1 \oplus \ldots \oplus U_s  ~~~~\text{and}~~~~  W = W_1 \oplus \ldots \oplus W_s
\end{equation*}
such that $W_i = \langle \CAL{L} \circ U_i \rangle$~ for all $i \in [s]$. In this section, we give a brief overview of the adjoint algebra associated with $\CAL{L}$ and show how analyzing the adjoint provides an avenue to showing uniqueness of decomposition of the above spaces. We will explain this by assuming\footnote{This assumption is without any loss of generality (see Section \ref{sec: space to module decomposition}).} $U = W$ and $U_i = W_i$ for all $i \in [s]$. Let $m = \dim U$. Once a basis of $U$ is fixed, $U$ can be identified with $\F^{m}$ and elements of $\CAL{L}$ are $m \times m$ matrices in $M_{m}(\F)$. Let $\CAL{R} \subseteq M_{m}(\F)$ be the $\F$-algebra generated by\footnote{An $\F$-algebra $\CAL{R}$ has two binary operations $+$ and $\cdot$ defined on its elements such that $(\CAL{R}, +)$ is a $\F$-vector space, $(\CAL{R}, +,\cdot)$ is an associative ring, and for every $a, b \in \F$ and $B, C \in \CAL{R}$ it holds that $(aB)C = B(aC) = a(BC)$. The $\F$-algebra $\CAL{R} \subseteq M_{m}(\F)$ generated by $\CAL{L}\subseteq M_{m}(\F)$ is the set of all finite $\F$-linear sums of finite products of elements of $\CAL{L}$.} $\CAL{L}\cup \{I_{m}\}$, where $I_{m}$ is the $m \times m$ identity matrix. As $U = \langle \CAL{L} \circ U \rangle$ and $U_i = \langle \CAL{L} \circ U_i \rangle$, we have $L\vecu \in U$ and $L\vecu_i \in U_i$ for all $L \in \CAL{L}$, $\vecu \in U$ and $\vecu_i \in U_i$. This gives $U$ an $\CAL{R}$-module\footnote{Let $\CAL{R}$ be an $\F$-algebra with a multiplicative identity $I$. A vector space $U$ is an $\CAL{R}$-module if there is a bilinear map $\circ$ from $\CAL{R} \times U$ to $U$ such that $I \circ \vecu = \vecu$ and $(RS) \circ \vecu = R \circ (S \circ \vecu)$ for all $\vecu \in U$ and $R,S \in \CAL{R}$. In our case, $\circ$ is simply the matrix-vector multiplication operation.} structure and $U_1, \ldots, U_s$ are $\CAL{R}$-submodules of $U$. We say $U_i$ is an \emph{indecomposable} $\CAL{R}$-module if there are no proper $\CAL{R}$-submodules $U_{i1}$ and $U_{i2}$ of $U_i$ such that $U_i = U_{i1} \oplus U_{i2}$. A decomposition of an $\CAL{R}$-module $U$ as
$$U = U_1 \oplus \ldots \oplus U_s, $$
where $U_1, \ldots, U_s$ are indecomposable $\CAL{R}$-submodules of $U$, is \emph{unique} if it is the only possible decomposition of $U$ into indecomposable $\CAL{R}$-submodules (up to reordering of the $U_i$'s).

\subsection{Module homomorphisms} \label{sec:module homomorphisms}
A map $\phi$ from an $\CAL{R}$-module $U$ to another $\CAL{R}$-module $V$ is an $\CAL{R}$-module \emph{homomorphism} from $U$ to $V$ if $\phi(R\vecu + S\vecv) = R \phi(\vecu) + S \phi(\vecv)$ for all $R,S \in \CAL{R}$ and $\vecu, \vecv \in U$. Such a $\phi$ is an $\CAL{R}$-module \emph{isomorphism} from $U$ to $V$ if it is a bijection. An $\CAL{R}$-module homomorphism from $U$ to $U$ is called an $\CAL{R}$-module \emph{endomorphism} of $U$, and an $\CAL{R}$-module isomorphism from $U$ to $U$ is called an $\CAL{R}$-module \emph{automorphism} of $U$. It turns out that the set of $\CAL{R}$-module endomorphisms of $U$ can be computed efficiently as follows: Recall that in our case, $U = \F^{m}$ and $\CAL{R} \subseteq M_{m}(\F)$. Define the \emph{adjoint} of $\CAL{R}$ as
\begin{equation} \label{eqn:adjoint}
\adj(\CAL{R}) := \{D \in M_{m}(\F) ~:~ LD = DL \text{ for all } L \in \CAL{L}\}.
\end{equation} 
Observe that $\adj(\CAL{R})$ is an $\F$-subalgebra of $M_{m}(\F)$.
\begin{proposition} \label{prop:adjoint same as endomorphisms}
The adjoint $\adj(\CAL{R})$ is precisely the set of all $\CAL{R}$-module endomorphism of $U$.
\end{proposition}
\begin{proof}
Let $\phi$ be an $\CAL{R}$-module endomorphism of $U$. As $\CAL{R}$ contains the identity matrix $I_{m}$, $\F \subseteq \CAL{R}$ and so $\phi$ is a linear transformation from $U$ to $U$. Let $D_{\phi} \in M_{m}(\F)$ be the matrix corresponding to $\phi$. Since $\phi(R\vecu) = R \phi(\vecu)$, we have $D_{\phi}R\vecu = RD_{\phi}\vecu$ for all $R \in \CAL{R}$ and $\vecu \in U$. Hence, $D_{\phi}R = RD_{\phi}$ for all $R \in \CAL{R}$ implying $D_{\phi} \in \adj(\CAL{R})$. On the other hand, if $D \in \adj(\CAL{R})$ then the map $\phi_D: U \rightarrow U$ defined as $\phi_D(\vecu) := D\vecu$ satisfies $\phi_D(R\vecu + S\vecv) = R \phi_D(\vecu) + S \phi_D(\vecv)$ for all $R,S \in \CAL{R}$ and $\vecu, \vecv \in U$. So, $\phi_D$ is an $\CAL{R}$-module endomorphism of $U$. 
\end{proof}
A basis of the adjoint can be computed efficiently by solving a system of linear equations arising from the equation $LD = DL$ for all $L \in \CAL{L}$.

\subsection{Module decomposition} \label{sec:module decomposition}
Let $U = \F^m$ be an $\CAL{R}$-module, where $\CAL{R} \subseteq M_m(\F)$. By Proposition \ref{prop:adjoint same as endomorphisms}, the invertible elements of $\adj(\CAL{R})$ are the $\CAL{R}$-module automorphisms of $U$ and these can be used to describe all possible decomposition of $U$ into indecomposable $\CAL{R}$-modules. 

\begin{proposition} \label{prop:all possible decomposition}
\begin{enumerate} [(a)]
\item \label{D gives new decomposition} If $U = U_1 \oplus \ldots \oplus U_s$ is a decomposition of $U$ into indecomposable $\CAL{R}$-submodules and $D\in \adj(\CAL{R})$ is invertible then 
$$U = DU_1 \oplus \ldots \oplus DU_s$$
is another decomposition of $U$ into indecomposable $\CAL{R}$-submodules.
\item \label{new decomposition gives D} If $U = U_1' \oplus \ldots \oplus U_l'$ is any other decomposition of $U$ into indecomposable $\CAL{R}$-submodules then $l=s$ and there is an invertible $D\in \adj(\CAL{R})$ and a permutation $\sigma$ of $[s]$ such that 
$$U_i' = D U_{\sigma(i)} ~~~\text{for all } i \in [s].$$
\end{enumerate}
\end{proposition}
\begin{proof}
The proof of \eqref{D gives new decomposition} follows from the easy observation that $DU_i$ is an $\CAL{R}$-submodule of $U$. \\

To prove \eqref{new decomposition gives D} we will make use of the Krull-Schmidt theorem for modules (\cite{Jacobson89} p. 110, \cite{KSBlog15}).
\begin{theorem} [Krull-Schmidt] \label{thm:Krull-Schmidt}
Let $\CAL{R}$ be an $\F$-algebra and $U$ a finite dimensional vector space that is also an $\CAL{R}$-module. If 
$$U = U_1 \oplus \ldots \oplus U_s ~~~\text{and}~~~ U = U_1' \oplus \ldots \oplus U_l'$$
are two decomposition of $U$ into indecomposable $\CAL{R}$-submodules then $l=s$ and there is a permutation $\sigma$ of $s$ such that $U_i'$ and $U_{\sigma(i)}$ are isomorphic as $\CAL{R}$-modules for all $i \in [s]$. 
\end{theorem}
The theorem holds for any module that is both Noetherian and Artinian -- a finite dimensional module is trivially Noetherian and Artinian. Applying the Krull-Schmidt theorem to our setting, we get $l = s$ and that there is a permutation $\sigma$ of $[s]$ such that $U_i' \cong U_{\sigma(i)}$ as $\CAL{R}$-modules for all $i \in [s]$. Let these ismorphisms be $\phi_1, \ldots, \phi_s$, i.e.,
$$U_i' = \phi_i(U_{\sigma(i)}) ~~~\text{for all } i\in [s].$$
Define a map $\phi$ from $U$ to $U$ as follows: Let $\vecu\in U$. If $\vecu = \vecu_1 + \ldots + \vecu_s$, where $\vecu_i \in U_i$, then $\phi(\vecu) := \phi_1(\vecu_{\sigma(1)}) + \ldots + \phi_s(\vecu_{\sigma(s)})$. Observe that $\phi$ restricted to $U_{\sigma(i)}$ is just $\phi_i$. It is easy to verify that $\phi$ is an $\CAL{R}$-module automorphism of $U$. Hence, by Proposition \ref{prop:adjoint same as endomorphisms}, there is an invertible $D \in \adj(\CAL{R})$ such that $\phi(\vecu) = D \vecu$ and so $U_i' = D U_{\sigma(i)}$ for all $i \in [s]$.  
\end{proof}
Proposition \ref{prop:all possible decomposition} implies that the invertible elements of $\adj(\CAL{R})$ exactly capture the various possible decompositions of $U$ into indecomposable $\CAL{R}$-modules. So, analyzing the adjoint becomes vital in showing uniqueness of a module decomposition.

\subsection{Uniqueness of decomposition} \label{sec:uniqueness of decomposition}
It turns out that showing uniqueness of module decomposition is essentially equivalent to showing that the elements of the adjoint are simultaneously block-diagonalizable. As before, let $U_1 \oplus \ldots \oplus U_s$ be a decomposition of the $\CAL{R}$-module $U=\F^m$ into indecomposable $\CAL{R}$-submodules. For simplicity, assume $\dim U_i = r$ for all $i \in [s]$. Let $(\vecu_{i1}, \ldots, \vecu_{ir})$ be a basis of $U_i$, and $A \in \GL_m(\F)$ be the basis change matrix from the standard basis of $\F^m$ to $(\vecu_{11}, \vecu_{12}, \ldots, \vecu_{1r},~ \ldots~, \vecu_{s1}, \vecu_{s2}, \ldots, \vecu_{sr})$. 

\begin{proposition} \label{prop:uniqueness of decomposition}
\begin{enumerate} [(a)]
\item \label{uniqueness implies block-diagonalization} If $U = U_1 \oplus \ldots \oplus U_s$ is the unique decomposition of $U$ into indecomposable $\CAL{R}$-submodules, and $U_i \ncong U_j$ as $\CAL{R}$-modules for $i \neq j$, then $A \cdot \adj(\CAL{R}) \cdot A^{-1}$ consists of block-diagonal matrices (with block size $r$). 
\item \label{block-diagonalization implies uniqueness} If $A \cdot \adj(\CAL{R}) \cdot A^{-1}$ consists of block-diagonal matrices (with block size $r$) then $U = U_1 \oplus \ldots \oplus U_s$ is the unique decomposition of $U$ into indecomposable $\CAL{R}$-submodules.
\end{enumerate} 
\end{proposition}
\begin{proof}
Let $D \in \adj(\CAL{R})$ be invertible. By Proposition \ref{prop:all possible decomposition} (\ref{D gives new decomposition}), 
$$U = DU_1 \oplus \ldots \oplus DU_s$$
is another decomposition of $U$ into indecomposable $\CAL{R}$-submodules. If $U = U_1 \oplus \ldots \oplus U_s$ is the unique decomposition and $U_i \ncong U_j$ as $\CAL{R}$-modules for $i \neq j$, then $DU_i = U_i$ for all $i \in [s]$. In other words, $A \cdot D \cdot A^{-1}$ is block-diagonal for every invertible $D \in \adj(\CAL{R})$. Now, a simple application of the Schwartz-Zippel lemma implies $A \cdot \adj(\CAL{R}) \cdot A^{-1}$ consists of block-diagonal matrices if $|\F| > 3sr$. This completes the proof of part (\ref{uniqueness implies block-diagonalization}). \\

Suppose $U = U_1' \oplus \ldots \oplus U_s'$ be another decomposition of $U$ into indecomposable $\CAL{R}$-submodules. By Proposition \ref{prop:all possible decomposition} (\ref{new decomposition gives D}), there is an invertible $D \in \adj(\CAL{R})$ and a permutation $\sigma$ of $[s]$ such that 
$$U_i' = DU_{\sigma(i)} ~~~\text{for all } i \in [s].$$
If $A \cdot \adj(\CAL{R}) \cdot A^{-1}$ consists of only block-diagonal matrices then
$$D U_{\sigma(i)} \subseteq U_{\sigma(i)}, ~~~\text{implying}~~~ U_i' \subseteq U_{\sigma(i)} ~~~\text{for all } i \in [s].$$
This further implies $U_i' = U_{\sigma(i)}$ for all $i \in [s]$ as $\sum_{i \in [s]}{\dim U_i'} = \sum_{i \in [s]}{\dim U_{\sigma(i)}}$. Thus, the decomposition $U = U_1 \oplus \ldots \oplus U_s$ is unique.   
\end{proof}

\subsubsection*{The adjoint algebra arising in our case}
As briefed in Section \ref{sec:learning from lower bound}, our learning problem is essentially reduced to the following module decomposition problem: We are given a basis of an appropriate $\CAL{R}$-module $U$ that decomposes as 
\begin{equation} \label{eqn:our module decomposition} 
U = U_1 \oplus \ldots \oplus U_s,
\end{equation}
where $U_i$ is an $\CAL{R}$-submodule of $U$ that is not guaranteed to be indecomposable and $\dim U_i = r$ for all $i \in [s]$. We are required to 
\begin{itemize}
\item show that each $U_i$ is an indecomposable $\CAL{R}$-module,
\item show that the above decompostion is unique,
\item find the decomposition, i.e., compute bases of $U_1, \ldots U_s$.
\end{itemize} 
Here, $\CAL{R}$ is the $\F$-algebra generated by a set of linear operators $\CAL{L}$ on $U$. Guided by Proposition \ref{prop:uniqueness of decomposition}, we analyze the adjoint $\adj(\CAL{R})$. It turns out that the ``richness'' of the carefully chosen set of linear operators $\CAL{L}$ implies that 
$$A \cdot \adj(\CAL{R}) \cdot A^{-1} ~=~ \CAL{D} ~:=~ \{\diag(a_1, \ldots, a_s) \otimes I_r ~:~ a_i \in \F \text{ for all } i \in [s]\},$$
where $A$ is as defined at the beginning of this section. In other words, elements of the adjoint are simultaneously diagonalizable. The following is an easy corollary of Proposition \ref{prop:uniqueness of decomposition}.

\begin{corollary} \label{cor:diagonalizable adjoint}
If $A \cdot \adj(\CAL{R}) \cdot A^{-1} = \CAL{D}$ then the $\CAL{R}$-modules $U_1, \ldots, U_s$ (in Equation \eqref{eqn:our module decomposition}) are indecomposable and $U = U_1 \oplus \ldots \oplus U_s$ is the unique decomposition of $U$ into indecomposable $\CAL{R}$-submodules.
\end{corollary}

Finally, we find the decomposition by simultaneously diagonalizing the basis elements of $\adj(\CAL{R})$. 
\section{Reducing vector space decomposition to module decomposition} \label{sec: space to module decomposition}

In this section, we reduce the vector space decomposition problem to the module decomposition problem. In fact, our reduction works for a more general problem, which we call generalized vector space decomposition. We describe this setting below.\\

Suppose we have a directed graph $G = (V,E)$. At each vertex $v \in V$, we have a vector space $U_v$ and each edge $(v,w) \in E$ carriers a set of linear maps $\cL_{v,w}$ from $U_v$ to $U_w$. A vector space decomposition of the collection of vector spaces $(U_v)_{v \in V}$ is a collection of decompositions
$$
U_v = U_{v,1} \oplus \cdots \oplus U_{v,s}
$$
such that $ \langle \cL_{v,w} \circ U_{v,i} \rangle \subseteq U_{w,i}$ for all $i \in [s]$ and $(v,w) \in E$. The collection of decompositions is indecomposable if there are no \emph{finer} decompositions, i.e., there are no proper subspaces $U'_{v, i}, U''_{v, i}$ of $U_{v, i}$ (and $U'_{w, i}, U''_{w, i}$ of $U_{w, i}$) such that $U_{v, i} = U'_{v, i} \oplus U''_{v, i}$ (and $U_{w, i} = U'_{w, i} \oplus U''_{w, i}$), and $\langle \cL_{v,w} \circ U'_{v,i} \rangle \subseteq U'_{w,i}$, $ \langle \cL_{v,w} \circ U''_{v,i} \rangle \subseteq U''_{w,i}$ for all $i \in [s]$ and $(v,w) \in E$. The generalized vector space decomposition problem is the task of computing a collection of indecomposable decompositions of the spaces $(U_v)_{v \in V}$ from the graph $G$. %The collection is indecomposable if no such decomposition exists (into $s \ge 2$ parts). 
Note that the module isomorphism problem corresponds to a single loop on one vertex and the vector space decomposition problem corresponds to two vertices and a single edge between them.\\

There is a simple reduction from the generalized vector space decomposition problem to the module decomposition problem. Given the above instance, we consider the vector space $U = \oplus_{v \in V} U_v$. We define some special linear maps from $U$ to $U$ which will be central to the reduction. We note that we just need to describe the behaviour of the linear maps on each of the $U_v$'s, as we can extend the maps linearly to the whole space $U$. The first set of linear maps are projections onto $U_v$'s.
\[   
\Pi_{v}(u) = 
     \begin{cases}
      u &\text{~~if~~} u \in U_v\\
      0  &\text{~~if~~} u \in U_{v'} \text{~~for~~} v' \neq v.
     \end{cases}
\]
That is, $\Pi_v$ is the projector onto $U_v$. The second set of linear maps are the natural extensions of $\cL_{v,w}$'s to the whole space. Given $L \in \cL_{v,w}$, we define the extension of $L$ as
\[   
\ext(L)(u) = 
     \begin{cases}
      L(u) &\text{~~if~~} u \in U_v\\
      0  &\text{~~if~~} u \in U_{v'} \text{~~for~~} v' \neq v.
     \end{cases}
\]
Then, we can define $\widetilde{\cL}_{v,w} = \{\ext(L): L \in \cL_{v,w}\}$. Let $\cR$ be the algebra generated by $\{I_m\} \cup \{\Pi_v\}_{v \in V} \cup \{\widetilde{\cL}_{v,w}\}_{(v,w) \in E}$, where $m := \dim(U)$. Observe that $U$ can be naturally treated as an $\cR$-module.
Now, we have the following elementary proposition which characterizes $\cR$-submodules of $U$ (i.e., subspaces of $U$ that are invariant with respect to $\cR$). %invariant subspaces of $\cR$.

\begin{proposition}\label{prop:invariant_reduction}
A subspace $U' \subseteq U$ is an $\cR$-submodule of $U$ (i.e., $\langle \cR \circ U' \rangle \subseteq U'$) if and only if it is of the form $\oplus_{v \in V} U_v'$ such that $U_v' \subseteq U_v$ for all $v \in V$ and $\langle \cL_{v,w} \circ U_v' \rangle \subseteq U_w'$ for all $(v,w) \in E$.
\end{proposition}

\begin{proof}
One direction is clear. If $U'$ is of the form $\oplus_{v \in V} U_v'$ such that $U_v' \subseteq U_v$ for all $v \in V$ and $\langle \cL_{v,w} \circ U_v' \rangle \subseteq U_w'$ for all $(v,w) \in E$, then $U'$ is an $\cR$-submodule of $U$. In the other direction, suppose $U'$ is an $\cR$-submodule of $U$. Let $U_v' = \Pi_v \circ U'$. Then, $U_v' \subseteq U'$ for all $v \in V$ since $U'$ is an $\cR$-module. On the other hand, $U' \subseteq \oplus_{v \in V} U_v'$, hence $U' = \oplus_{v \in V} U_v'$. Another consequence of $U'$ being an $\cR$-module is that 
$$
\langle \widetilde{\cL}_{v,w} \circ U_v'\rangle  = \langle \widetilde{\cL}_{v,w} \circ \Pi_v \circ U'\rangle  \subseteq U'
$$
Since every map in $\widetilde{\cL}_{v,w}$ maps $U$ to $U_w$, we have that,
$$
\langle \widetilde{\cL}_{v,w} \circ U_v'\rangle  \subseteq U' \cap U_w = U_w'
$$
which is the same as $\langle \cL_{v,w} \circ U_v'\rangle  \subseteq  U_w'$. This completes the proof.
%U_v' := \{u \in U_v: \exists \{u_{v'} \in U_{v'}\}_{v' \neq v} \: \text{s.t.} \:  u + \sum_{v' \neq v} u_{v'} \in U'\}

\end{proof}

This yields the following corollary which characterizes decomposition of $U$ into $\cR$-submodules.

\begin{corollary}[Reduction to module decomposition] \label{corr:decomp_reduction} 
$U = U_1 \oplus \cdots \oplus U_s$ is a decomposition of $U$ into $\cR$-submodules if and only if each $U_i$ is of the form $\oplus_{v \in V} U_{v,i}$ such that $U_v = U_{v,1} \oplus \cdots \oplus U_{v,s}$ for all $v \in V$ and $\langle \cL_{v,w} \circ U_{v,i} \rangle \subseteq U_{w,i}$ for all $i \in [s]$, $(v,w) \in E$.
\end{corollary}

\begin{proof}
Again one direction is clear. In the other direction, suppose $U = U_1 \oplus \cdots \oplus U_s$ is a decomposition of $U$ into $\cR$-submodules. This means that each of the $U_i$'s is an $\cR$-submodule of $U$. By Proposition \ref{prop:invariant_reduction}, each $U_i$ is of the form $\oplus_{v \in V} U_{v,i}$ such that $U_{v,i} \subseteq U_v$ for all $i \in [s], v \in V$ and $\langle \cL_{v,w} \circ U_{v,i} \rangle \subseteq U_{w,i}$ for all $i \in [s]$, $(v,w) \in E$. Now $U = U_1 \oplus \cdots \oplus U_s$ and $U_{v, i} \subseteq U_v, U_i$, hence $U_{v,1},\ldots, U_{v,s}$ form a direct sum and
$$
 U_{v,1} \oplus \cdots \oplus U_{v,s} \subseteq U_v
$$
for all $v \in V$. What remains to prove is that
$$
U_v = U_{v,1} \oplus \cdots \oplus U_{v,s} 
$$
for all $v \in V$. Suppose there is some $w \in V$ such that 
$$
U_{w,1} \oplus \cdots \oplus U_{w,s} \subsetneq U_w.
$$
Then
$$
U = \oplus_{i \in [s]} U_i = \oplus_{i \in [s]} \oplus_{v \in V} U_{v,i} = \oplus_{v \in V} \oplus_{i \in [s]} U_{v,i} ~\subsetneq~ \oplus_{v \in V} U_v = U
$$
which is a contradiction. Hence,
$$
U_v = U_{v,1} \oplus \cdots \oplus U_{v,s} 
$$
for all $v \in V$. This completes the proof.
\end{proof}

The Krull-Schmidt theorem for module decomposition and the above reduction allows one to obtain a uniqueness theorem for generalized vector space decomposition.

\begin{theorem}[Generalized vector space decomposition: uniqueness]
Suppose $U_v = U_{v,1} \oplus \cdots \oplus U_{v,s}$ and $U_v = U'_{v,1} \oplus \cdots \oplus U'_{v,s'}$ are two collection of decompositions (which are further indecomposable) for the generalized vector space decomposition problem. Then $s = s'$. Furthermore, there exist linear maps $L_v: U_v \rightarrow U_v$ and a permutation $\sigma: [s] \rightarrow [s]$ such that $U'_{v,i} = L_v \circ U_{v,\sigma(i)}$ for all $i \in [s]$ and $v \in V$. Also, $L_{w} \circ L_{v,w} = L_{v,w} \circ L_v$ for all $L_{v,w} \in \cL_{v,w}$ and $v,w \in V$.
\end{theorem}

\begin{proof}
We look at the reduction to module decomposition, the algebra $\cR$ discussed above and the vector space $U = \oplus_{v \in V} U_v$. Let us define $U_i = \oplus_{v \in V} U_{v,i}$ and $U'_j = \oplus_{v \in V} U_{v,j}'$. Then $U = U_1 \oplus \cdots \oplus U_s$ and $U = U_1' \oplus \cdots \oplus U_{s'}'$ are two decompositions of $U$ into indecomposable $\cR$-submodules (because of Corollary \ref{corr:decomp_reduction}). Hence by Theorem \ref{thm:Krull-Schmidt}, $s = s'$, and there exist a permutation $\sigma: [s] \rightarrow [s]$ and a linear map $L: U \rightarrow U$  such that $U_i' = L \circ U_{\sigma(i)}$ and $L \circ R = R \circ L$ for every $R \in \cR$.
\\
\\
Now for every $v \in V$, $L \circ \Pi_v = \Pi_v \circ L$ which implies that $L \circ U_v \subseteq U_v$. We can call the restriction of $L$ to $U_v$ as the map $L_v: U_v \rightarrow U_v$. Now take an operator $L_{v,w} \in \cL_{v,w}$ and $\ext(M) \in \widetilde{\cL}_{v,w}$. We have that $L \circ \ext(L_{v,w}) = \ext(L_{v,w}) \circ L$. This implies that $L_w \circ L_{v,w} = L_{v,w} \circ L_v$ for all $L_{v,w} \in \cL_{v,w}$. Also,
$$
U_{v,i}' = \Pi_v \circ U_i' = \Pi_v \circ L \circ U_{\sigma(i)} = L \circ \Pi_v \circ U_{\sigma(i)} = L \circ U_{v, \sigma(i)}.
$$
This completes the proof.
\end{proof}

As a corollary, we get a uniqueness theorem for vector space decomposition.

\begin{corollary}[Vector space decomposition: uniqueness]\label{corrLuniqueness_vector_space_decomp}
Suppose $\cL$ is a set of linear maps between vector spaces $U$ and $W$. Suppose $U = U_1 \oplus \cdots \oplus U_s, W = W_1 \oplus \cdots \oplus W_s$ and $U = U_1' \oplus \cdots \oplus U_{s'}', W = W_1' \oplus \cdots \oplus W_{s'}'$ are two indecomposable decompositions with respect to $\cL$. Then $s = s'$. Furthermore, there exist linear maps $D: U \rightarrow U$ and $E: W \rightarrow W$ and a permutation $\sigma: [s] \rightarrow [s]$ such that $U_i' = D \circ U_{\sigma(i)}$, $W_i' = E \circ W_{\sigma(i)}$ for all $i \in [s]$ and $E \circ L = L \circ D$ for all $L \in \cL$.
\end{corollary}

We also mention that via the above reduction, we get a polynomial time algorithm for generalized vector space decomposition (over finite fields, reals and complex numbers) using the polynomial time algorithm for module decomposition in \cite{ChistovIK97}. However, we do not use this algorithm for our learning problem since we also want the algorithm to work over rationals which is possible to do in our setting with a simpler specialized algorithm.
\section{Why doesn't the shifted partials measure work?} \label{sec: limitation of SP}
In this section, we explain why the shifted partials measure (as it is) is unlikely to satisfy the basic non-degeneracy condition given by Equation \eqref{eqn:direct_sum} in Section \ref{sec:intro}, if $n \geq d^2$. The shifted partials measure ($\mathsf{SP}$), introduced in \cite{Kayal12eccc}, is defined as follows: Let $f \in \F[\vecx]$ be an $n$-variate degree-$d$ homogeneous polynomial and $k, \ell \in \N$. Then,
$$\mathsf{SP}_{k, \ell}(f) := \dim \left \langle \vecx^{\ell} \cdot \partial^k_{\vecx}~ f\right \rangle.$$ 
Clearly, $\mathsf{SP}_{k, \ell}(f)$ is upper bounded by $\min \left({n+k-1 \choose k} \cdot {n+\ell-1 \choose \ell}, {n+d-k+\ell-1 \choose d-k+\ell} \right)$. Suppose
\begin{equation*}
f = c_1Q_1^m + \ldots + c_s Q_s^m,
\end{equation*}
where each $c_i\in \F^{\times}$, $Q_i$ is a homogeneous polynomial of degree $t$, and $tm = d$. Let $U(f) := \left \langle \vecx^{\ell} \cdot \partial^k_{\vecx}~ f\right \rangle$. We wish to satisfy the main non-degeneracy condition
\begin{equation} \label{eqn: main non-degeneracy again} 
U(f) = U(Q_1^m) \oplus \ldots \oplus U(Q_s^m),  
\end{equation}
for random $Q_1, \ldots, Q_s$. This imposes the restriction $\ell < d$, as otherwise $U(Q_i^m) \cap U(Q_j^m) \neq \{0\}$ for $i \neq j$. On the other hand, $\mathsf{SP}_{k, \ell}(Q_i^m)$ is upper bounded by ${n + k(t-1) + \ell -1 \choose k(t-1) + \ell }$. If
$$s \cdot {n + k(t-1) + \ell -1 \choose k(t-1) + \ell} \leq \min \left({n+k-1 \choose k} \cdot {n+\ell-1 \choose \ell}, {n+d-k+\ell-1 \choose d-k+\ell} \right),$$
then we may be able to satisfy the direct sum given by Equation \eqref{eqn: main non-degeneracy again}. For this, we need $kt \leq d$. But, with both $\ell$ and $kt$ upper bounded by $d$, ${n + k(t-1) + \ell -1 \choose k(t-1) + \ell}$ cannot be less that ${n+k-1 \choose k} \cdot {n+\ell-1 \choose \ell}$ with growing $t$, if $n \geq d^2$. \\

Thus, it seems difficult to satisfy the direct sum condition using the shifted partials measure if $n \geq d^2$. However, if $n$ is much smaller than $d$ then it may be possible to achieve the same. This is what spurred us to think in the direction of reducing the number of variables to below $d$ using affine projections. Indeed, we have shown in this work that such affine projections do work (for both lower bound and learning) even \emph{without shifts by monomials}. But, shifts may play a crucial role if $n$ is much smaller than $d$ to begin with (say, if $n$ is a constant), in which case doing affine projections does not seem to help.  
\section{Proofs from Section \ref{sec:learning}} \label{sec:proofs from sec learning}

\subsection*{Proof of Observation \ref{obs: random L gives U a direct sum structure}}
As $C$ is non-degenerate, Condition \ref{cond1} of Definition \ref{defn: non-degeneracy} implies, 
$$\APP_{k,n_0}(f) = s \cdot {n_0 + k(t-1) -1 \choose k(t-1)}.$$
By Proposition \ref{prop:setting parameters to satisfy a bunch of conditions}, $|\F| \gg (d-k) \cdot {n+k \choose k}$. Arguing as in Observation \ref{obs:random L}, with probability $1 - o(1)$,
$$\APP_{k,n_0}(f) = \dim U = s \cdot {n_0 + k(t-1) -1 \choose k(t-1)},$$
which implies $U = U_1 \oplus \ldots \oplus U_s$ and $\dim U_i = {n_0 + k(t-1) -1 \choose k(t-1)}$ for all $i \in [s]$.

\subsection*{Proof of Observation \ref{obs: random L gives tilde U a direct sum structure}}
As $C$ is non-degenerate, Condition \ref{cond3} of Definition \ref{defn: non-degeneracy} implies that there exists an $L$ such that 
$$\left \langle \vecz^{2k(t-1)} \cdot \pi_L(Q_1)^{e} \right \rangle + \ldots + \left \langle \vecz^{2k(t-1)} \cdot \pi_L(Q_s)^{e} \right \rangle = \left \langle \vecz^{2k(t-1)} \cdot \pi_L(Q_1)^{e} \right \rangle \oplus \ldots \oplus \left \langle \vecz^{2k(t-1)} \cdot \pi_L(Q_s)^{e} \right \rangle.$$
For any tuple of $n$ linear forms $L$, the degree of a polynomial in $\vecz^{2k(t-1)} \cdot \pi_L(Q_i)^{e}$ is at most $2d$ as $kt \leq d$ (by Proposition \ref{prop:setting parameters to satisfy a bunch of conditions}). If $L$ is a tuple of random linear forms then the polynomials in the set
$$\vecz^{2k(t-1)} \cdot \pi_L(Q_1)^{e} ~\bigcup~ \ldots ~\bigcup~ \vecz^{2k(t-1)} \cdot \pi_L(Q_s)^{e}$$
are $\F$-linearly independent with probability $1 - o(1)$ if $|\F| \gg 2ds \cdot {n_0 + 2k(t-1) -1 \choose 2k(t-1)}$, which is ensured by Proposition \ref{prop:setting parameters to satisfy a bunch of conditions}. 

\subsection*{Proof of Proposition \ref{prop: multi-gcd}}
Recall that $U_i = \left \langle \vecz^{k(t-1)} \cdot G_i^e \right \rangle$. If $g \in V = \left\langle G_1^{e}, \ldots, G_s^{e}\right\rangle$ then $z_1^{k(t-1)} \cdot g$ and $z_2^{k(t-1)} \cdot g$ belong to $U_1 + \ldots + U_s = U$. Hence, there are $a_1, \ldots a_{sr}, b_1, \ldots, b_{sr} \in \F$ such that
\begin{eqnarray*}
a_1f_1 + \ldots + a_{sr}f_{sr} &=& z_1^{k(t-1)} \cdot g ~~~~\text{and}\\
b_1f_1 + \ldots + b_{sr}f_{sr} &=& z_2^{k(t-1)} \cdot g.
\end{eqnarray*} 
On the other hand, suppose that there exist $a_1, \ldots a_{sr}, b_1, \ldots, b_{sr} \in \F$ such that 
\begin{equation} \label{eqn:ratio}
\frac{a_1f_1 + \ldots + a_{sr}f_{sr}}{z_1^{k(t-1)}} = \frac{b_1f_1 + \ldots + b_{sr}f_{sr}}{z_2^{k(t-1)}}.
\end{equation}
As $f_1, \ldots, f_{sr}$ is a basis of $U$, there are polynomials $P_1, \ldots, P_s, P'_1, \ldots, P'_s \in \left \langle \vecz^{k(t-1)} \right \rangle$ that satisfy 
\begin{eqnarray*}
a_1f_1 + \ldots + a_{sr}f_{sr} &=& P_1G_1^e + \ldots + P_sG_s^e ~~~~\text{and} \\
b_1f_1 + \ldots + b_{sr}f_{sr} &=& P'_1G_1^e + \ldots + P'_sG_s^e.
\end{eqnarray*} 
From Equation \eqref{eqn:ratio}, we have
$$\left(z_2^{k(t-1)}P_1 - z_1^{k(t-1)}P'_1 \right) \cdot G_1^e + \ldots + \left(z_2^{k(t-1)}P_s - z_1^{k(t-1)}P'_s \right) \cdot G_s^e = 0.$$
As $\left(z_2^{k(t-1)}P_i - z_1^{k(t-1)}P'_i \right) \in \left \langle \vecz^{2k(t-1)} \right \rangle$, by Observation \ref{obs: random L gives tilde U a direct sum structure}, $z_2^{k(t-1)}P_i - z_1^{k(t-1)}P'_i = 0$ for all $i \in [s]$. Hence, $z_1^{k(t-1)}$ divides $P_i$ and $z_2^{k(t-1)}$ divides $P'_i$ for all $i \in [s]$. But, $\deg(P_i) = \deg(P'_i) = k(t-1)$. Therefore, there are $\hat{a}_1, \ldots, \hat{a}_s \in \F$ such that
$$\frac{a_1f_1 + \ldots + a_{sr}f_{sr}}{z_1^{k(t-1)}} = \frac{b_1f_1 + \ldots + b_{sr}f_{sr}}{z_2^{k(t-1)}} = \hat{a}_1G_1^e + \ldots + \hat{a}_s G_s^e =: g(\vecz) \in V.$$

\subsection*{Proof of Proposition \ref{prop: decomposition of W}}
As $C$ is non-degenerate, Condition \ref{cond2} of Definition \ref{defn: non-degeneracy} implies that there exist $L$ and $P$ such that
\begin{eqnarray} \label{eqn: non-degeneracy condition 2} 
\app{P}{\vecz}{k}{(G_1^{e} + \ldots + G_s^{e})} &=& \app{P}{\vecz}{k}{G_1^{e}} \oplus \ldots \oplus \app{P}{\vecz}{k}{G_s^{e}}, ~~~~\text{and} \nonumber \\
\dim \app{P}{\vecz}{k}{G_i^{e}} &=&  {m_0 + k(t-1) -1 \choose k(t-1)} ~~\text{ for all }~~ i \in [s],
\end{eqnarray}
where $G_i = \pi_L(Q_i)$ and $e = m-k$. If $L$ and $P$ are tuples of random linear forms (as in Step \ref{mainalgo:step1} and \ref{mainalgo:step4} of Algorithm \ref{alg:learning sums of powers}) then the above equation holds with probability $1 - o(1)$ provided $|\F| \gg 2ds \cdot {m_0 + k(t-1) -1 \choose k(t-1)}$ (which is ensured by Proposition \ref{prop:setting parameters to satisfy a bunch of conditions}). Let $g_0 = G_1^{e} + \ldots + G_s^{e}$. By Equation \eqref{eqn: non-degeneracy condition 2},
\begin{equation} \label{eqn: non-degeneracy condition 2 applied to our setting}
\app{P}{\vecz}{k}{g_0} = W_1 \oplus \ldots \oplus W_s ~\text{ and }~ \dim W_i = {m_0 + k(t-1) -1 \choose k(t-1)} ~\text{ for all }~ i \in [s],
\end{equation}
where $W_i = \app{P}{\vecz}{k}{G_i^e}$. Recall that $W = \app{P}{\vecz}{k}{g}$, where $g$ is a random element of $V$. As $G_1^{e}, \ldots, G_s^{e}$ is a basis of $V$, we have $g = b_1G_1^{e} + \ldots + b_sG_s^{e}$, where $b_1, \ldots, b_s \in_r \F$. The next claim completes the proof of the proposition.

\begin{claim}
If $g = b_1G_1^{e} + \ldots + b_sG_s^{e}$ such that $b_1, \ldots, b_s \in_r \F^{\times}$, then $\app{P}{\vecz}{k}{g} = \app{P}{\vecz}{k}{g_0}$ with probability $1 - o(1)$.
\end{claim} 
\begin{proof}
With every polynomial $\hat{g} \in V$, associate a ${n_0 + k -1 \choose k} \times {m_0 + et - k - 1 \choose et-k}$ matrix $M(\hat{g})$ as follows: The rows of $M(\hat{g})$ are indexed by all monomials in $\vecz$-variables of degree $k$ and the columns are indexed by all monomials in $\vecw$-variables of degree $et-k$. If $\alpha$ is a $\vecz$-monomial of degree $k$ and $\beta$ is a $\vecw$-monomial of degree $(et-k)$ then the $(\alpha,\beta)$-th entry of $M(\hat{g})$ is the coefficient of $\beta$ in $\pi_P\left(\frac{\partial^k \hat{g}}{\partial \alpha}\right)$. In other words, $M(\hat{g})$ is the \emph{coefficient matrix} consisting of the coefficients of the polynomials in $\pi_{P}(\der{\vecz}{k}{\hat{g}})$. Let $q = {m_0 + k(t-1) -1 \choose k(t-1)}$. Clearly, for every $\hat{g} \in V$,
\begin{eqnarray*}
\app{P}{\vecz}{k}{\hat{g}} &\subseteq& W_1 \oplus \ldots \oplus W_s = \app{P}{\vecz}{k}{g_0}, ~~~~~~\text{by Equation \eqref{eqn: non-degeneracy condition 2 applied to our setting}} \\
\Rightarrow \rank~[M(\hat{g})] &\leq& \rank~[M(g_0)] = s \cdot {m_0 + k(t-1) -1 \choose k(t-1)} = sq. 
\end{eqnarray*}   
There exist a $sq \times {n_0 + k -1 \choose k}$ matrix $R$ and a ${m_0 + et - k - 1 \choose et-k} \times sq$ matrix $C$ such that
\begin{equation} \label{eqn: rank of N0} 
\rank~[ R \cdot M(g_0) \cdot C] = \rank~[ M(g_0)] = sq.
\end{equation} 
For any $\hat{g} \in V$, denote the $sq \times sq$ matrix $R \cdot M(\hat{g}) \cdot C$ by $N(\hat{g})$ and $R \cdot M(g_0) \cdot C$ by $N(g_0)$. Let $\hat{g} = y_1G_1^e + \ldots + y_sG_s^e$ be an arbitrary element of $V$, where $y_1, \ldots, y_s \in \F$. Then
\begin{eqnarray*}
M(\hat{g}) &=& y_1\cdot M(G_1^e) + \ldots + y_s\cdot M(G_s^e) \\
\Rightarrow N(\hat{g}) &=& y_1\cdot N(G_1^e) + \ldots + y_s\cdot N(G_s^e).
\end{eqnarray*}
Treating $y_1, \ldots, y_s$ as formal variables, we can infer that $\det(N(\hat{g}))$ is a non-zero polynomial in $y_1, \ldots, y_s$ of degree at most $sq$. This is because, by setting $y_1 = \ldots = y_s = 1$ we get $\hat{g} = g_0$, and we already know that $\det(N(g_0)) \neq 0$ from Equation \eqref{eqn: rank of N0}. Thus, if $g = b_1G_1^{e} + \ldots + b_sG_s^{e}$ such that $b_1, \ldots, b_s \in_r \F^{\times}$, then with probability $1-o(1)$ we have $\det(N(g)) \neq 0$ as $|\F| \gg sq$ (by Proposition \ref{prop:setting parameters to satisfy a bunch of conditions}). That is, $\rank~[ N(g)] = sq = \rank~[ M(g)]$ which implies $\app{P}{\vecz}{k}{g} = \app{P}{\vecz}{k}{g_0}$.
\end{proof}

\subsection*{Proof of Proposition \ref{prop: adjoint is diagonalizable}}
Treat $\vecw^{k(t-1)}$ as an ordered set and let $B \in \GL_{sq}(\F)$ be the basis change matrix from $(h_1, \ldots, h_{sq})$ to $(\vecw^{k(t-1)} \cdot \pi_P(G_1^{e-k}), \ldots, \vecw^{k(t-1)} \cdot \pi_P(G_s^{e-k}))$. Let $K$ be an arbitrary element of $\left \langle \CAL{L}_2 \right \rangle$. As $V = V_1 \oplus \ldots \oplus V_s,~ W = W_1 \oplus \ldots \oplus W_s$ is an indecomposable decomposition of $V$ and $W$ under the action of $\CAL{L}_2$, the matrix $B K A^{-1}$ has the following structure: The columns of $BKA^{-1}$ are indexed by $(G_1^e, \ldots, G_s^e)$ and the rows are indexed by $(\vecw^{k(t-1)} \cdot \pi_P(G_1^{e-k}), \ldots, \vecw^{k(t-1)} \cdot \pi_P(G_s^{e-k}))$. The $G_j^e$-th column of $BKA^{-1}$ has its non-zero entries confined to the $q$ rows indexed by $\vecw^{k(t-1)} \cdot \pi_P(G_j^{e-k})$. \\

By definition of the adjoint, $(D,E) \in \adj(\CAL{L}_2)$ if and only if
\begin{equation} \label{eqn: adjoint equation in the canonical basis}
BKA^{-1} \cdot ADA^{-1} ~=~ BEB^{-1} \cdot BKA^{-1} ~~~~~\text{ for all }~ K \in \left \langle \CAL{L}_2 \right \rangle. 
\end{equation}

Expressed in the basis $(G_1^e, \ldots, G_s^e)$ of $V$, the element $g_0 = G_1^e + \ldots + G_s^e$ is the all-one vector $\mathbf{1} \in \F^s$. Let $\beta \in \vecw^{k(t-1)}$ and $j \in [s]$ be arbitrarily chosen. From the proof of Proposition \ref{prop: decomposition of W}, it follows that there is a $K \in \left \langle \CAL{L}_2 \right \rangle$ such that $BKA^{-1} \cdot \mathbf{1}$ is the unit vector whose $(\beta \cdot \pi_P(G_j^{e-k}))$-th entry is one and all other entries are zero. In other words, all but the $(G_j^e,~ \beta \cdot \pi_P(G_j^{e-k}))$-th entry of $BKA^{-1}$ is zero, and the $(G_j^e,~ \beta \cdot \pi_P(G_j^{e-k}))$-th entry is $1$. As $ADA^{-1}$ and $BEB^{-1}$ satisfy Equation \eqref{eqn: adjoint equation in the canonical basis} for every such $K$ (as we vary $\beta \in \vecw^{k(t-1)}$ and $j \in [s]$), both $ADA^{-1}$ and $BEB^{-1}$ are diagonal matrices, i.e., $A \cdot \adj(\CAL{L}_2)_1 \cdot A^{-1} \subseteq \CAL{D}$. \\

Using Equation \eqref{eqn: adjoint equation in the canonical basis}, it is an easy exercise to show that $\CAL{D} \subseteq A \cdot \adj(\CAL{L}_2)_1 \cdot A^{-1}$. Therefore, $A \cdot \adj(\CAL{L}_2)_1 \cdot A^{-1} = \CAL{D}$.

\subsection*{Proof of Proposition \ref{prop: equality of cond1}} 
Clearly, $U \subseteq U_1 + \ldots + U_s$. We will prove that $U_i \subset U$ for every $i \in [s]$. Observe that 
$$U_i = \app{L}{\vecx}{k}{Q_i^m} \subseteq \left \langle \vecz^{k(t-1)} \cdot \pi_L(Q_i)^{m-k} \right \rangle.$$
We will now show that $\left \langle \vecz^{k(t-1)} \cdot \pi_L(Q_i)^{m-k} \right \rangle \subset U$. Let $\mu$ be an arbitrary $\vecz$-monomial of degree $k(t-1)$. Then
$$\mu = \gamma_{l_1}^{k_1} \cdot \gamma_{l_2}^{k_2} \cdots \gamma_{l_r}^{k_r}$$
for some distinct $l_1, \ldots, l_r \in [b]$, where $k_1 + \ldots + k_r = k$. Let $\alpha_i := (y_{i1l_1} \cdot y_{12l_1} \cdots y_{ik_1 l_1}) \cdot (y_{i1l_2} \cdot y_{12l_2} \cdots y_{ik_2 l_2}) \cdots (y_{i1l_r} \cdot y_{12l_r} \cdots y_{ik_r l_r})$. Using the combinatorial design of the sets $S_1, \ldots S_s$, we get
$$\frac{\partial^k f}{\partial \alpha_i} = \frac{\partial^k Q_i^m}{\partial \alpha_i} = k! \cdot {m \choose k} \cdot \mu \cdot Q_i^{m-k} .$$ 
As $\mu$ is arbitrary and $\char(\F) \nmid k! \cdot {m \choose k}$, we have $\left \langle \vecz^{k(t-1)} \cdot \pi_L(Q_i)^{m-k} \right \rangle \subset U$. From the above equation, it is also easy to notice that $U_i = \left \langle \vecz^{k(t-1)} \cdot \pi_L(Q_i)^{m-k} \right \rangle$.

\subsection*{Proof of Proposition \ref{prop: direct sum of cond1 and cond2}}
We will show that, for every $i \in [s]$, there is a monomial in $P_i\cdot G_i^{m-k}$ that cannot be generated by any other $P_{j} \cdot G_{j}^{m-k}$ for $i \neq j$. The following observation will be useful.

\begin{observation} \label{obs: all monomials present}
Consider a product $P \cdot \ell^{\hat{d}}$, where $P$ is a non-zero polynomial in $\F[\vecz]$ and $\ell = \sum_{z \in \hat{\vecz}}{z}$ for some $\hat{\vecz} \subseteq \vecz$. Let $\char(\F) > \deg_{\vecz}(P \cdot \ell^{\hat{d}})$. Then, for every monomial $\mu \in \hat{\vecz}^{\hat{d}}$, there is a monomial $\beta$ (with non-zero coefficient) in $P \cdot \ell^{\hat{d}}$ such that $\mu$ divides $\beta$. 
\end{observation}
\begin{proof}
Let $\mu$ be a monomial in $\hat{\vecz}^{\hat{d}}$. Write the product $P \cdot \ell^{\hat{d}}$ as $P' \cdot \ell^{d'}$, where $P'$ is coprime to $\ell$ and $d' \geq \hat{d}$. For contradiction, suppose that there is no monomial in $P' \cdot \ell^{d'}$ that is divisible by $\mu$. Then,
\begin{eqnarray*}
\frac{\partial^{\hat{d}}}{d \mu} (P' \cdot \ell^{d'}) &=& 0 \\
\Rightarrow \frac{d'!}{(d'-\hat{d})!} \cdot P' \cdot \ell^{d'-\hat{d}} + g \cdot \ell^{d'-\hat{d} +1} &=& 0, \quad \quad \text{for some $g \in \F[\vecz]$~~ (by chain rule)} \\
\Rightarrow \frac{d'!}{(d'- \hat{d})!} \cdot P' + g \cdot \ell &=& 0.
\end{eqnarray*}
This gives a contradiction as $\char(\F) > \deg_{\vecz}(P \cdot \ell^{\hat{d}})$ and $P'$ is non-zero and not divisible by $\ell$.
\end{proof}
Let $\vecz_i = \{z_{i,1}, \ldots, z_{i,p}\}$. We do the analysis for the two cases $t(m-k) \leq \sqrt{n_0}$ and $t(m-k) \geq \sqrt{n_0}$. \\

Suppose $t(m-k) \leq \sqrt{n_0}$~ so that $|\vecz_i \cap \vecz_{j}| \leq \left \lfloor \frac{t(m-k)}{10} \right \rfloor$ for $i \neq j$. By Observation \ref{obs: all monomials present}, there is a monomial $\beta_i$ in $P_i \cdot \left(z_{i,1} + \ldots + z_{i,p} \right)^{t(m-k)}$ that is divisible by $z_{i,1} \cdot z_{i,2} \cdots z_{i, \lfloor \frac{t(m-k)}{2} \rfloor}$ as $p = \lfloor \frac{\sqrt{n_0}}{2} \rfloor \geq \lfloor \frac{t(m-k)}{2} \rfloor$. If $\beta_i$ is generated by some other term $P_{j} \cdot \left(z_{j,1} + \ldots + z_{j,p} \right)^{t(m-k)}$ then there is a monomial in $P_j$ that is divisible by at least $\lfloor \frac{t(m-k)}{2} \rfloor - \lfloor \frac{t(m-k)}{10} \rfloor$ distinct $\vecz$-variables, as $|\vecz_i \cap \vecz_{j}| \leq \left \lfloor \frac{t(m-k)}{10} \right \rfloor$. But this is not possible as $2k(t-1) < \lfloor \frac{t(m-k)}{2} \rfloor - \lfloor \frac{t(m-k)}{10} \rfloor$ for $m > 7k$. \\

Suppose $t(m-k) \geq \sqrt{n_0}$, in which case $\frac{t(m-k)}{p} \geq 2$. By Observation \ref{obs: all monomials present}, there is a monomial $\beta_i$ in $P_i \cdot \left(z_{i,1} + \ldots + z_{i,p} \right)^{t(m-k)}$ that is divisible by
$$z_{i,1}^{\lfloor \frac{t(m-k)}{p} \rfloor} \cdot z_{i,2}^{\lfloor \frac{t(m-k)}{p} \rfloor} \cdots z_{i,p}^{\lfloor \frac{t(m-k)}{p} \rfloor}.$$
If $\beta_i$ is generated by some other term $P_{j} \cdot \left(z_{j,1} + \ldots + z_{j,p} \right)^{t(m-k)}$ then there is a monomial in $P_j$ that is divisible by at least $p - \frac{p+1}{5} = \frac{4p-1}{5}$ distinct $\vecz$-variables each with multiplicity $\lfloor \frac{t(m-k)}{p} \rfloor$ (as $|\vecz_i \cap \vecz_{j}| \leq \frac{p+1}{5}$). But this is not possible as $2k(t-1) < \frac{4p-1}{5} \cdot \lfloor \frac{t(m-k)}{p} \rfloor$ for $m > 7k$. 

\section{Proofs from Section \ref{sec:lowerbound}} \label{sec:proofs from sec lowerbound}

\subsection*{Proof of Observation \ref{obs:parameters}}
As $k = \lfloor \delta \cdot \frac{d}{t} \rfloor$ and $t \leq \frac{\delta \cdot d}{\ln d}$, we have $k \geq \lfloor \ln d \rfloor$. By definition,
$$c ~=~ \frac{3}{4} \cdot \frac{\ln \frac{n}{k}}{\ln \frac{d}{k}} ~=~ \frac{3}{4} \cdot \left[ \frac{\ln \frac{n}{d}}{\ln \frac{d}{k}} + 1 \right] ~\geq~ \frac{3}{4} \cdot 2 ~=~ \frac{3}{2} ~~~~\text{(as $n \geq d^2$)}.$$
By choice, $n_0 = \lfloor c \cdot k \rfloor$. So,
\begin{eqnarray*}
n_0 \leq ck  &\leq& \frac{3}{4} \cdot \left[ \frac{\ln \frac{n}{d}}{\ln \frac{t}{\delta}} + 1 \right] \cdot \frac{\delta d}{t} ~~~~~~~~~~~\text{(as $k \leq \frac{\delta d}{t}$)} \\
						 &\leq& \frac{3}{4} \cdot \left[ \frac{\ln \frac{n}{d}}{\ln \frac{\ln \frac{n}{d}}{\delta}} + 1 \right] \cdot \frac{\delta d}{\ln \frac{n}{d}} ~~~~~\text{(as $\ln \frac{n}{d} \leq t$)} \\
						 &=& \frac{3}{4} \cdot \left[ \frac{1}{\ln \ln \frac{n}{d} + \ln \frac{1}{\delta}} + \frac{1}{\ln \frac{n}{d}} \right] \cdot \delta d \\
						 &\leq& \frac{3}{4} \cdot \frac{2}{\ln \ln \frac{n}{d}} \cdot \delta d ~\leq~ \frac{d}{\ln \ln d}.
\end{eqnarray*}

\subsection*{Proof of Observation \ref{obs:trivial upper bound}} 
Recall from Equation \eqref{eqn:the measure}, 
$$\APP_{k,n_0}(f) = \max_L~ \dim \app{L}{\vecx}{k}{f},$$
where $L = (\ell_1(\vecz), \ldots, \ell_n(\vecz))$ is an $n$-tuple of linear forms in $\F[\vecz]$ and $|\vecz| = n_0$. An element of $\pi_{L}(\der{\vecx}{k}{f})$ is a homogeneous polynomial of degree $d-k$ in $\vecz$-variables, such a polynomial can have at most ${d-k+n_0-1 \choose n_0-1}$ many $\vecz$-monomials. Hence, $\APP_{k,n_0}(f) \leq {d-k+n_0-1 \choose n_0-1}$. 

\subsection*{Proof of Proposition \ref{prop:circuit upper bound}}
For $i\in[s]$, let $T_i = Q_{i1}Q_{i2} \cdots Q_{im_i}$ be a term of the formula $C$ given by Equation \eqref{eqn:Sigma Pi Sigma Pi t formula}, where degree of every $Q_{ij}$ is in $[t,2t]$. Observe that for any $n$-tuple of linear forms $L = (\ell_1(\vecz), \ldots, \ell_n(\vecz))$, 
$$\app{L}{\vecx}{k}{T_i} \subseteq \left\langle \vecz^{\leq 2tk} \cdot \bigcup_{S \in {[m_i] \choose k}} \left\{ \prod_{j\in [m_i] \backslash S}{Q_{ij}} \right\} \right\rangle.$$
Hence, $\APP_{k,n_0}(T_i) \leq {m \choose k} \cdot {n_0 + 2kt \choose n_0}$ as $m_i \leq m$. By subadditivity of the $\APP$ measure, we have
$$\APP_{k,n_0}(T_1 + \ldots + T_s) \leq s \cdot {m \choose k} \cdot {n_0 + 2kt \choose n_0}.$$ 

\subsection*{Proof of Proposition \ref{prop:lower bound on top fanin high t}}
Suppose $C = f_{n,d,t}$ in Equation \eqref{eqn:Sigma Pi Sigma Pi t formula}. Then, by Proposition \ref{prop:circuit upper bound}, $\APP_{k,n_0}(f_{n,d,t}) \leq s \cdot {m \choose k} \cdot {n_0 + 2kt \choose n_0}$. On the other hand, by Proposition \ref{prop:hard poly lower bound}, $\APP_{k,n_0}(f_{n,d,t}) = {d-k+n_0-1 \choose n_0-1}$. Therefore,
\begin{eqnarray*}
s \geq \frac{{d-k+n_0-1 \choose n_0-1}}{{m \choose k} \cdot {n_0 + 2kt \choose n_0}} &=& \frac{n_0}{d-k+n_0} \cdot \frac{{d-k+n_0 \choose n_0}}{{m \choose k} \cdot {n_0 + 2kt \choose n_0}} \\
	&\geq& \frac{k}{d} \cdot \frac{{d \choose n_0}}{{m \choose k} \cdot {n_0 + 2kt \choose n_0}}, \quad \quad \quad \quad \text{(as $n_0 \geq k$)} \\
	&\geq& \frac{k}{d} \cdot \frac{\left(\frac{d}{n_0}\right)^{n_0}}{\left(\frac{em}{k} \right)^k \cdot \left(\frac{e\cdot(n_0 + 2kt)}{n_0} \right)^{n_0}} \\
	&=& \frac{k}{d} \cdot \frac{1}{\left(\frac{em}{k} \right)^k} \cdot \left[ \frac{d}{e\cdot(n_0 + 2kt)} \right]^{n_0} \\
	&\geq& \frac{k}{d} \cdot \frac{1}{\left(\frac{em}{k} \right)^k} \cdot \left[ \frac{d}{4ekt} \right]^{n_0}, \quad \quad ~~\text{(as $n_0 \leq 2kt$)} \\
	&\geq& \frac{k}{d} \cdot \frac{1}{\left(4.01 \cdot e^{11} \right)^k} \cdot e^{9n_0}, \quad \quad ~\text{(plugging in the values of $k, m$ and $\delta$)} \\
	&\geq& \frac{k}{de^9} \cdot \frac{1}{e^{12.4 \cdot k}} \cdot e^{9 \cdot \frac{3}{4} \cdot \left[ \frac{\ln \frac{n}{d}}{\ln \frac{d}{k}} + 1 \right] \cdot k} ~~~\text{(plugging in the values of $n_0, c$, putting $4.01 < e^{1.4}$)} \\
	&\geq& \frac{k}{de^9} \cdot \frac{1}{e^{5.65 \cdot k}} \cdot e^{6.75 \cdot \left[ \frac{\ln \frac{n}{d}}{\ln \frac{d}{k}}\right] \cdot k} \\
	&\geq& \frac{k}{de^9} \cdot e^{1.1 \cdot \left[ \frac{\ln \frac{n}{d}}{\ln \frac{d}{k}}\right] \cdot k} \quad \quad \quad \quad \quad \text{(as $\frac{n}{d} \geq \frac{d}{k}$)} \\
	&\geq& \frac{1}{e^9} \cdot e^{0.1 \cdot \left[ \frac{\ln \frac{n}{d}}{\ln \frac{d}{k}}\right] \cdot k}  ~~\quad \quad \quad \quad \quad \text{(as $\frac{\ln \frac{n}{d}}{\ln \frac{d}{k}} \cdot k \geq \ln \frac{d}{k}$)} \\
	&=& \left(\frac{n}{d}\right)^{\Omega\left(\frac{d}{t \ln t}\right)}  \quad \quad \quad \quad \quad \quad \quad \text{(as $\frac{d}{k} = \Theta(t)$)}.
\end{eqnarray*}

\subsection*{Proof of Proposition \ref{prop:lower bound on top fanin low t}}
As in the proof of Proposition \ref{prop:lower bound on top fanin high t}, we have
\begin{eqnarray*}
s \geq \frac{{d-k+n_0-1 \choose n_0-1}}{{m \choose k} \cdot {n_0 + 2kt \choose n_0}} &=& \frac{n_0}{d-k+n_0} \cdot \frac{{d-k+n_0 \choose n_0}}{{m \choose k} \cdot {n_0 + 2kt \choose n_0}} \\
	&\geq& \frac{1}{2} \cdot \frac{{n_0 + d-k \choose n_0}}{{2^{O(\frac{d}{t})}} \cdot {n_0 + 2kt \choose n_0}}, \quad \quad \quad \quad \text{(as $n_0 \geq d$, $m = \Theta(d/t)$ and $k = \Theta(d/t)$)} \\
	&=& \frac{1}{2^{O(\frac{d}{t})}} \cdot \left( 1+ \frac{n_0}{d-k}\right) \cdots \left( 1 + \frac{n_0}{d-k-(d-k-2kt-1)}\right) \\
	&\geq& \frac{1}{2^{O(\frac{d}{t})}} \cdot \left(\frac{n_0}{d}\right)^{d-k-2kt} \\
	&\geq& \frac{1}{2^{O(\frac{d}{t})}} \cdot n^{\frac{k}{2d}\cdot(d-k-2kt)} \quad \quad \quad \quad \text{(as $\sqrt{n_0} \geq n^{\frac{k}{2d}} \geq d$)} \\
	&=& n^{\Omega(\frac{d}{t})} \quad \quad \quad \quad \quad \quad \quad \quad \quad \quad \text{(plugging in the value of $k$)}.
\end{eqnarray*}

\subsection*{Proof of Proposition \ref{prop:number of y monomials is greater}}
\textbf{High $t$ case.} In this case $n_0 = \lfloor c \cdot k \rfloor$, where $c = \frac{3}{4} \cdot \frac{\ln \frac{n}{k}}{\ln \frac{d}{k}}$. On one hand,
\begin{eqnarray*}
{d-k+n_0-1 \choose n_0-1} \leq {d + n_0 \choose n_0} &\leq& \left(e \cdot \left( \frac{d}{n_0} + 1 \right) \right)^{n_0} \\
																										 &\leq& \left(e \cdot \left( \frac{d}{ck} + 1 \right) \right)^{ck} \quad \quad \text{(as $n_0 \leq ck$)} \\
																										 &\leq& \left(e \cdot \frac{d}{k} \right)^{ck} \quad \quad\quad\quad\quad ~~\text{(as $c \geq \frac{3}{2}$)} \\
																										 &=& e^{ck} \cdot \left(\frac{n}{k}\right)^{0.75 \cdot k} \quad \quad\quad ~~\text{(plugging in the value of $c$)} \\
																										 &\leq& \left(\frac{n}{k}\right)^{0.76 \cdot k} \quad \quad\quad\quad\quad ~~\text{(as $e^{ck} \ll \left(\frac{n}{k}\right)^{0.01 \cdot k}$).}
\end{eqnarray*}
On the other hand, 
\begin{eqnarray*}
{n_1 \choose k} = {n-n_0(d-k) \choose k} &\geq& \left( \frac{n-n_0(d-k)}{k}\right)^k \\
																				 &\geq& \left( \frac{n}{k}-c(d-k)\right)^k		\quad \quad \text{(as $n_0 \leq ck$)} \\
																				 &\geq& \left(\frac{n}{k}\right)^{0.76 \cdot k}  \quad \quad \quad \quad \quad \text{(as $cd \ll \frac{n}{k}$).}
\end{eqnarray*}

\textbf{Low $t$ case.} In this case $n_0 = \lceil n^{\frac{k}{d}} \rceil$. Verify that $n_0 \geq d$.
\begin{eqnarray*}
{d-k+n_0-1 \choose n_0-1} \leq {d + n_0 \choose n_0} &\leq& \left( e \cdot \frac{n_0 + d}{d} \right)^d \\
																										 &\leq& \left( \frac{2e\cdot n_0}{d} \right)^d \quad \text{(as $n_0 \geq d$)}\\	
																										 &\leq& \left( \frac{n_1}{k} \right)^k \quad\quad\quad \text{(verify after putting the values of $n_0$ and $n_1$)} \\
																										 &\leq& {n_1 \choose k}
\end{eqnarray*}

\subsection*{Proof of Proposition \ref{prop:hard poly lower bound}}
Observe that $\der{\vecy}{k}{f_{n,d,t}} = B$ and $\pi(B) = \vecz^{d-k}$. So, $\APP_{k,n_0}(f_{n,d,t}) \geq {d-k+n_0-1 \choose n_0-1}$. On the other hand, by Observation \ref{obs:trivial upper bound}, $\APP_{k,n_0}(f_{n,d,t}) \leq {d-k+n_0-1 \choose n_0-1}$. Hence, we get the equality.

\end{document}